\documentclass{lmcs}
\pdfoutput=1

\usepackage{lastpage}
\lmcsdoi{17}{4}{22}
\lmcsheading{}{\pageref{LastPage}}{}{}%
{Feb.~16,~2019}{Dec.~20,~2021}{}

\usepackage[utf8]{inputenc}
\usepackage{localnotations}

\usepackage{tikz}
\usetikzlibrary{decorations.markings,calc,fit,calc,positioning,decorations.pathreplacing,matrix}
\tikzset{
	lien/.style={draw,circle,inner sep=0,minimum size=3ex,font=\small},
	reseau/.style={draw, rounded corners},
	par dessus/.style={draw=white, double=black, double distance=.4pt,thick},
	multi/.style={double=gray!50!white, double distance=2pt},
	points/.style={inner sep=0pt, align=center, font={\tiny}},
	nomchemin/.style={blue},
	nomcoupure/.style={red},
	slipknot/.style={gray},
	saut/.style={dotted,->,thick},
	acc/.style={decorate,decoration=brace,line width=0.5pt},
	chemin/.style={red,thick,rounded corners},
	atome/.style={inner sep=1,font={\small}},
}

\begin{document}

\title[An appl. of parallel cut elim. in MLL
to the Taylor exp. of proof nets]
{An application of parallel cut elimination in multiplicative
linear logic to the Taylor expansion of proof nets}

\author[J.~Chouquet]{Jules Chouquet}
\address{Université d’Orléans, LIFO EA 4022, FR-45067 Orléans, France}
\email{jules.chouquet@univ-orleans.fr}
\urladdr{https://orcid.org/0000-0003-2676-0297}

\author[L.~Vaux Auclair]{Lionel Vaux Auclair}
\address{Aix-Marseille Univ, CNRS, I2M, Marseille, France}
\email{lionel.vaux@univ-amu.fr}
\urladdr{https://orcid.org/0000-0001-9466-418X}

\begin{abstract}
	We examine some combinatorial properties of parallel cut elimination in
	multiplicative linear logic (MLL) proof nets.
	We show that, provided we impose a constraint on some paths,
	we can bound the size of all the nets satisfying
	this constraint and reducing to a fixed resultant net.
	This result gives a sufficient condition for an infinite weighted sum of nets
	to reduce into another sum of nets, while keeping coefficients finite.
	We moreover show that our constraints are stable under reduction.

	Our approach is motivated by the quantitative semantics of linear logic:
	many models have been proposed, whose structure reflect the Taylor expansion
	of multiplicative exponential linear logic (MELL) proof nets into 
	infinite sums of differential nets. In order to simulate one cut elimination
	step in MELL, it is necessary to reduce an arbitrary number of cuts 
	in the differential nets of its Taylor expansion.
	It turns out our results apply to differential nets, because their cut
	elimination is essentially multiplicative. We moreover show
	that the set of differential nets that occur in the Taylor expansion of an
	MELL net automatically satisfies our constraints.

	Interestingly, our nets are untyped: we only rely on the sequentiality of
	linear logic nets and the dynamics of cut elimination.
	The paths on which we impose bounds are the switching paths involved in the
	Danos--Regnier criterion for sequentiality.
	In order to accommodate multiplicative units and weakenings, our nets come
	equipped with jumps: each weakening node is connected to some other node.
	Our constraint can then be summed up as a bound on both the length of
	switching paths,
	and the number of weakenings that jump to a common node.
\end{abstract}

\maketitle
\section{Introduction}

\subsection{Context: quantitative semantics and Taylor expansion}

Linear logic takes its roots in the denotational semantics of
$\lambda$-calculus:
it is often presented, by Girard himself \cite{Girard87}, as the result of a
careful investigation of the model of coherence spaces.
Since its early days, linear logic has thus generated a rich ecosystem of
denotational models, among which we distinguish the family of \emph{quantitative
semantics}.
Indeed, the first ideas behind linear logic were exposed even before coherence
spaces, in the model of normal functors \cite{Girard88}, in which Girard
proposed to consider analyticity, instead of mere continuity, as the key property
of the interpretation of $\lambda$-terms:
in this setting, terms denote power series, representing analytic maps between
modules.

This quantitative interpretation reflects precise operational properties of
programs:
the degree of a monomial in a power series is closely related to the number of
times a function uses its argument.
Following this framework, various models were considered --- among which we shall
include the multiset relational model as a degenerate, boolean-valued instance.
These models allowed to represent and characterize quantitative properties such as the
execution time \cite{deCarvalho09}, including best and worst case analysis for
non-deterministic programs \cite{LairdMMP13}, or the probability of reaching a
value \cite{DE11}.
It is notable that this whole approach gained momentum in the early 2000’s,
after the introduction by Ehrhard of models \cite{Ehrhard02,Ehrhard05} in which
the notion of analytic maps interpreting $\lambda$-terms took its usual sense,
while Girard's original model involved set-valued formal power series.
Indeed, the keystone in the success of this line of work is an analogue of the
Taylor expansion formula, that can be established both for $\lambda$-terms
and for linear logic proofs.

Mimicking this denotational structure, Ehrhard and Regnier introduced the
differential $\lambda$-calculus \cite{ER03} and differential linear logic
\cite{ER05}, which allow to formulate a syntactic version of Taylor expansion:
to a $\lambda$-term (resp. to a linear logic proof), we associate an infinite
linear combination of approximants \cite{ER08,Ehrhard16}.
In particular, the dynamics (\emph{i.e.} $\beta$-reduction or cut elimination)
of those systems is dictated by the identities of quantitative semantics.
In turn, Taylor expansion has become a useful device to design and study new
models of linear logic, in which morphisms admit a matrix representation:
the Taylor expansion formula allows to describe the interpretation of promotion
--- the operation by which a linear resource becomes freely duplicable ---
in an explicit, systematic manner.
It is in fact possible to show that any model of differential linear logic
without promotion gives rise to a model of full linear logic in this way
\cite{deCarvalho07}:
in some sense, one can simulate cut elimination through Taylor expansion.

\subsection{Motivation: reduction in Taylor expansion}

There is a difficulty, however: Taylor expansion generates infinite sums and,
\emph{a priori}, there is no guarantee that the coefficients in these sums will
remain finite under reduction.
In previous works \cite{deCarvalho07,LairdMMP13}, coefficients 
were thus required to be taken in a complete semiring: all sums should converge.
In order to illustrate this requirement, let us first consider the case of
$\lambda$-calculus.

The linear fragment of differential $\lambda$-calculus,
called \textit{resource $\lambda$-calculus}, is the target of the
syntactic Taylor expansion of $\lambda$-terms.
In this calculus, the application of a term to another is replaced with
a multilinear variant: $\rappl{s}{\mset{t_1,\dotsc,t_n}}$ denotes the $n$-linear 
symmetric application of resource term $s$ to the multiset of resource terms 
$[t_1,\dotsc,t_n]$.
Then, if $x_1,\dotsc,x_k$ denote the occurrences of $x$ in $s$, the redex
$\rappl{\lambda x.s}{\mset{t_1,\dotsc,t_n}}$ reduces to the sum
$\sum_{f:\{1,\dotsc,k\}\stackrel\sim\to \{1,\dotsc, n\}}
s[t_{f(1)}/x_1,\dotsc,t_{f(k)}/x_k]$: here $f$ ranges over all bijections
$\{1,\dotsc,k\}\stackrel\sim\to \{1,\dotsc, n\}$ so this sum is zero if
$n\not=k$.
As sums are generated by reduction, it should be noted that all the syntactic
constructs are linear, both in the sense that they commute to sums, and in the
sense that, in the elimination of a redex, no subterm of the argument multiset
is copied nor erased.
The key case of Taylor expansion is that of application:
\begin{equation}\label{eq:taylor}
	\taylor(MN)=\sum_{n\in N} \frac{1}{n!} \rappl{\taylor(M)}{\taylor(N)^n}
\end{equation}
where $\taylor(N)^n$ is the multiset made of $n$ copies of $\taylor(N)$
--- by $n$-linearity, $\taylor(N)^n$ is itself an infinite linear 
combination of multisets of resource terms appearing in $\taylor(N)$.
Admitting that $\rappl{M}{\mset{N_1,\dotsc,N_n}}$ represents 
the $n$-th derivative of $M$, computed at $0$, and $n$-linearly
applied to $N_1$, \ldots, $N_n$, 
one immediately recognizes the usual Taylor expansion formula.

From (\ref{eq:taylor}), it is immediately clear that, 
to simulate one reduction step occurring in $N$, 
it is necessary to reduce in parallel in an unbounded number 
of subterms of each component of the expansion.
Unrestricted parallel reduction, however, is ill defined in this setting.
Consider the sum
$\sum_{n\in\N} \rappl{\lambda xx}{\mset{\cdots\rappl{\lambda xx}{\mset y}\cdots}}$
where each summand consists of $n$ successive linear applications of the identity
to the variable $y$: then by simultaneous reduction of all redexes in each component,
each summand yields $y$, so the result should be $\sum_{n\in\N} y$ which is not
defined unless the semiring of coefficients is complete in some sense.

Those considerations apply to linear logic as well as to $\lambda$-calculus.
We will use proof nets \cite{Girard87} as the syntax for proofs of
multiplicative exponential linear logic (MELL).
The target of Taylor expansion is then in promotion-free differential nets \cite{ER05},
which we call \emph{resource nets} in the following, by analogy with the resource
$\lambda$-calculus:
these form the multilinear fragment of differential linear logic.

\begin{figure}[t]
\begin{center}
	\begin{tikzpicture}[scale=.6,baseline=(cdots.base)]
	% contour de boîte
	\node[lien] (oc) at (4,0) {$\oc$};
	\draw[reseau] (0.5,0.5)--(0.5,1.5)--(4.5,1.5)--(4.5,0.5)--cycle;
	\node at (2,0.25) {$\cdots$};
	\draw (oc)--(0.25,0)--(0.25,1.75)--(4.75,1.75)--(4.75,0)--(oc);
	% fils
	\coordinate[above=.5ex] (in) at (oc.center);
	\node[lien] (d1) at (1,-0.75){$?$};
	\node[lien] (d2) at (3,-0.75){$?$};
	\draw[multi](d1)--(d1|-oc) (d1|-in)--(1,0.5);
	\draw[multi](d2)--(d2|-oc) (d2|-in)--(3,0.5);
	\draw(oc)--(4,0.5);
	\draw(d1)--(1,-1.5);\draw(d2)--(3,-1.5);\draw(oc)--(4,-1.5);
	\node (cdots) at (2,-0.75) {$\cdots$};
	% intérieur
	\node at (2.5,1.0){$P$};
\end{tikzpicture}
\qquad
expands to
\qquad
$\sum_{n\in\N}\frac{1}{n!}
\begin{tikzpicture}[scale=.6,baseline=(cdots.base)]
	\draw[reseau](0.5,0)--(0.5,1)--(4.5,1)--(4.5,0)--cycle;%boîte gauche
	\draw(2.2,-0.3)node{$\cdots$};
	\draw(2.5,0.5)node{$\taylor(P)$};
	
	\draw(5,0.5)node{$\stackrel{n}{\cdots}$}; %boîtes intermédiaires
	
	\draw[reseau](5.5,0)--(5.5,1)--(9.5,1)--(9.5,0)--cycle;%boîte droite
	\draw(7.5,0.5)node{$\taylor(P)$};
	\draw(6.8,-0.3)node{$\cdots$};
	
	\node[lien] (c1) at (3.5,-1.75){$?$};
	\node[lien] (c2) at (5.5,-1.75){$?$};
	\node[lien] (cc) at (6.5,-1.75){$!$};%(co)contractions
	\node[above=2pt of c1, points] {$\ldots$};
	\node[above=2pt of c2, points] {$\ldots$};
	\node[above=2pt of cc, points] {$\ldots$};
	\node (cdots) at (4.5,-1.75) {$\cdots$};
	
	\draw(c1)--(3.5,-2.5);
	\draw(c2)--(5.5,-2.5);
	\draw(cc)--(6.5,-2.5);%conclusions
	
	\draw[multi](1,0) to [out=270,in=125] (c1);
	\draw[multi](3,0) to [out=270,in=125] (c2);
	\draw(4,0) to [out=270,in=125] (cc);%branchement gauche
	
	\draw[multi](6,0) to [out=270,in=55] (c1);
	\draw[multi](8,0) to [out=270,in=55] (c2);
	\draw(9,0) to [out=270,in=55] (cc);%branchement droit
	
\end{tikzpicture}$
\end{center}
\caption{Taylor expansion of a promotion box 
(thick wires denote an arbitrary number of wires)}
\label{fig:taylor}
\end{figure}

In linear logic, Taylor expansion consists in replacing duplicable subnets,
embodied by promotion boxes, with explicit copies, as in Figure~\ref{fig:taylor}:
if we take $n$
copies of the box, the main port of the box is replaced with an $n$-ary $\oc$-link, 
while the $\wn$-links at the border of the box collect all copies of the
corresponding auxiliary ports.
Again, to follow a single cut elimination step in $P$, 
it is necessary to reduce an arbitrary number of copies.
And unrestricted parallel cut elimination in an infinite sum of resource nets
is broken, as one can easily construct an infinite family of 
nets, all reducing to the same resource net $p$ in a single step of parallel
cut elimination: see Figure~\ref{fig:collapse}.

\begin{figure}[t]
\begin{center}
\begin{tikzpicture}[
	baseline=(cdots.base),
	rounded corners,
]
	\node (ax1) at (0.5,0.5) [lien] {ax};
	\node (cut1) at (1.5,0) [lien] {cut};
	\node (ax2) at (3.5,0.5) [lien] {ax};
	\node (cut2) at (4.5,0) [lien] {cut};
	\draw (0,0) to (0,0.5)--(ax1)
		--(1,0.5)--(1,0)--(cut1)
		--(2,0)--(2,0.2);
	\node (cdots) at (2.5,0.25){$\stackrel{n}{\cdots}$};
	\draw(3,0.3)--(3,0.5)--(ax2)
		--(4,0.5)--(4,0)--(cut2)
		--(5,0)--(5,0.5);
	\node at (5,0.75) {$p$};
	\draw[reseau](4.5,0.5)--(5.5,0.5)--(5.5,1)--(4.5,1)--cycle;
\end{tikzpicture}
\end{center}
\caption{Example of a family of nets, all reducing to a single net $p$,
by the parallel elimination of axiom cuts.}
\label{fig:collapse}
\end{figure}

\subsection{Our approach: taming the combinatorial explosion of antireduction}

\label{subsection:approach}

The problem of convergence of series of linear approximants under reduction 
was first tackled by Ehrhard and Regnier, for 
the normalization of Taylor expansion of ordinary $\lambda$-terms \cite{ER08}.
Their argument relies on a uniformity property, specific to the pure
$\lambda$-calculus: the support of the Taylor expansion of a $\lambda$-term
forms a clique in some fixed coherence space of resource terms.
This method cannot be adapted to proof nets: there is no coherence relation on
differential nets such that all supports of Taylor expansions are cliques
\cite[Section~V.4.1]{Tasson09}.

An alternative method to ensure convergence without any uniformity hypothesis was first
developed by Ehrhard for typed terms in a $\lambda$-calculus extended with
linear combinations of terms \cite{Ehrhard10}: there, the presence of sums also
forbade the existence of a suitable coherence relation.
This method can be generalized to strongly normalizable \cite{PTV16}, or even
weakly normalizable \cite{Vaux17} terms.
One striking feature of this approach is that it concentrates on
the support (\emph{i.e.} the set of terms having non-zero coefficients) 
of the Taylor expansion.
In each case, one shows that, given a normal resource term $t$ and a
$\lambda$-term $M$, there are finitely many terms $s$, such that:
\begin{itemize}
	\item the coefficient of $s$ in $\taylor(M)$ is non zero; and
	\item the coefficient of $t$ in the normal form of $s$ is non zero.
\end{itemize}
This allows to normalize the Taylor expansion:
simply normalize in each component, then compute the sum,
which is component-wise finite.

The second author then remarked that
the same could be done for $\beta$-reduction \cite{Vaux17},
even without any uniformity, typing or normalizability requirement.
Indeed, writing $s\toto t$ if $s$ and $t$ are resource terms such that
$t$ appears in the support of a parallel reduct of $s$, 
the size of $s$ is bounded by a function of the size of $t$
\emph{and} the height of $s$.
So, given that if $s$ appears in $\taylor(M)$
then its height is bounded by that of $M$,
it follows that, for a fixed resource term $t$ there are finitely 
many terms $s$ in the support of $\taylor(M)$ such that $s\toto t$:
in short, parallel reduction is always well-defined on the Taylor expansion
of a $\lambda$-term.

Our purpose in the present paper is to develop a similar technique for MELL
proof nets: we show that one can bound the size of a resource net 
$p$ by a function of the size of any of its parallel reducts, and of an
additional quantity on $p$, yet to be defined.
The main challenge is indeed to circumvent the lack of inductive structure in
proof nets: in such a graphical syntax, there is no structural notion of
height.

We claim that a side condition on switching paths, \emph{i.e.} paths in the
sense of Danos--Regnier's correctness criterion \cite{DR89}, is an appropriate
replacement.
Backing this claim, there are first some intuitions:
\begin{itemize}
	\item the main culprits for the unbounded loss of size in reduction
		are the chains of consecutive cuts, as in Figure~\ref{fig:collapse};
	\item we want the validity of our side condition to be stable
		under reduction so, rather than chains of cuts,
		we should consider the length of switching paths;
	\item indeed, if $p$ reduces to $q$ \emph{via}
		cut elimination, then the switching paths of $q$
		are somehow related with those of $p$;
	\item and the switching paths of a resource net in $\taylor(P)$
		are somehow related with those of $P$.
\end{itemize}
In the following we will establish precise formulations of 
those last two points: we study the structure of switching
paths through cut elimination in Section~\ref{section:vc};
and we describe the switching paths of the elements of $\taylor(P)$
in Section~\ref{section:Taylor}.

\begin{figure}[t]
	\begin{center}
		\begin{tikzpicture}[
				baseline=(cdots.base),
				rounded corners,
			]
			\node(1-1) at (0,0) [lien] {$+$};
			\node(cut1) at (0.5,-0.5) [lien] {cut};
			\node(bot-1) at (1,0) [lien] {$-$};
			\draw(1-1)--(0,-0.5)--(cut1)--(1,-0.5)--(bot-1);
			
			\node(1-n) at (3,0) [lien] {$+$};
			\node(cutn) at (3.5,-0.5) [lien] {cut};
			\node(bot-n) at (4,0) [lien] {$-$};
			\draw(1-n)--(3,-0.5)--(cutn)--(4,-0.5)--(bot-n);

			\node (cdots) at (2,-0.25){${\cdots}$};

			\draw[reseau](4.5,-0.5)--(5.5,-0.5)--(5.5,0)--(4.5,0)--cycle;
			\node at (5,-0.25){$p$};

		 \end{tikzpicture}
	\end{center}
	\caption{Evanescent cuts: 
		here each $(+)$ node can denote a tensor unit $\one$ or a
		coweakening (a nullary $\oc$-link), and then the corresponding $(-)$ node
		should be the dual unit $\bot$ or a weakening (a nullary $\wn$-link).
		Then the depicted net reduces to $p$ in one parallel cut elimination step.}
	\label{fig:jumps_collapse}
\end{figure}

In presence of multiplicative units, or of weakenings (nullary $\wn$-links)
and coweakenings (nullary $\oc$-links), we must also take special 
care of another kind of cuts, that we call \definitive{evanescent cuts}:
when a cut between such nullary links is eliminated, it simply vanishes,
leaving the rest of the net untouched, as in Figure~\ref{fig:jumps_collapse},
which is obviously an obstacle for our purpose.\footnote{
	The treatment of weakenings is indeed the main novelty of the present 
	extended version over our conference paper \cite{CV18}.
}

In order to deal with nullary links, a well known trick is to attach each
weakening (or $\bot$-link) to another node in the net: switching paths can then
follow such jumps, which is useful to characterize exactly those nets
that come from proof trees \cite[Appendix A.2]{Girard96}.
Here we will rely on this structure to control the effect of eliminating
evanescent cuts on the size of a net.

In all our exposition, we adopt a particular presentation of nets:
we consider $n$-ary exponential links rather than separate
(co)dereliction and (co)contraction, as this allows to reduce the dynamics
of resource nets to that of multiplicative linear logic (MLL) proof nets.\footnote{
			In other words, we adhere to a version of linear logic proof nets and
			resource nets which is sometimes called \emph{nouvelle syntaxe}, although
			it dates back to Regnier's PhD thesis \cite{Regnier92}.
			For the linear logic \emph{connoisseur}, this is already 
			apparent in Figure~\ref{fig:taylor}.
			See also the discussion in our conclusion 
			(Section~\ref{section:conclusion}).
}

\subsection{Outline}

In Section~\ref{section:nets}, we first introduce MLL proof nets formally,
in the term-based syntax of Ehrhard \cite{Ehrhard14}.
We define the parallel cut elimination relation $\toto$ in this setting,
that we decompose into multiplicative reduction $\totom$,
axiom-cut reduction $\totoa$
and evanescent reduction $\totoe$.
We also present the notion of switching path for this syntax,
and introduce the two quantities that will be our main objects of study in the following:
\begin{itemize}
	\item the maximum number $\jd(p)$ of $\bot$-links that jump to a common target;
	\item the maximum length $\ln(p)$ of any switching path in the net $p$.
\end{itemize}
Let us mention that typing plays absolutely no role in our approach, so we do
not even consider formulas of linear logic in our exposition:
we will rely on the geometrical structure of nets only.

We show in Section~\ref{section:size} that, if $p\totom q$, $p\totoa q$ 
or $p\totoe q$ then the size of $p$
is bounded by a function of $\ln(p)$, $\jd(p)$, and the size of $q$.
In order to be able to iterate this combinatorial argument, we must show that,
given bounds for $\ln(p)$ and $\jd(p)$, we can infer bounds on $\ln(q)$ and $\jd(q)$:
this is the subject of Sections~\ref{section:vc} and \ref{section:qi}.

Section~\ref{section:vc} is dedicated to the proof that we can bound
$\ln(q)$ by a function of $\ln(p)$: the main case is the
multiplicative reduction, as this may create new switching paths in $q$ that we
must relate with those in $p$.
In this task, we concentrate on the notion of \emph{slipknot}: a pair of
residuals of a cut of $p$ occurring in a path of $q$.
Slipknots are essential in understanding how switching paths 
are structured after cut elimination: this analysis is motivated by a technical
requirement of our approach, but it can also be considered as a contribution 
to the theory of MLL nets \emph{per se}.

In Section~\ref{section:qi}, we show that $\jd(q)$ is bounded by a function 
of $\ln(p)$ and $\jd(p)$: the critical case here is that of chains of jumps 
between evanescent cuts.

We leverage all of the above results in Section~\ref{section:general}, to
generalize them to a reduction $p\toto q$, or even an arbitrary sequence 
of reductions.
In particular, if $p\toto q$ then the size of $p$ is bounded by a function of
the size of $q$ and of $\ln(p)$ and $\jd(p)$.
Again, this result is motivated by 
the study of quantitative semantics, but it is essentially a theorem about MLL.

We establish the applicability of our approach to the Taylor expansion
of MELL proof nets in Section~\ref{section:Taylor}:
we show that if $p$ is a resource net of $\taylor(P)$, then $\ln(p)$
is bounded by a function of the size of $P$,
and $\jd(p)$ is bounded by the size of $P$.

Finally, we discuss the scope of our results in the concluding Section~\ref{section:conclusion}.

\section{Definitions}
\label{section:nets}

We provide here the minimal definitions necessary for us to work with MLL
proof nets.
As stated before, let us stress the fact that 
the choice of MLL is not decisive for the development of
Sections~\ref{section:nets} to \ref{section:general}.
The reader can check that we rely on three ingredients only:
\begin{itemize}
	\item the definition of switching paths;
	\item the fact that multiplicative reduction amounts to plug 
		bijectively the premises of a $\otimes$-link 
		with those of $\parr$-link (in the nullary case,
		evanescent cuts simply vanish);
	\item the definition of jumps and how they are affected 
		by cut elimination.
\end{itemize}
The results of those sections are thus directly applicable to 
resource nets, thanks to our choice of generalized exponential links:
this will be done in Section~\ref{section:Taylor}.

\subsection{Nets}

A proof net is usually presented as a graphical object 
such as that of Figure~\ref{fig:net-shape}.
Following Ehrhard \cite{Ehrhard14, Ehrhard16}, 
we will rely on a term syntax for denoting such nets.
This is based on a quite standard trichotomy: a proof net can be divided into a top
layer of axioms, followed by trees of connectives, down to cuts between the
conclusions of some trees.

\begin{figure}[t]
	\begin{center}
		\begin{tikzpicture}[
			baseline=(cdots.base),
			rounded corners,
			scale = 1,
		]
			\node(bot0) at (0,0) [lien] {$\bot$};
			\node(par)  at (1,0) [lien] {$\parr$};
			\node(ax-x) at (1,1) [lien] {ax};
			\node(cutm) at (2,-.75) [lien] {cut};
			\node(onet) at (2.5,1) [lien] {$\one$};
			\node(ten)  at (3,0) [lien] {$\otimes$};
			\node(ax-y) at (4,1) [lien] {ax};
			\node(bot1) at (4,0) [lien] {$\bot$};
			\node(cuta) at (5.25,0) [lien] {cut};
			\node(bota) at (6,1) [lien] {$\bot$};
			\node(onee) at (7,0) [lien] {$\one$};
			\node(cute) at (8,-.75) [lien] {cut};
			\node(bote) at (9,0) [lien] {$\bot$};
			\node(one)  at (10,0) [lien] {$\one$};

			\draw (bot0)--(bot0|- 0,-1);
			\draw (ax-x) to [out=west,in=north west] (par);
			\draw (ax-x) to [out=east,in=north east] (par);
			\draw (bot1)--(bot1|- 0,-1);
			\draw (par) to [out=south, in=west] (cutm);
			\draw (ten) to [out=south, in=east] (cutm);
			\draw (onet) to [out=south, in=north west] (ten);
			\draw (ax-y) to [out=west, in=north east] (ten);
			\draw (ax-y) to [out=east, in=west] (cuta);
			\draw (bota) to [out=south, in=east] (cuta);
			\draw (onee) to [out=south, in=west] (cute);
			\draw (bote) to [out=south, in=east] (cute);
			\draw (one)--(one|- 0,-1);
	\end{tikzpicture}
	\end{center}
	\caption{An example of multiplicative net}
	\label{fig:net-shape}
\end{figure}

We will represent the conclusions of axiom rules by variables: the duality
between two conclusions of an axiom rule is given by an involution $x\mapsto
\bar{x}$ over the set $\fv$ of variables.
Our nets will be finite families of trees and cuts, where trees are inductively
generated from variables by the application of MLL connectives, of arbitrary
arity: $\otimes(t_1,\dotsc,t_n)$ and $\parr(t_1,\dotsc,t_n)$.
A tree thus represents a conclusion of a net, together with the nodes above it,
up to axiom conclusions. A cut is then given by the pair of trees $\cut{t_1}{t_2}$,
whose conclusions it cuts together.
In order to distinguish between various occurrences of nullary connectives
$\one=\otimes()$ and $\bot=\parr()$, we will index them with labels taken from 
sets $\units$ and $\counits$.

Formally, the set of \definitive{raw trees} (denoted by $s$, $t$, \emph{etc.})
is generated as follows:
\[
	t \recdef x\mid \unit{\cw}\mid \counit{\w}\mid \otimes(t_1,\dots,t_n)\mid \parr(t_1,\dots,t_n)
\]
where $x$ ranges over $\fv$, $\cw$ ranges over $\units$,
$\w$ ranges over $\counits$ and we require $n\not=0$
in the two last cases. We assume $\fv$,
$\units$ and $\counits$ are pairwise disjoint and all three are
denumerably infinite.
We will always identify a nullary connective tree $\unit\cw$ or $\counit\w$
with its label $\cw$ or $\w$, so that $\atoms=\fv\cup\units\cup\counits$ is just the set
of \definitive{atomic trees}.
We will generally use letters $x$, $y$, $z$ for variables, $\w$
for the elements of $\counits$, $\cw$ for the elements of $\units$,
and $s$, $t$, $u$, $v$ for arbitrary raw trees.

We write $\ST(t)$ for the set of \definitive{subtrees} of a given raw tree $t$,
which is defined inductively in the natural way :
if $t\in\atoms$, then $\ST(t)=\{t\}$;
if $t=\diamond(t_1,\dots,t_n)$ with $\diamond\in\{\otimes,\parr\}$, then
$\ST(t)=\{t\}\cup\bigcup_{i\in\{1,\dots,n\}}\ST(t_i)$.
We moreover write $\fv(t)$ for $\ST(t)\cap\fv$,
and similarly for $\units(t)$, $\counits(t)$ and $\atoms(t)$.
A \definitive{tree} is then a raw tree $t$ such that if
$\diamond(t_1,\dotsc,t_n)\in\ST(t)$ then the sets $\atoms(t_i)$ for $1\le i\le n$
are pairwise disjoint: in other words, each atom occurs at most once in
$t$.
As a consequence, each subtree $u\in\ST(t)$ occurs exactly once in a tree $t$.

A $\definitive{cut}$ is an unordered pair $c=\cut{t}{s}$ of trees
such that $\atoms(t)\cap\atoms(s)=\emptyset$,
and then we set $\ST(c)=\ST(t)\cup\ST(s)$, and similarly for 
$\fv(c)$, $\units(c)$, $\counits(c)$ and $\atoms(c)$.
Note that, in the absence of typing, we do not put any compatibility
requirement on cut trees.

Given a set $A$, we denote by $\vec{a}$ any finite family $(a_i)_{i\in I}\in
A^I$ of elements of $A$.
In general, we abusively identify $\vec{a}$ with any enumeration
$(a_1,\dotsc,a_n)\in A^n$ of its elements, and we may even write
$\vec{a}=a_1,\dotsc,a_n$ in this case; moreover, we simply write $\vec{a},\vec{b}$ for the 
concatenation of families $\vec a$ and $\vec b$
(whose index set is implicitly the sum of the index sets of $\vec a$ and $\vec b$).
We may also write, e.g.,
$a_i\in\vec a$, identifying the family $\vec a$ with its support set.
Since we only consider families of pairwise distinct elements, such abuse of
notation is generally harmless: in this case,
the only difference between $(a_i)_{i\in I}$ and its support set $\{ a_i\tq i\in I\}$
is whether the bijection $i\mapsto a_i$ is part of the data or not.
If $f$ is a function from $A$ to any powerset,
we extend it to families in the obvious way, setting $f(\vec a)=\bigcup_{a\in\vec a} f(a)$.
E.g., if $\vec{\gamma}$ is a family of trees or cuts we write
$\fv(\vec{\gamma})=\bigcup_{\gamma\in\vec{\gamma}}\fv(\gamma)$.

An MLL \definitive{bare proof net} is a pair $p=(\vec{c};\vec{t})$
of a finite family $\vec{c}$ of pairwise distinct cuts and a finite family
$\vec{t}$ of pairwise distinct trees such that:
for all distinct cuts or trees $\gamma,\gamma'\in  \vec c\cup \vec t$,
$\atoms(\gamma)\cap \atoms(\gamma')=\emptyset$;
and $\fv(p)=\fv(\vec c)\cup\fv(\vec t)$ is closed under the involution
$x\mapsto \bar x$.
We write $\cuts(p)=\vec c$ for the family of cuts of $p$.
For any tree, cut or bare proof net $\gamma$,
we define the \definitive{size} of $\gamma$ as $\size(\gamma)=\card\ST(\gamma)$:
graphically, $\size(p)$ is nothing but the number of wires in $p$. 

\begin{rem}\label{rem:alpha}
	In a graphical structure such as that of Figure~\ref{fig:net-shape},
	the interface (\emph{i.e.} the set of extremities of dangling wires, which
	represent the conclusions of the net) is relevant:
	in particular, cut elimination preserves this interface.
	So, in $p=(\vec c;\vec t)$,
	$\vec t$ is intrinsically a family, whose index set is precisely the
	interface of the structure.

	On the other hand, the rest of the net should be considered up to
	isomorphism: in our case, this amounts to the reindexing of cuts,
	and the renaming of atoms, preserving the duality involution on variables.
	We may call $\alpha$-equivalence the corresponding equivalence 
	relation on bare proof nets, as it has the very same status
	as the renaming of bound variables in the ordinary 
	$\lambda$-calculus.
	In particular, $\vec c$ should be considered as a set,
	although we introduce it as a family here,
	just because it will be convenient to treat the concatenation $\vec c,\vec t$
	as a family of cuts and trees in the following.

	The reader may check that bare proof nets quotiented by 
	$\alpha$-equivalence, as introduced above,
	are exactly the usual (untyped) proof structures 
	for MLL (with connectives of arbitrary arity).
	We keep this quotient implicit whenever possible in the remaining:
	in any case, $\alpha$-equivalence preserves the size of nets, as well as the
	length of paths to be introduced later.\footnote{
		Note that this situation differs slightly from the
		case of interaction nets~\cite{lafont:interactionnets},
		where explicit axioms and cuts links are missing
		and there is no top-down orientation \emph{a priori}.
		Term syntaxes have been proposed for
		those~\cite[among others]{ms:chaminets,fm:calcinets}
		but the correspondence is less immediate:
		it must be restricted to deadlock-free 
		interaction nets and, in addition to $\alpha$-equivalence,
		one must introduce some mechanism to deal with implicit axiom-cut
		elimination in the application of reduction rules.
	}

\end{rem}

As announced in our introduction, our nets will be equipped with
jumps from $\counit{}$ nodes to other nodes.
An MLL \definitive{proof net} will thus be the data of a bare proof net $p$ and
of a \textit{jump function} $\jump : \counits(p)\to \ST(p)$.
We will often identify a proof net with its underlying bare net $p$, 
and then write $\jump_p$ for the associated jump function.
Figure~\ref{fig:example-net} presents such a net, whose underlying 
graphical structure is that of Figure~\ref{fig:net-shape}.

\begin{figure}[t]
	\begin{center}
		\begin{tikzpicture}[
			baseline=(cdots.base),
			rounded corners,
			scale = 1,
		]
			\node(bot0) at (0,0) [lien] {$\bot$};
			\node at (bot0.south) [below right, atome] {$\w_1$};
			\node(par)  at (1,0) [lien] {$\parr$};
			\node(ax-x) at (1,1) [lien] {ax};
			\node(x) at (ax-x.west) [above left,atome] {$x$};
			\node    at (ax-x.east) [above right,atome] {$\bar x$};
			\node(cutm) at (2,-.75) [lien] {cut};
			\node(onet) at (2.5,1) [lien] {$\one$};
			\node at (onet.south) [below left, atome] {$\cw_1$};
			\node(ten)  at (3,0) [lien] {$\otimes$};
			\node(ax-y) at (4,1) [lien] {ax};
			\node(y) at (ax-y.west) [above left,atome] {$y$};
			\node(by)at (ax-y.east) [above right,atome] {$\bar y$};
			\node(bot1) at (4,0) [lien] {$\bot$};
			\node at (bot1.south) [below right, atome] {$\w_2$};
			\node(cuta) at (5.25,0) [lien] {cut};
			\node(bota) at (6,1) [lien] {$\bot$};
			\node at (bota.south) [below right, atome] {$\w_3$};
			\node(onee) at (7,0) [lien] {$\one$};
			\node at (onee.south) [below left, atome] {$\cw_2$};
			\node(cute) at (8,-.75) [lien] {cut};
			\node(bote) at (9,0) [lien] {$\bot$};
			\node at (bote.south) [below right, atome] {$\w_4$};
			\node(one)  at (10,0) [lien] {$\one$};
			\node at (one.south) [below right, atome] {$\cw_3$};

			\draw (bot0)--(bot0|- 0,-1);
			\draw (ax-x) to [out=west,in=north west] (par);
			\draw (ax-x) to [out=east,in=north east] (par);
			\draw (bot1)--(bot1|- 0,-1);
			\draw (par) to [out=south, in=west] (cutm);
			\draw (ten) to [out=south, in=east] (cutm);
			\draw (onet) to [out=south, in=north west] (ten);
			\draw (ax-y) to [out=west, in=north east] (ten);
			\draw (ax-y) to [out=east, in=west] (cuta);
			\draw (bota) to [out=south, in=east] (cuta);
			\draw (onee) to [out=south, in=west] (cute);
			\draw (bote) to [out=south, in=east] (cute);
			\draw (one)--(one|- 0,-1);

			\draw[saut] (bot0) to [out=north east, in=west] (par);
			\draw[saut] (bot1) to [out=north, in=south] ($(ax-y.east)+(.1,-.1)$);
			\draw[saut] (bota) to [out=east, in=north] (onee);
			\draw[saut] (bote) to [out=north east, in=north west] (one);
	\end{tikzpicture}
	\end{center}
	\caption{The net $p_0=(\cut{\parr(x,\bar x)}{\otimes(\unit{\lambda_1},y)},\cut{\bar y}{\counit{\mu_3}},\cut{\unit{\lambda_2}}{\counit{\mu_4}};\counit{\mu_1},\counit{\mu_2},\unit{\lambda_3})$
	with $\jump_{p_0}:\mu_1\mapsto \parr(x,\bar x),\mu_2\mapsto \bar y,\mu_3\mapsto \lambda_2,\mu_4\mapsto\lambda_3$.}
	\label{fig:example-net}
\end{figure}

We can already introduce the first of our two key quantities:
the \definitive{jump degree} $\jd(p)$ of a net $p$.
We first define the jump degree of any tree $t\in\ST(p)$,
setting $\jd_p(t)=\card\{\w\in\counits(p)\tq \jump_p(\w)=t\}$.
We will often write $\jd(t)$ instead of $\jd_p(t)$ if $p$ is clear from
the context.
Then we set $\jd(p)=\max\{\jd(t)\tq t\in\ST(p)\}$.

\begin{rem}\label{rem:jumps}
	Originally, jumps were introduced as \emph{pis aller} for the
	characterization of sequentializable proof nets \cite[Appendix A.2]{Girard96}.
	Indeed, in presence of multiplicative units and without jumps,
	Danos--Regnier's correctness criterion, requiring the connectedness and
	acyclicity of switching graphs, fails to accept some proof nets
	corresponding to proof trees using the $\bot$-rule.
	So, to characterize all sequentializable nets,
	one has to require the existence of a jump function that makes 
	all the switching graphs connected and acyclic.
	This additional structure is somewhat arbitrary,
	and it restores a form of bureaucratic sequentiality:
	distinct jumping functions on the same bare net
	may yield equivalent sequentializations.
	And there is no satisfactory solution to that issue:
	if a notion of proof net could capture proof equivalence in MLL with units,
	then deciding the identity of proof nets in that setting would not be tractable,
	simply because proof equivalence is PSPACE-complete
	\cite{hh:mll}.

	A simple, consensual alternative is to forget about jumps and
	drop the connectedness requirement: the acyclicity criterion 
	characterizes exactly those nets that are sequentializable
	using an additional $mix$-rule, corresponding 
	to the parallel juxtaposition of nets.
	In particular, this weaker requirement is sufficient to
	avoid the problematic cases of cut elimination.
	This is the approach we adopt: 
	after defining switchings and paths in Subsection~\ref{subsection:paths},
	we will restrict our attention to acyclic nets only.
	
	We associate jump functions with nets nonetheless, but for a different purpose:
	bounding the jump degree in a net will allow us to control the combinatorics
	of the elimination of evanescent cuts, in situations such as 
	that of Figure~\ref{fig:jumps_collapse}.
	Since switching paths can follow jumps, and those paths will be our main
	focus throughout the remaining of the paper, we chose to consider proof nets
	as equipped with jumps by default.
	Still, the reader should be aware that the main subject of interest is the
	underlying structure of bare proof nets.

\end{rem}

\subsection{Cut elimination}\label{cut-elim}

A $\definitive{reducible cut}$ is a cut $\cut{t}{s}$
such that:
\begin{itemize}
	\item $t$ is a variable and $\bar t\not\in\fv(s)$ (\definitive{axiom cut});
	\item or $t\in\units$ and $s\in\counits$, and $\jump(s)\not\in\{t,s\}$
		(\definitive{evanescent cut});
	\item or we can write $t=\otimes(t_1,\dots,t_n)$ and $s=\parr(s_1,\dots,s_n)$
		(\definitive{multiplicative cut}).
\end{itemize}

The \emph{substitution} $\gamma[t/x]$ of a tree $t$ for a variable $x$ in a tree
(or cut, or family of trees and/or cuts) $\gamma$ is defined in the usual way, with the additional
assumption that $\atoms(t)$ and $\atoms(\gamma)$ are disjoint.
By the definition of trees, this substitution is essentially linear:
each variable $x$ appears at most once in $\gamma$.

There are three basic cut elimination steps defined for bare proof nets, one for each kind of reducible cut:
\begin{itemize}
	\item the elimination of a multiplicative cut yields a family of cuts:
		we write
		\[\cut{\otimes(t_1,\dots,t_n)}{\parr(s_1,\dots,s_n)}
		\tom
		\cut{t_1}{s_1},\dotsc,\cut{t_n}{s_n}\]
		that we extend to nets by setting
		$(c,\vec{c};\vec{t})\tom(\vec c',\vec c;\vec t)$
		whenever $c\tom \vec c'$;
	\item the elimination of an axiom cut generates a substitution:
		we write
		$(\cut{x}{s},\overrightarrow{c};\overrightarrow{t})
		\toa
		(\overrightarrow{c};\overrightarrow{t})[s/\bar{x}]$ 
		whenever $\bar x\not\in\fv(s)$;
	\item the elimination of an evanescent cut just deletes that cut:
		we write $(\cut{\cw}{\w},\vec c;\vec t) \toe (\vec c;\vec t)$
		whenever $\jump_p(\w)\not\in\{\w,\cw\}$.\footnote{
			Since the cuts of a net are given as a family rather than a sequence,
			the order in which we write cuts in this definition is not relevant:
			despite our abusive notation,
			the reduced cut need not be the first in the enumeration,
			because this enumeration is not fixed.
		}
\end{itemize}
Then we write $p\to p'$ if $p\tom p'$ or $p\toa p'$ or $p\toe p'$.
Observe that if $p\to p'$ then $\atoms(p')\subseteq\atoms(p)$.

In order to define cut elimination between proof nets (and not bare
proof nets only), we need to modify the jump function.
Indeed, assume $p=(\cut{t}{s},\vec c;\vec t)$ and $p'$ is obtained from $p$
by reducing the cut $\cut{t}{s}$.
Then $\counits(p')\subseteq\counits(p)$, but if $\w\in\counits(p')$ and 
$\jump_p(\w)=t$, we need to redefine $\jump_{p'}(\w)$, as in general $t\not\in\ST(p')$.
This is done as follows:
\begin{itemize}
	\item if
		$\cut{t}{s}=\cut{\otimes(t_1,\dots,t_n)}{\parr(s_1,\dots,s_n)}$
		then for all $\w\in\counits(p)=\counits(p')$ such that
		$\jump_p(\w)=\otimes(t_1,\dots,t_n)$ (resp.
		$\parr(s_1,\dots,s_n)$), we set $\jump_{p'}(\w)= t_1$ (resp.
		$s_1$);%
		\footnote{
			We arbitrarily redirect the jumps to the first subtree to simplify the
			presentation, but we could equivalently have set $\jump_{p'}(\w)$ to be any of
			the immediate subtrees of $\jump_p(\w)$, non deterministically:
			in fact, this slight generalization is necessary to deal with cut
			elimination in resource nets.

			Other strategies for choosing the destination of a jump exist in the literature:
			for instance, one may be tempted to systematically redirect jumps to atoms,
			as it is done by Tortora de Falco~\cite[Definition 1.3.3]{Tortora00}.
			But this kind of transformation is not local
			and it would certainly complicate our arguments.
		}
	\item if $\cut{t}{s}=\cut{t}{x}$ and $p'$ is obtained from $p$ 
		by substituting $t$ for $\bar x$, then for all $\w\in\counits(p)=\counits(p')$ such that
		$\jump_p\in\{x,\bar x\}$, we set $\jump_{p'}(\w)=t$;
	\item if $\cut{t}{s}=\cut{\w}{\cw}$, then for all
		$\w'\in\counits(p')=\counits(p)\setminus\{\mu\}$ such that $\jump_p(\w')\in\{\w,\cw\}$, we set
		$\jump_{p'}(\w')=\jump_p(\w)$.
\end{itemize}
The result of eliminating the multiplicative cut (resp. axiom cut; evanescent cut)
of the net $p_0$ of Figure~\ref{fig:example-net} is depicted in
Figure~\ref{fig:redm} (resp. Figure~\ref{fig:reda}; Figure~\ref{fig:rede}).

\begin{figure}[t]
	\begin{center}
		\begin{tikzpicture}[
			baseline=(cdots.base),
			rounded corners,
			scale = 1,
		]
			\node(bot0) at (0,0) [lien] {$\bot$};
			\node at (bot0.south) [below right, atome] {$\mu_1$};
			\node(ax-x) at (1,1) [lien] {ax};
			\node(x) at (ax-x.west) [above left,atome] {$x$};
			\node    at (ax-x.east) [above right,atome] {$\bar x$};
			\node(cutm1) at (1,0) [lien] {cut};
			\node(cutm2) at (2.5,-.5) [lien] {cut};
			\node(onet) at (2.5,1) [lien] {$\one$};
			\node at (onet.south) [below right, atome] {$\lambda_1$};
			\node(ax-y) at (4,1) [lien] {ax};
			\node(y) at (ax-y.west) [above left,atome] {$y$};
			\node(by)at (ax-y.east) [above right,atome] {$\bar y$};
			\node(bot1) at (4,0) [lien] {$\bot$};
			\node at (bot1.south) [below right, atome] {$\mu_2$};
			\node(cuta) at (5.25,0) [lien] {cut};
			\node(bota) at (6,1) [lien] {$\bot$};
			\node at (bota.south) [below right, atome] {$\mu_3$};
			\node(onee) at (7,0) [lien] {$\one$};
			\node at (onee.south) [below left, atome] {$\lambda_2$};
			\node(cute) at (8,-.75) [lien] {cut};
			\node(bote) at (9,0) [lien] {$\bot$};
			\node at (bote.south) [below right, atome] {$\mu_4$};
			\node(one)  at (10,0) [lien] {$\one$};
			\node at (one.south) [below right, atome] {$\lambda_3$};

			\draw (bot0)--(bot0|- 0,-1);
			\draw (ax-x) to [out=west,in=west] (cutm1);
			\draw (ax-x) to [out=east,in=west] (cutm2);
			\draw (bot1)--(bot1|- 0,-1);
			\draw[par dessus] (onet) to [out=south, in=east] (cutm1);
			\draw (ax-y) to [out=west, in=east] (cutm2);
			\draw (ax-y) to [out=east, in=west] (cuta);
			\draw (bota) to [out=south, in=east] (cuta);
			\draw (onee) to [out=south, in=west] (cute);
			\draw (bote) to [out=south, in=east] (cute);
			\draw (one)--(one|- 0,-1);

			\draw[saut] (bot0) to [out=north, in=south west] (x);
			\draw[saut] (bot1) to [out=north, in=south] ($(ax-y.east)+(.1,-.1)$);
			\draw[saut] (bota) to [out=east, in=north] (onee);
			\draw[saut] (bote) to [out=north east, in=north west] (one);
	\end{tikzpicture}
	\end{center}
	\caption{The net $p'_m=(\cut {x}{\unit{\lambda_1}},\cut{\bar x}{y},\cut{\bar y}{\counit{\mu_3}},\cut{\unit{\lambda_2}}{\counit{\mu_4}};\counit{\mu_1},\counit{\mu_2},\unit{\lambda_3})$
	with $\jump_{p'_m}:\mu_1\mapsto x,\mu_2\mapsto \bar y,\mu_3\mapsto \lambda_2,\mu_4\mapsto\lambda_3$,
	so that $p_0\tom p'_m$.}
	\label{fig:redm}
\end{figure}

\begin{figure}[t]
	\begin{center}
		\begin{tikzpicture}[
			baseline=(cdots.base),
			rounded corners,
			scale = 1,
		]
			\node(bot0) at (0,0) [lien] {$\bot$};
			\node at (bot0.south) [below right, atome] {$\mu_1$};
			\node(par)  at (1,0) [lien] {$\parr$};
			\node(ax-x) at (1,1) [lien] {ax};
			\node(x) at (ax-x.west) [above left,atome] {$x$};
			\node    at (ax-x.east) [above right,atome] {$\bar x$};
			\node(cutm) at (2,-.75) [lien] {cut};
			\node(onet) at (2.5,1) [lien] {$\one$};
			\node at (onet.south) [below left, atome] {$\lambda_1$};
			\node(ten)  at (3,0) [lien] {$\otimes$};
			\node(bot1) at (4,0) [lien] {$\bot$};
			\node at (bot1.south) [below right, atome] {$\mu_2$};
			\node(bota) at (3.5,1) [lien] {$\bot$};
			\node at (bota.south) [below left, atome] {$\mu_3$};
			\node(onee) at (5.5,0) [lien] {$\one$};
			\node at (onee.south) [below left, atome] {$\lambda_2$};
			\node(cute) at (6.5,-.75) [lien] {cut};
			\node(bote) at (7.5,0) [lien] {$\bot$};
			\node at (bote.south) [below right, atome] {$\mu_4$};
			\node(one)  at (8.5,0) [lien] {$\one$};
			\node at (one.south) [below right, atome] {$\lambda_3$};

			\draw (bot0)--(bot0|- 0,-1);
			\draw (ax-x) to [out=west,in=north west] (par);
			\draw (ax-x) to [out=east,in=north east] (par);
			\draw (bot1)--(bot1|- 0,-1);
			\draw (par) to [out=south, in=west] (cutm);
			\draw (ten) to [out=south, in=east] (cutm);
			\draw (onet) to [out=south, in=north west] (ten);
			\draw (bota) to [out=south, in=north east] (ten);
			\draw (onee) to [out=south, in=west] (cute);
			\draw (bote) to [out=south, in=east] (cute);
			\draw (one)--(one|- 0,-1);

			\draw[saut] (bot0) to [out=north east, in=west] (par);
			\draw[saut] (bot1) to [out=north, in=south east] (bota);
			\draw[saut] (bota) to [out=east, in=north] (onee);
			\draw[saut] (bote) to [out=north east, in=north west] (one);
	\end{tikzpicture}
	\end{center}
	\caption{The net $p'_a=(\cut{\parr(x,\bar x)}{\otimes(\unit{\lambda_1},\counit{\mu_3}},\cut{\unit{\lambda_2}}{\counit{\mu_4}};\counit{\mu_1},\counit{\mu_2},\unit{\lambda_3})$
		with $\jump_{p'_a}:\mu_1\mapsto \parr(x,\bar x),\mu_2\mapsto \mu_3,\mu_3\mapsto \lambda_2,\mu_4\mapsto\lambda_3$,
		so that $p_0\toa p'_a$.}
	\label{fig:reda}
\end{figure}

\begin{figure}[t]
	\begin{center}
		\begin{tikzpicture}[
			baseline=(cdots.base),
			rounded corners,
			scale = 1,
		]
			\node(bot0) at (0,0) [lien] {$\bot$};
			\node at (bot0.south) [below right, atome] {$\mu_1$};
			\node(par)  at (1,0) [lien] {$\parr$};
			\node(ax-x) at (1,1) [lien] {ax};
			\node(x) at (ax-x.west) [above left,atome] {$x$};
			\node    at (ax-x.east) [above right,atome] {$\bar x$};
			\node(cutm) at (2,-.75) [lien] {cut};
			\node(onet) at (2.5,1) [lien] {$\one$};
			\node at (onet.south) [below left, atome] {$\lambda_1$};
			\node(ten)  at (3,0) [lien] {$\otimes$};
			\node(ax-y) at (4,1) [lien] {ax};
			\node(y) at (ax-y.west) [above left,atome] {$y$};
			\node(by)at (ax-y.east) [above right,atome] {$\bar y$};
			\node(bot1) at (4,0) [lien] {$\bot$};
			\node at (bot1.south) [below right, atome] {$\mu_2$};
			\node(cuta) at (5.25,0) [lien] {cut};
			\node(bota) at (6,1) [lien] {$\bot$};
			\node at (bota.south) [below right, atome] {$\mu_3$};
			\node(one)  at (7.5,0) [lien] {$\one$};
			\node at (one.south) [below right, atome] {$\lambda_3$};

			\draw (bot0)--(bot0|- 0,-1);
			\draw (ax-x) to [out=west,in=north west] (par);
			\draw (ax-x) to [out=east,in=north east] (par);
			\draw (bot1)--(bot1|- 0,-1);
			\draw (par) to [out=south, in=west] (cutm);
			\draw (ten) to [out=south, in=east] (cutm);
			\draw (onet) to [out=south, in=north west] (ten);
			\draw (ax-y) to [out=west, in=north east] (ten);
			\draw (ax-y) to [out=east, in=west] (cuta);
			\draw (bota) to [out=south, in=east] (cuta);
			\draw (one)--(one|- 0,-1);

			\draw[saut] (bot0) to [out=north east, in=west] (par);
			\draw[saut] (bot1) to [out=north, in=south] ($(ax-y.east)+(.1,-.1)$);
			\draw[saut] (bota) to [out=east, in=north]  (one);
	\end{tikzpicture}
	\end{center}
	\caption{The net $p'_e=(\cut{\parr(x,\bar x)}{\otimes(\unit{\lambda_1},y)},\cut{\bar y}{\counit{\mu_3}};\counit{\mu_1},\counit{\mu_2},\unit{\lambda_3})$
	with $\jump_{p'_e}:\mu_1\mapsto \parr(x,\bar x),\mu_2\mapsto \bar y,\mu_3\mapsto \lambda_3$,
	so that $p_0\toe p'_e$.}
	\label{fig:rede}
\end{figure}

\begin{figure}[t]
	\begin{center}
		\begin{tikzpicture}[
			baseline=(cdots.base),
			rounded corners,
			scale = 1,
		]
			\node(bot0) at (0,0) [lien] {$\bot$};
			\node at (bot0.south) [below right, atome] {$\mu_1$};
			\node(ax-x) at (1,1) [lien] {ax};
			\node(x) at (ax-x.west) [above left,atome] {$x$};
			\node    at (ax-x.east) [above right,atome] {$\bar x$};
			\node(cutm1) at (1,0) [lien] {cut};
			\node(cutm2) at (2.5,-.5) [lien] {cut};
			\node(onet) at (2.5,1) [lien] {$\one$};
			\node at (onet.south) [below left, atome] {$\lambda_1$};
			\node(bot1) at (4,0) [lien] {$\bot$};
			\node at (bot1.south) [below right, atome] {$\mu_2$};
			\node(bota) at (3.5,1) [lien] {$\bot$};
			\node at (bota.south) [below left, atome] {$\mu_3$};
			\node(one)  at (5.5,0) [lien] {$\one$};
			\node at (one.south) [below right, atome] {$\lambda_3$};

			\draw (bot0)--(bot0|- 0,-1);
			\draw (ax-x) to [out=west,in=west] (cutm1);
			\draw (ax-x) to [out=east,in=west] (cutm2);
			\draw (bot1)--(bot1|- 0,-1);
			\draw[par dessus] (onet) to [out=south, in=east] (cutm1);
			\draw (bota) to [out=south, in=east] (cutm2);
			\draw (one)--(one|- 0,-1);

			\draw[saut] (bot0) to [out=north, in=south west] (x);
			\draw[saut] (bot1) to [out=north, in=south east] (bota);
			\draw[saut] (bota) to [out=east, in=north] (one);
	\end{tikzpicture}
	\end{center}
	\caption{The net $p'=(\cut{x}{\unit{\lambda_1}},\cut{\bar x}{\counit{\mu_3}};\counit{\mu_1},\counit{\mu_2},\unit{\lambda_3})$
		with $\jump_{p'}:\mu_1\mapsto x,\mu_2\mapsto \mu_3,\mu_3\mapsto\lambda_3$,
		so that $p_0\toto p'$.}
	\label{fig:toto}
\end{figure}

We are in fact interested in the simultaneous elimination of any number of
reducible cuts, that we describe as follows.
We write $p\toto p'$ if
\[p=(c_1,\dotsc,c_k,\cut{x_1}{t_1},\dotsc,\cut{x_n}{t_n},\cut{\w_1}{\cw_1},\dots,\cut{\w_l}{\cw_l},\vec c;\vec t)\]
and \[p'=(\vec c'_1,\dotsc,\vec c'_k,\vec c;\vec t)[t_1/\bar x_1]\cdots[t_n/\bar x_n],\]
assuming that:
\begin{itemize}
	\item $c_i\tom\vec c'_i$ for $1\le i\le k$, 
	\item $\bar x_i\not\in\{x_1,\dotsc,x_n\}$ and $\bar x_i\not\in\fv(t_j)$ for $1\le i\le j\le n$, and
	\item $\jump_p(\w_i)\not\in\{\w_j,\cw_j\}$ for $1\le i\le j\le l$.
\end{itemize}
It should be clear that $p'$ is then obtained from $p$ by successively
eliminating the particular cuts we have selected, thus performing $k$ steps of
$\tom$, $n$ steps of $\toa$, $l$ steps of $\toe$, \emph{in no particular order}:
indeed, one can check that any two elimination steps of distinct cuts commute
on the nose.
The resulting jump function $\jump_{p'}$ can be described directly,
by inspecting the possible cases for $\jump_{p}(\w')$ with $\w'\in\counits(p')$:
\begin{itemize}
	\item if $c_i=\cut{\otimes(u_1,\dotsc,u_r)}{\parr(v_1,\dotsc,v_r)}$
		and, e.g., $\jump_p(\w')=\otimes(u_1,\dotsc,u_r)$
		then $\jump_{p'}(\w')=u_1[t_1/\bar x_1]\cdots[t_n/\bar x_n]$;
	\item if $\jump_p(\w')\in\{x_i,\bar x_i\}$
		then $\jump_{p'}(\w')=t_i[t_{i+1}/\bar x_{i+1}]\cdots[t_n/\bar x_n]$;
	\item if $\jump_p(\w')\in\{\w_i,\cw_i\}$
		then $\jump_{p'}(\w')=\rho(i)[t_1/\bar x_1]\cdots[t_n/\bar x_n]$,
		where $\rho:\{1,\dotsc,l\}\to\ST(p)$ is the redirection function
		inductively defined by $\rho(j)=\rho(i)$ if $\jump_{p}(\w_j)\in\{\w_i,\cw_i\}$
		(in which case $i<j$)
		and $\rho(j)=\jump_{p}(\w_j)$ otherwise;
	\item otherwise $\jump_{p'}(\w')=\jump_p(\w')[t_1/\bar x_1]\cdots[t_n/\bar x_n]$.
\end{itemize}
The result of simultaneously eliminating all the cuts
of the net $p_0$ of Figure~\ref{fig:example-net} is depicted in
Figure~\ref{fig:toto}.

This general description of parallel cut elimination is obviously not very handy.
In order not to get lost in notation, we will restrict our attention to the
particular case in which only cuts of the same nature are simultaneously
eliminated: we write $p\totom p'$ if $n=l=0$ (multiplicative cuts only), $p\totoa p$ if $k=l=0$
(axiom cuts only), and
$p\totoe p'$ if $n=k=0$ (evanescent cuts only).
Then we can decompose any parallel reduction $p\toto p'$ into three separate steps:
e.g., $p\totom \cdot \totoa \cdot \totoe p'$.\footnote{
Of course, the converse does not hold: for instance the reductions
$(\cut{\parr(x,\bar x)}{\otimes(y,z)};\bar y,\bar z)
\totom
(\cut x y,\cut {\bar x} z;\bar y,\bar z)
\totoa
(\cut y z;\bar y,\bar z)
$ cannot be performed in a single step, as the cut $\cut xy$ was newly created.
}

\subsection{Paths}
\label{subsection:paths}

In order to control the effect of parallel reduction on the size of proof nets,
we rely on a side condition involving the length of switching
paths, \emph{i.e.} paths in the sense of Danos--Regnier's correctness criterion \cite{DR89}.

Let us write $\parrs(p)$ (resp. $\tensors(p)$) for the 
set of the subtrees of $p$ of the form $\parr(t_1,\dotsc,t_n)$
(resp. $\otimes(t_1,\dotsc,t_n)$).
In our setting, a \definitive{switching} of a net $p$
is a map $I : \parrs(p)\to \ST(p)$ such that,
for each $t=\parr(t_1,\dots t_n)\in\parrs(p)$, $I(t)\in\{t_1,\dots,t_n\}$.
Given a net $p$ and a switching $I$ of $p$, the associated 
\definitive{switching graph} is the unoriented graph with 
vertices in $\ST(p)$ and edges given as follows:
\begin{itemize}
	\item one \definitive{axiom edge} $\sim_{\{x,\bar x\}}$ for each axiom $\{x,\bar x\}\subseteq\fv(p)$,
		connecting $x$ and $\bar x$;
	\item one \definitive{$\otimes$-edge} $\sim_{t,t_i}$ for each pair $(t,t_i)$ with
		$t=\otimes(t_1,\dotsc,t_n)\in\tensors(p)$,
		connecting $t$ and $t_i$;
	\item one \definitive{$\parr$-edge} $\sim_{t}$ for each $t\in\parrs(p)$,
		connecting $t$ and $I(t)$;
	\item one \definitive{jump edge} $\sim_\w$ for each $\w\in\counits(p)$,
		connecting $\w$ and $\jump_p(\w)$;
	\item one \definitive{cut edge} $\sim_c$ for each cut $c=\cut{t}{s}\in\cuts(p)$,
		connecting $t$ and $s$.
\end{itemize}
Whenever necessary, we may write, e.g., $\sim_e^{p}$ or $\sim_e^{p,I}$
for the edge $\sim_e$ to make the underlying net and switching explicit.
On the other hand, we will often simply write $e$ instead of $\sim_e$
for denoting an edge.
Each edge $e$ induces a symmetric relation, involving at most two subtrees of $p$:
we write $t\sim_e u$, and say $t$ and $u$ are \definitive{adjacent}
whenever $t$ and $u$ are connected by $e$.
\emph{A priori}, it might be the case that distinct edges induce the same relation:
for instance, if $c=\cut{x}{\bar x}\in\cuts(p)$,
we have $x\sim_c \bar x$ as well as $x\sim_{\{x,\bar x\}} \bar x$;
and if $\jump_p(\w)=\w'$ and $\jump_p(\w')=\w$,
we have $\w\sim_{\w} \w'$ as well as $\w\sim_{\w'} \w'$.
Avoiding such cycles is precisely the purpose of the correctness
criterion.

% Pour que la définition de chemin soit standard il *faut* demander que les l_i soint deux à deux distincts.
% Autrement il faut exclure tous les rebroussements (et pas seulement ceux des tenseurs) et ça devient adhoc.
Given a switching $I$ in $p$, an \definitive{$I$-path} is the data of a tree
$t_0\in\ST(p)$ and of a sequence $(e_1,\dotsc,e_n)$ of pairwise distinct
consecutive edges starting from $t_0$:
in other words, we require that there exist $t_1,\dotsc,t_n\in\ST(p)$ such
that, for each $i\in\{1,\dots,n\}$, $t_{i-1}\sim_{e_i} t_i$.%
\footnote{
	In standard terminology of graph theory, an $I$-path in $p$ is a trail
	in the switching graph induced by $p$ and $I$.
}
For instance, if $p=(;\otimes(x,y),\parr(\bar y,\bar x))$ and
$I(\parr(\bar y,\bar x))=\bar x$, then the chain of adjacencies
		$\parr(\bar x,\bar y)
		\sim_{\parr(\bar x,\bar y)} \bar x
		\sim_{\{x,\bar x\}} x
		\sim_{\otimes(x,y),x}\otimes(x,y)
		\sim_{\otimes(x,y),y} y
		\sim_{\{y,\bar y\}}\bar y$
		defines a maximal $I$-path in $p$ (see Figure~\ref{fig:path}).
\begin{figure}[t]
\begin{center}
	\begin{tikzpicture}
		\node[lien] (ten) at (-.8,-.8) {$\otimes$};
		\node[lien] (par) at (.8,-.8) {$\parr$};
		\node[lien] (axl) at (-.8,0) {ax};
		\node[lien] (axr) at (.8,0) {ax};
		\draw (axl.west) to[out=210,in=150]
			(ten.west);
		\node[anchor=base] at (-1.2,0) {$x$};
		\draw (axl.east) to[out=-30,in=150]
			(par.west);
		\node[anchor=base] at (-.4,0) {$\bar x$};
		\draw[par dessus] (axr.west) to[out=210,in=30]
			(ten.east);
		\node[anchor=base] at (.4,0) {$y$};
		\draw (axr.east) to[out=-30,in=30]
			node[sloped,pos=.8]{$\mid$} (par.east);
		\node[anchor=base] at (1.2,0) {$\bar y$};
		\draw (ten.south) -- (-.8,-1.3);
		\draw (par.south) -- ( .8,-1.3);
		\draw [chemin, dashed, red] 
			(.7,-1.3) .. controls +(-.1,.4) and +(.4,0) ..
			(-.8,-.1) .. controls +(-.6,0) and +(-.6,0) ..
			(-.8,-.95) .. controls +(.6,0) and +(-.6,0) ..
			(.8,-.1) .. controls +(.4,0) and +(0,.3) ..
			(1.3,-.6)
			;
	\end{tikzpicture}
\end{center}
\caption{A path in $(;\otimes(x,y),\parr(\bar y,\bar x))$
	for switching $I:\parr(\bar y,\bar x)\mapsto \bar x$
	(we strike out the other premise).
}
\label{fig:path}
\end{figure}

We write $\paths(p,I)$ for the set of all $I$-paths in $p$.
We write $\chi:t_0\leadsto_{p,I} t_n$ whenever
$\chi=t_0\sim_{e_1}\cdots\sim_{e_n}t_n$
is an $I$-path from $t_0$ to $t_n$ in $p$:
with these notations, we say $\chi$ \definitive{visits}
the trees $t_0,\dotsc,t_n$, 
and $\chi$ \definitive{crosses} the edges
$e_1,\dotsc,e_n$;
moreover we write $\ln(\chi)=n$ for the length of $\chi$.
A \definitive{subpath} of $\chi$ is any $I$-path of the form
$t_i\sim_{e_{i+1}}\cdots \sim_{e_{i+k}} t_{i+k}$.
We write $\reverse{\chi}$ for the reverse $I$-path:
$\reverse{\chi}=t_n\sim_{e_n}\cdots\sim_{e_0}t_0$.
The \definitive{empty path} from $t\in\ST(p)$
is the only $\epsilon_t:t\leadsto_{p,I} t$ of length $0$.
We say $I$-paths $\chi$ and $\chi'$ are \definitive{disjoint}
if no edge is crossed by both $\chi$ and $\chi'$.
Observe that disjoint paths can visit common trees:
in particular, if $\chi:t\leadsto_{p,I} s$ and $\chi':s\leadsto_{p,I} u$
are disjoint, we write $\chi\chi':t\leadsto_{p,I} u$ for the concatenation of
$\chi$ and $\chi'$.

We call \definitive{path} in $p$ any $I$-path for $I$ a switching of $p$,
we write $\paths(p)$ for the set of all paths in $p$,
and we denote by $\ln(p)=\max\{\ln(\chi)\tq\chi\in\paths(p)\}$
the maximal length of a path in $p$.
We then write $\chi:t\leadsto_p s$ if $\chi\in\paths(p)$ is a path from $t$ to $s$ in $p$,
and we write $t\leadsto_p s$, or simply $t\leadsto s$, whenever such a path exists.
This relation on $\ST(p)$ is reflexive (\emph{via} the empty path $\epsilon_t:t\leadsto t$)
and symmetric (\emph{via} reversing paths).
Observe that if	$\chi_1,\cdots,\chi_n\in\paths(p)$ are pairwise disjoint,
then there exists $I$ such that $\chi_i\in\paths(p,I)$ for $1\le i\le n$.
This does not make the relation $\leadsto_p$ transitive:
for instance, if $p=(;x,\parr(\bar x,\bar y),y)$,
we have paths $x\leadsto \parr(\bar x,\bar y)$ and $\parr(\bar x,\bar y)\leadsto y$,
but both cross $\parr(\bar x,\bar y)$ for different switchings,
and indeed, there is no path $x\leadsto y$.

We say a net $p$ is \definitive{acyclic} if, for all $\chi\in \paths(p)$
and $t\in\ST(p)$, $\chi$ visits $t$ at most once:
in other words, there is no (non-empty) \definitive{cycle} $\chi:t\leadsto t$.
Notice that, given any family of pairwise distinct cuts in an acyclic net $p$,
it is always possible to satisfy the side conditions on
free variables and on jumps necessary to reduce these cuts in parallel
(provided each cut in the family has the shape of a multiplicative,
axiom or evanescent cut).
Moreover, it is a very standard result that acyclicity is preserved 
by cut elimination:

\begin{lem}
	\label{lemma:acyclicity}
	If $p'$ is obtained from $p$ by cut elimination 
	and $p$ is acyclic then so is $p'$.
\end{lem}
\begin{proof}
	It suffices to check that if $p\to p'$ then 
	any cycle in $p'$ induces a cycle in $p$.\footnote{
		We do not detail the proof as it is quite standard.
		We will moreover generalize this technique to all paths
		(and not only cycles) in the next section.
	}
\end{proof}

From now on, we consider acyclic nets only.

\section{Bounding the size of antireducts: three kinds of cuts}\label{section:size}

In this section, we show that the loss of size during a parallel reduction
$p\totom q$, $p\totoa q$ or $p\totoe q$
is directly controlled by $\ln(p)$, $\jd(p)$ and $\size(q)$:
more precisely, we show that the ratio $\frac{\size(p)}{\size(q)}$ is bounded
by a function of $\ln(p)$ and $\jd(p)$ in each case.

\subsection{Elimination of multiplicative cuts}
The elimination of multiplicative cuts 
cannot decrease the size by more than a half:
\begin{lem}\label{mult}
	If $p\totom q$ then $\size(p)\leq 2\size(q)$.
\end{lem}

\begin{proof}
	Since the elimination of a multiplicative cut does not affect the rest of the
	(bare) net, it is sufficient to observe that if $c\tom\vec c$
	then $\size(c)=2+\size(\vec c)\le 2\size(\vec c)$.\footnote{
		This is due to the fact that we distinguish between 
		strict connectives and their nullary versions,
		that are subject to evanescent reductions.
	}
\end{proof}
So in this case, $\ln(p)$ and $\jd(p)$ actually play no rôle.

\subsection{Elimination of axiom cuts}

Observe that:
\begin{itemize}
	\item if $x\in\fv(\gamma)$ then $\size(\gamma[t/x])=\size(\gamma)+\size(t)-1$;
	\item if $x\not\in\fv(\gamma)$ then $\size(\gamma[t/x])=\size(\gamma)$.
\end{itemize}
It follows that, in the elimination of a single axiom cut
$p\toa q$, we have $\size(p)=\size(q)+2$.
But we cannot reproduce the proof of Lemma~\ref{mult} for $\totoa$:
as depicted in Figure~\ref{fig:collapse},
a chain of arbitrarily many axiom cuts may reduce into a single wire.
We can bound the length of those chains by $\ln(p)$, however,
and this allows us to bound the loss of size during reduction.

\begin{lem}
	\label{ax} If $p\totoa q$ then $\size(p)\leq (\ln(p)+1)\size(q)$.
\end{lem}
\begin{proof}
	Assume
	$p=(\cut{x_1}{t_1},\dotsc,\cut{x_n}{t_n},\vec c;\vec s)$
	and $q=(\vec c;\vec s)[t_1/\bar x_1]\cdots[t_n/\bar x_n]$
	with $\bar x_i\not\in\{x_1,\dotsc,x_n\}$ and
	$\bar x_i\not\in\fv(t_j)$ for $1\le i\le j\le n$.
	To establish the result in this case,
	we make the chains of eliminated axiom cuts explicit.

	Due to the condition on free variables, we can partition
	$\cut{x_1}{t_1},\dotsc,\cut{x_n}{t_n}$
	into tuples $\vec c_1,\dotsc,\vec c_k$ of the shape
	$\vec c_i=(\cut{x^i_0}{\bar x^i_1},\dotsc,\cut{x^i_{n_i-1}}{\bar x^i_{n_i}},\cut{x^i_{n_i}}{t^i})$
	so that:
	\begin{itemize}
		\item $x_j^i\in\{x_1,\dotsc,x_n\}$ for $1\le i\le k$ and $0\le j\le n_i$;
		\item $\bar x_j^i\in\{t_1,\dotsc,t_n\}$ for $1\le i\le k$ and $1\le j\le n_i$;
		\item $t^i\in\{t_1,\dotsc,t_n\}$ for $1\le i\le k$;
		\item each $\vec c_i$ is maximal with this shape, 
			\emph{i.e.} $\bar x^i_0\not\in\{t_1,\dotsc,t_n\}$
			and $t^i\not\in\{\bar x_1,\dotsc,\bar x_n\}$.
	\end{itemize}
	Without loss of generality, we can moreover require that,
	if $i<i'$, then $\cut{x^i_{n_i}}{t^i}$
	occurs before $\cut{x^{i'}_{n_{i'}}}{t^{i'}}$ in the
	tuple $(\cut{x_1}{t_1},\dotsc,\cut{x_n}{t_n})$.
	Moreover observe that, by the condition on free variables,
	the order of the cuts in each $\vec c_i$ 
	is necessarily the same as in $(\cut{x_1}{t_1},\dotsc,\cut{x_n}{t_n})$.

	By a standard result on substitutions, if $x\not=y$, $x\not\in\fv(v)$
	and $y\not\in\fv(u)$ then $\gamma[u/x][v/y]=\gamma[v/y][u/x]$:
	\begin{itemize}
		\item if $i\not =i'$,
			we have $\bar x^i_j\not=\bar x^{i'}_{j'}$
			for $0\le j\le n_i$ and $0\le j'\le n_{i'}$, 
			so the substitutions $[\bar x^i_{j+1}/\bar x^i_j]$
			and $[\bar x^{i'}_{j'+1}/\bar x^{i'}_{j'}]$
			always commute for $1\le i<i'\le k$,
			$0\le j<n_i$ and $0\le j'<n_{j'}$;
		\item if $i<i'$ and $\cut{x^{i'}_{j'}}{\bar x^{i'}_{j'+1}}$ occurs before
			$\cut{x^i_{n_i}}{t^i}$ in $(\cut{x_1}{t_1},\dotsc,\cut{x_n}{t_n})$, 
			the condition on free variables
			imposes that $\bar x^{i'}_{j'}\not\in\fv(t^i)$
			so the substitutions $[t^i/\bar x^i_{n_i}]$
			and $[\bar x^{i'}_{j'+1}/\bar x^{i'}_{j'}]$
			commute in this case.
	\end{itemize}

	By iterating those two observations,
	we can reorder the substitutions in $q$ 
	and obtain:
	\begin{align*}
		q&=(\vec c;\vec s)[t_1/\bar x_1]\cdots[t_n/\bar x_n]
		\\&=
		(\vec c;\vec s)
		[\bar x^1_1/\bar x^1_0]\cdots [\bar x^1_{n_1}/\bar x^1_{n_1-1}][t^1/\bar x^1_{n_1}]
		\cdots
		[\bar x^k_1/\bar x^k_0]\cdots [\bar x^k_{n_k}/\bar x^k_{n_k-1}][t^k/\bar x^k_{n_k}]
		\\&=
		(\vec c;\vec s)[t^1/\bar x^1_0]\cdots[t^k/\bar x^k_0]
		.
	\end{align*}
	It follows that $\size(q)=\size(\vec c)+\size(\vec s)+\sum_{i=1}^k\size(t^i)-k$.
	For $1\le i\le k$, $\vec c_i$ induces a path $\bar x^0_i\leadsto t^i$
	of length $2n_i+2$ ($n_i+1$ cuts and $n_i+1$ axioms). 
	Hence $2n_i\le\ln(p)-2$ and:
	\begin{align*}
		\size(p)
		&=\size(\vec c)+\size(\vec s)+\sum_{i=1}^k(\size(t^i)+2n_i+1)
		\\
		&\le \size(\vec c)+\size(\vec s)+\sum_{i=1}^k \size(t^i)+k(\ln(p)-1)
		\\
		&\le \size(q)+k\ln(p).
	\end{align*}

	To conclude, it will be sufficient to prove that $\size(q)\ge k$.
	For $1\le i\le k$, let $A_i=\{j>i\tq \bar x^j_0\in\fv(t^i)\}$,
	and then let $A_0=\{i\tq \bar x^i_0\in\fv(\vec c,\vec s)\}$.
	It follows from the construction that $\{A_0,\dotsc,A_{k-1}\}$
	is a partition (possibly including empty sets) of $\{1,\dotsc,k\}$.
	By construction, for each $j\in A_i$, $\bar x^j_0$ 
	is a strict subtree of $t^i$: it follows that $\size(t^i)>\card A_i$.
	Now consider $q_i=(\vec c;\vec s)[t^1/\bar x^1_0]\cdots[t^i/\bar x^i_0]$
	for $0\le i\le k$ so that $q=q_k$.
	For $1\le i\le k$, we obtain $\size(q_i)=\size(q_{i-1})+\size(t^i)-1\ge\size(q_{i-1})+\card A_i$.
	Also observe that $\size(q_0)=\size(\vec c;\vec s)\ge\card A_0$.
	Then we can conclude: $\size(q)=\size(q_k)\ge\sum_{i=0}^k\card A_i=k$.
\end{proof}

\subsection{Elimination of evanescent cuts}

We now consider the case of a reduction $p\totoe q$:
we bound the maximal number of evanescent cuts appearing in $p$ 
by a function of $\ln(p)$, $\jd(p)$ and $\size(q)$.

We rely on the basic fact that if $t\in\ST(q)\subseteq\ST(p)$,
then there are at most $\jd(p)$ evanescent cuts of $p$ that jump to $t$.
The main difficulty is that an evanescent cut of $p$ can jump 
to another evanescent cut of $p$, that is also eliminated in the step $p\totoe q$.
See Figure~\ref{fig:jumps} for a graphical representation of the critical case.
To deal with this phenomenon, we observe that a sequence of cuts
$\cut{\w_1}{\cw_1},\dots,\cut{\w_n}{\cw_n}$ with $\jump_p(\w_i)\in\{\cw_{i+1},\w_{i+1}\}$
for all $i\in\{1,\dots,n-1\}$, induces a path of length at least $n$:
hence $n\leq\ln(p)$.
\begin{defi}
	We define for all $n\in\N$ and all $t\in\ST(p)$, $\I^n(t)$ as follows :
	$\I^0(t)=\jump_p^{-1}(t)=\{\w\in\counits(p)\tq \jump_p(\w)=t\}$,
	and $\I^{m+1}(t)=\{\w'\in\counits(p)\tq\jump_p(\w')\in \{\w,\cw\},
	\cut{\w}{\cw}\in\cuts(p), \w\in \I^m(t)\}$.
\end{defi}
We can already observe that $\card\I^0(t)=\jd(t)$.
This definition is parametrized by $\jump_p$, and may we write
$\I^n_p(t)$ to make the underlying net explicit.

\begin{figure}[t]
	\begin{center}
		\begin{tikzpicture}[
				decoration={brace},
				baseline=(cdots.base),
			]

			%t:
			\node(t)at(9,0){$t$};

			%J1(t):
			\node[lien](w1-1)at(7,1){$\bot$};
			\node[lien](w1-n)at(7,-1){$\bot$};
			\node[lien](cw1-1)at(6,1){$1$};
			\node[lien](cw1-n)at(6,-1){$1$};
			\draw[reseau](cw1-1)--++(0,-0.5)--++(1,0)--(w1-1);
			\draw[reseau](cw1-n)--++(0,-0.5)--++(1,0)--(w1-n);

			\node at($(w1-1)+(-0.5,-1)$){$\rvdots$};

			\draw[acc]($(cw1-n)+(-0.5,-0.2)$)--($(cw1-1)+(-0.5,0.2)$);
			\node[rotate=-90, scale=0.7]at($(cw1-1)+(-1,-1)$){$\leq
			\jd(p)$};

			\draw[saut](w1-1)to[out=45,in=north](t);
			\draw[saut](w1-n)to[out=45,in=west](t);

			%J2(t):
			\node[lien](cw2-1-1)at(3,2){$1$};
			\node[lien](cw2-1-n)at(3,1){$1$};
			\node[lien](w2-1-1)at(4,2){$\bot$};
			\node[lien](w2-1-n)at(4,1){$\bot$};
			
			\node at($(w2-1-1)+(-0.5,-0.8)$){$\rvdots$};
			
			\node[lien](cw2-n-n)at(3,-2){$1$};
			\node[lien](cw2-n-1)at(3,-1){$1$};
			\node[lien](w2-n-n)at(4,-2){$\bot$};
			\node[lien](w2-n-1)at(4,-1){$\bot$};
			
			\node at($(w2-n-1)+(-0.5,-0.8)$){$\rvdots$};
			
			\node at($(w2-1-n)+(-0.5,-1.25)$){$\rvdots$};

			\draw[reseau](cw2-1-1)--++(0,-0.5)--++(1,0)--(w2-1-1);
			\draw[reseau](cw2-1-n)--++(0,-0.5)--++(1,0)--(w2-1-n);
			\draw[reseau](cw2-n-1)--++(0,-0.5)--++(1,0)--(w2-n-1);
			\draw[reseau](cw2-n-n)--++(0,-0.5)--++(1,0)--(w2-n-n);

			\draw[saut](w2-1-1) to[out=15,in=90] ($(w1-1)+(-0.3,0.3)$);
			\draw[saut](w2-1-n) to[out=45,in=90] ($(cw1-1)+(0.3,0.3)$);
			
			\draw[saut](w2-n-1) to[out=15,in=90]($(w1-n)+(-0.3,0.3)$);
			\draw[saut](w2-n-n) to[out=45,in=90]($(cw1-n)+(0.3,0.3)$);

			\draw[acc]($(cw2-1-n)+(-0.5,-0.2)$)--($(cw2-1-1)+(-0.5,0.2)$);
			\node[rotate=-90,scale=0.7]at($(cw2-1-1)+(-1,-1)$){$\leq
			2\jd(p)$};

			\draw[acc]($(cw2-n-n)+(-0.5,-0.2)$)--($(cw2-n-1)+(-0.5,0.2)$);
			\node[rotate=-90,scale=0.7]at($(cw2-n-1)+(-1,-1)$){$\leq
			2\jd(p)$};

			%Jn-1(t):
			\node[lien](cwn-1)at(-1,2.5){$1$};
			\node[lien](wn-1)at(0,2.5){$\bot$};
			\node[lien](cwn-n)at(-1,-2.5){$1$};
			\node[lien](wn-n)at(0,-2.5){$\bot$};
			
			\draw[reseau](cwn-1)--++(0,-0.5)--++(1,0)--(wn-1);
			\draw[reseau](cwn-n)--++(0,-0.5)--++(1,0)--(wn-n);

			\draw[saut](wn-1)to[out=45,in=125]($(wn-1)+(0.6,-0.25)$);
			\draw[saut](wn-n)to[out=45,in=125]($(wn-n)+(0.6,0.25)$);

			\node at($(wn-1)+(1,-0.35)$){$\dots$};
			\node at($(wn-n)+(1,0)$){$\dots$};

			\node at(-0.5,0){$\vdots$};

			%Jn(t):
			\node[lien](wwn-1)at(-2.5,3.5){$\bot$};
			\node[lien](wwn-1n)at(-2.5,2){$\bot$};
			\draw[reseau](wwn-1)--++(0,-0.5);
			\draw[reseau](wwn-1n)--++(0,-0.5);

			\draw[saut](wwn-1)to[out=45,in=105]($(cwn-1)+(0.7,0.3)$);
			\draw[saut](wwn-1n)to[out=45,in=105]($(cwn-1)+(0.3,0.3)$);

			\node at($(wwn-1n)+(0,0.7)$){$\vdots$};
			
			\node[lien](wwn-n1)at(-2.5,-1.7){$\bot$};
			\node[lien](wwn-nn)at(-2.5,-3.2){$\bot$};
			\draw[reseau](wwn-n1)--++(0,-0.5);
			\draw[reseau](wwn-nn)--++(0,-0.5);

			\draw[saut](wwn-n1)to[out=45,in=105]($(cwn-n)+(0.7,0.3)$);
			\draw[saut](wwn-nn)to[out=45,in=105]($(cwn-n)+(0.3,0.3)$);

			\node at($(wwn-nn)+(0,0.7)$){$\vdots$};

			\node at(-2.5,0){$\vdots$};

			\draw[acc]($(wwn-nn)+(-0.5,-0.2)$)--($(wwn-1)+(-0.5,0.2)$);
			\node[rotate=-90,scale=0.7]at($(wwn-1)+(-1,-3)$){$\leq
			\left(2\jd(p)\right)^{n+1}$};

			%fleche basse et J^i
			\node(J1)at($(wwn-1)+(9.5,1)$){$\I^0(t)$};
			\node(J2)at($(wwn-1)+(6.5,1)$){$\I^1(t)$};
			\node(Jn-1)at($(wwn-1)+(2.8,1)$){$\I^{n-1}(t)$};
			\node(Jn)at($(wwn-1)+(0,1)$){$\I^{n}(t)$};

			\draw[->]($(J1)+(0,-8.3)$)--($(Jn)+(1,-8.3)$);
			\node[scale=0.7]at($(J2)+(-1,-8.1)$){$n\leq
			\ln(p)$};

		\end{tikzpicture}
	\end{center}
	\caption{Evanescent reductions : critical case}
	\label{fig:jumps}
\end{figure}

\begin{lem}\label{borne_qi} 
	Let $p,q$ be two nets such that $p\totoe q$. Then:
	\begin{enumerate}
		\item for all $t\in\ST(p)$, $\card\left(\bigcup_{i\in\N}\I_p^i(t)\right)\leq
			(2\jd(p))^{\ln(p)+1}$;
		\item there are at most $\size(q)\times
			(2\jd(p))^{\ln(p)+1}$
			evanescent cuts in $\cuts(p)$.
	\end{enumerate}
\end{lem}

\begin{proof}
	We first establish that the set $\{n\in\N\tq\I^n(t)\neq\emptyset\}$ is
	finite for all $t$. Indeed, for each $\w_{n}\in\I^{n}(t)$, there is a 
	sequence of cuts $c_0,\dotsc,c_{n-1}$ such that,
	writing $c_i=\cut{\cw_i}{\w_i}$,
	the unique path from $\w_{n}$ to $t$ is
	$\chi=\chi_n\cdots\chi_1(\w_0\sim_{\w_0}t)$ where,
	for $1\le i\le n$:
	\begin{itemize}
		\item 
			either $\jump_p(\w_{i})=\cw_{i-1}$ and $\chi_i=\w_{i}\sim_{\w_{i}}\cw_{i-1}\sim_{c_{i-1}}\w_{i-1}$;
		\item 
			or $\jump_p(\w_{i})=\w_{i-1}$ and $\chi_i=\w_{i}\sim_{\w_{i}}\w_{i-1}$.
	\end{itemize}
	We observe that $\ln(\chi)\ge n+1$, and then we
	deduce $n<\ln(p)$ as soon as $\I^n(t)\neq\emptyset$.

	Now we bound the size of each $\I_p^n(t)$: we show that 
	$\card\I_p^n(t)\le (2\jd(p))^{n+1}$, by induction on $n$.
	We already have $\card\I_p^0(t)=\jd(t)\leq\jd(p)$.
	Now assume the result holds for $n\ge 0$.
	Then, for each $c=\cut{\w}{\cw}\in\cuts(p)$ such that $\w\in\I_p^n(t)$, the
	number of $\w'$ such that $\jump_p(\w')\in c$ is 
	at most $\jd(\w)+\jd(\cw)\le 2\jd(p)$.
	We obtain:
	$\card\I_p^{n+1}(t)\leq 2\jd(p)\card\I_p^n$, which enables the induction.

	We thus get
	$\card\left(\bigcup_{i\in\N}\I_p^i(t)\right)\leq
	\sum_{i=0}^{\ln(p)-1}(2\jd(p))^{i+1}\leq (2\jd(p))^{\ln(p)+1}$
	which entails (1).
	To deduce (2) from (1), it will be sufficient to show that,
	for each $\w\in\counits(p)$, there exists $t\in\ST(q)$
	such that $\w\in\I_p^k(t)$ for some $k\in\N$:
	indeed the number of evanescent cuts in $p$ is obviously bounded 
	by $\card\counits(p)$.

	For that purpose, write $\w_0=\w$ and let
	$(c_i)_{i\in\{1,\dots,k\}}$ be the longest sequence 
	of cuts $c_i=\cut{\cw_i}{\w_i}\in\cuts(p)$
	such that, for all $i\in\{0,\dots,k-1\}$,
	$\jump_p(\w_i)\in\{\w_{i+1},\cw_{i+1}\}$:
	such a maximal sequence exists by acyclicity.
	Necessarily,
	$\jump_p(\w_k)$ is not part of an evanescent cut in $\cuts(p)$,
	so $\jump_p(\w_k)\in\ST(q)$: we conclude since $\w_0\in\I_p^k(\jump_p(\w_k))$.
\end{proof}

Writing $\psi(i,j,k)= i(1+2(2j)^{k+1})$, we obtain:
\begin{lem}\label{affs_size} 
	If $p\totoe q$, then $\size(p)\leq \psi(\size(q),\jd(p),\ln(p))$.
\end{lem}

\begin{proof}
	Writing 
	$p=(\cut{\w_1}{\cw_1},\dots,\cut{\w_n}{\cw_n},\vec c;\vec t)$ and 
	$q=(\vec c;\vec t)$, 
	we obtain
	\[\size(p)= \size(q) + 2n \leq \size(q)(1+2(2\jd(p))^{\ln(p)+1})\]
	by Lemma~\ref{borne_qi}.
\end{proof}

\subsection{Towards the general case}

Recall that any parallel cut elimination step $p\toto q$
can be decomposed into, e.g.: $p\totoe p'\totom p''\totoa q$.
We would like to apply the previous results to this sequence of 
reductions, in order to bound the size of $p$ by a function 
of $\size(q)$, $\ln(p)$ and $\jd(p)$.
Observe however that this would require us to infer a bound on $\ln(p'')$
from the bounds on $p$, in order to apply Lemma~\ref{ax}.

More generally, to be able to apply our results to a sequence 
of reductions $p\toto \cdots\toto q$, we need to ensure 
that for any reduction $p\toto p'$, we can bound $\ln(p')$ and $\jd(p')$
by functions of $\ln(p)$ and $\jd(p)$. This is the subject of 
the following two sections.

\section{Variations of $\ln(p)$ under reduction}\label{section:vc}

Here we establish that the possible increase of $\ln(p)$ under reduction is bounded.
It should be clear that:
\begin{lem}\label{cc_ax_e}
	If $p\totoa q$ or $p\totoe q$, then $\ln(q)\leq\ln(p)$.
\end{lem}
Indeed axiom and evanescent reductions only shorten paths, without really changing the
topology of the net.

In the case of multiplicative cuts however, cuts are
duplicated and new paths are created.
Consider for instance a net $r$, as in Figure~\ref{fig:slipknot},
obtained from three nets $p_1$, $p_2$ and $q$,
by forming the cut $\cut{\otimes(t_1,t_2)}{\parr(s_1,s_2)}$
where $t_1\in \ST(p_1)$,  $t_2\in \ST(p_2)$  and $s_1,s_2\in \ST(q)$.
Observe that, in the reduct $r'$ obtained by forming two cuts 
$\cut{t_1}{s_1}$ and $\cut{t_2}{s_2}$, we may very well
form a path that travels from $p_1$ to $q$ then $p_2$;
while in $p$, this is forbidden by any switching of $\parr(s_1,s_2)$.
For instance, if we consider $I(\parr(s_1,s_2))=s_1$,
we may only form a path between $p_1$ and $p_2$ through $\otimes(t_1,t_2)$, or
a path between $q$ and one of the $p_i$'s, through $s_1$ and the cut.

\begin{figure}[t]
	\begin{center}
	\begin{tikzpicture}[baseline=(cdots.base)]
		\draw[reseau] (-2.5,1)--(-1.25,1)--(-1.25,2)--(-2.5,2)--cycle;
		\node at (-1.875,1.5) {$p_1$};
		\draw[reseau] (-0.75,1)--(0.5,1)--(0.5,2)--(-0.75,2)--cycle;
		\node at (-0.125,1.5) {$p_2$};
		\node[lien] (ten) at (-1,0.35) {$\otimes$};
		\draw (-1.875,1) to[out=-90,in=125] (ten);
		\draw (-0.125,1) to[out=-90,in=55] (ten);
		\draw[reseau] (1,1)--(3,1)--(3,2)--(1,2)--cycle;
		\node at (2,1.5) {$q$};
		\node[lien] (par) at (2,0.35) {$\parr$};
		\draw (1.5,1) to[out=-90,in=125] (par);
		\draw (2.5,1) to[out=-90,in=55] node[sloped]{$\mid$} (par);
		\node[lien] (cut) at (0.5,-0.5) {cut};
		\draw (ten) to[out=-90,in=180] (cut);
		\draw (cut) to[out=0,in=-90] (par);
		\draw[chemin, dotted, blue] (-2.4,1.2) .. controls (-2.4,2) and (-2.1,1.95) 
			.. (-1.8,1.95) .. controls (-1.2,2) and (-1.3,1.5)
			.. (-1.5,1.1) .. controls (-1.5,0.5) and (-0.4,0.8)
			.. (-0.4,1.1) .. controls (-0.4,1.5) and (-0.2,1.8)
			.. (0,1.8) .. controls (0.2,1.8) and (0.3,1)
			.. (0.3,1.2) ;
		\draw[chemin, dashed, red] (-2.3,1.1) .. controls (-2.3,2) and (-1.55,2) 
			.. (-1.45,1.75) .. controls (-1.35,1.25) and (-1.7,1.5)
			.. (-1.7,1) .. controls (-1.7,-0.2) and  (0,-0.8)
			.. (0.5,-0.8).. controls (2,-0.8) and (1.8,0)
			.. (1.3,1.3) .. controls (1.3,1.8)
			.. (2,1.8) .. controls (2.6,1.8)
			.. (2.6,1.1) ;
	\end{tikzpicture}
	\qquad
	\begin{tikzpicture}[baseline=(cdots.base)]
		\draw[reseau] (-2.5,1)--(-1.25,1)--(-1.25,2)--(-2.5,2)--cycle;
		\node at (-1.875,1.5) {$p_1$};
		\draw[reseau] (-0.75,1)--(0.5,1)--(0.5,2)--(-0.75,2)--cycle;
		\node at (-0.125,1.5) {$p_2$};
		\node[lien] (cutPR) at (-0.5,-0.25) {cut};
		\node[lien] (cutQR) at (1.25,-0.5) {cut};
		\draw (-1.875,1) to[out=-90,in=180] (cutPR);
		\draw (-0.125,1) to[out=-90,in=180] (cutQR);
		\draw[reseau] (1,1)--(3,1)--(3,2)--(1,2)--cycle;
		\node at (2,1.5) {$q$};
		\draw[par dessus] (1.5,1) to[out=-90,in=0] (cutPR);
		\draw (2.5,1) to[out=-90,in=0] (cutQR);
		\draw[chemin,dash dot,purple] (-2.3,1.1) .. controls (-2.3,2) and (-1.55,2) 
			.. (-1.45,1.75) .. controls (-1.35,1.25) and (-1.7,1.5)
			.. (-1.7,1) .. controls (-1.7,-0.2) and (-1,-0.6)
			.. (-0.5,-0.6) .. controls (0.5,-0.6) and (1.3,0)
			.. (1.3,1.3) .. controls (1.3,1.8)
			.. (2,1.8) .. controls (2.6,1.8)
			.. (2.6,1) .. controls (2.6,-0.5) and (2,-0.85)
			.. (1.25,-0.85) .. controls (0.5,-0.85) and (-0.4,0)
			.. (-0.4,1.1) .. controls (-0.4,1.5) and (-0.2,1.8)
			.. (0,1.8) .. controls (0.2,1.8) and (0.3,1)
			.. (0.3,1.2) ;
	\end{tikzpicture}
	\end{center}
	\caption{A cut, the resulting slipknot, and examples of paths before and after reduction}
	\label{fig:slipknot}
\end{figure}

In the remainder of this section, we fix a reduction step $p\totom q$,
and we show that the previous example describes a general mechanism:
a path $\chi$ in $q$ that is not already in $p$ must involve
a subpath $\chi'$ between two residuals of a cut of $p$ that 
was eliminated in $p\totom q$.
We refer to this situation as a \definitive{slipknot} in $\chi$.

More formally,
consider $c=\cut{t_0}{s_0}\in \cuts(p)$
with $t_0=\otimes(t_1,\dotsc,t_n)$ and $s_0=\parr(s_1,\dotsc,s_n)$
and assume $c$ is eliminated in the reduction $p\totom q$:
then the \definitive{residuals} of $c$ in $q$ are the cuts
$\cut{t_i}{s_i}\in\cuts(q)$ for $1\le i\le n$.
For any edge $e$, we write $(e)$ for any length 1 path $t\sim_e s$.
If $\chi\in \paths(q)$, a \definitive{slipknot} of $\chi$ is 
any subpath $(d)\xi(d')$ where $d$ and $d'$ are (necessarily distinct)
residuals of the same cut in $p$.
In the remaining of this section,
we show that a path in $q$ is necessarily obtained by alternating paths
(essentially) in $p$ and slipknots in $q$, that recursively consist of such
alternations.
This will allow us to bound $\ln(q)$ depending on $\ln(p)$,
by reasoning inductively on these paths.

\subsection{Preserved paths}
Notice that $\ST(q)\subseteq\ST(p)$
and, given a switching $J$ of $q$, it is always possible to extend 
$J$ into a switching $I$ of $p$:
to determine $I$ uniquely amounts to select a premise for each
$\parr$-tree in an eliminated cut.
% À bien y réfléchir, les résidus ne dépendent que de $c$ (et du fait qu’elle est éliminée).
% Comme on est linéaire et grâce à l’arnaque, on peut se contenter de dire les choses comme suit:

Let $J$ be a switching of $q$ and $I$ an extension of $J$ on $p$.
Observe that if $t\sim^{p,I}_e t'$ and neither $t$ nor $t'$
is an element of an eliminated cut,
then $e$ is also an edge of $q$ that is not a residual cut;
conversely, if $t\sim^{q,J}_e t'$ and $e$ is not a residual cut,
then $e$ is also an edge of $p$.
We then say the edge $e$ is \definitive{preserved} by the reduction $p\totom q$.
If a preserved edge $e$ is a cut, an axiom, a $\otimes$-edge or a $\parr$-edge,
then $e$ has the same endpoints in $p$ and in $q$: $t\sim^{p,I}_e t'$ iff $t\sim^{q,J}_e t'$.
If $e=\w$ is a jump, one endpoint might be changed:
indeed, $\w\sim^p_e \jump_p(\w)$ and $\w\sim^q_e \jump_q(\w)$,
and we might have $\jump_p(\w)\not=\jump_q(\w)$ when $\jump_p(\w)$
is part of an eliminated cut.
In this case, we say $e$ is a \definitive{redirected jump}.
We say $u\in\ST(p)$ is an \definitive{anchor} of $v\in\ST(q)$, 
if either $u=v$ or $u$ is involved in an eliminated cut $c=\cut{u}{u'}$,
and either $u=\otimes(\vec t)$ and $v\in\vec t$
or $u=\parr(\vec s)$ and $v\in\vec s$.

\begin{lem}
	\label{lem:nearcut}
	Assume $(d)\chi(d')\in\paths(q)$
	and $d$ and $d'$ are residuals.
	Then $\chi$ is non empty, 
	and its first and last edges are preserved.
\end{lem}
\begin{proof}
	This is a direct consequence of the fact 
	that if $(c)(e)$ is a path and $c$ is a cut
	then $e$ is not a cut.
\end{proof}

Observe that even if a path $\chi\in\paths(q)$ crosses preserved edges only, it
is not sufficient to have $\chi\in\paths(p)$, because the endpoints of
redirected jumps might change.
We say $\chi$ is a \definitive{preserved path} if $\chi$ crosses preserved edges only,
and we can write either 
$\chi=\chi'$ 
or
$\chi=(t\sim_{\w}\counit{\w})\chi'$ 
or
$\chi'(\counit{\w'}\sim_{\w'} t')$ 
or
$\chi=(t\sim_{\w}\counit{\w})\chi'(\counit{\w'}\sim_{\w'} t')$ 
where $\chi'$ crosses no redirected jump.

\begin{lem}
	\label{lem:preserved}
	Any non empty preserved path $\chi:t\leadsto_q t'$
	induces a unique preserved path $\chi^-:s\leadsto_p s'$
	with the same sequence of edges: 
	in particular, $\chi^-\in\paths(p,I)$ as soon 
	as $\chi\in\paths(q,J)$ and $I$ is an extension of $J$;
	and $s$ (resp. $s'$) is an anchor of $t$ (resp. $t'$').
	Moreover, if $\chi_1\chi_2$ is a preserved path
	then $(\chi_1\chi_2)^-=\chi^-_1\chi^-_2\in\paths(p)$.
\end{lem}
\begin{proof}
	The first part is a direct consequence of the definition.
	If moreover $\chi_1\chi_2$ is a preserved path,
	and neither $\chi_1$ nor $\chi_2$ is an empty path,
	then neither the last edge of $\chi_1$ 
	nor the first edge of $\chi_2$ can be a 
	redirected jump.
\end{proof}

By convention, if $\chi$ is empty, we set $\chi^-=\chi$.
We say $\chi^-$ is the path of $p$ \definitive{generated by} $\chi$.
In the next two subsections, we extend the generation of paths in $p$
from paths in $q$, first to paths without slipknots, then to arbitrary paths.

% We say $e$ is an \definitive{inner edge} of $\chi$ if we can write
% $\chi=(e_1)\chi'(e_2)$ so that $e$ is an edge of $\chi'$
% (in particular $\ln(\chi)>2$).
% Let $\chi:u\leadsto u'$ be a path:
% \begin{itemize}
% 	\item we say $\chi$ is \definitive{bound to $c$} if each edge of $\chi$ is bound to $c$;
% 	\item we say $\chi$ \definitive{visits} the tree $t$ if $t$ occurs as a node of $\chi$,
% 		\emph{i.e.} $\chi$ has a subpath $u\leadsto t$;
% 	\item we say $\chi$ \definitive{visits} the cut $\cut ts$ if $\chi$ visits $t$ or $s$.
% \end{itemize}

\subsection{Bridges and straight paths}
Let $c=\cut{t_0}{s_0}\in \cuts(p)$ with $t_0=\otimes(\vec t)$ and $s_0=\parr(\vec s)$.
We say an edge $e$ is \definitive{bound to $c$} if either $e=c$,
or $e=(t_0,t)$ with $t\in\vec t$,
or $e=s_0$.
And we say $\chi\in\paths(p)$ is bound to $c$ if all 
the edges crossed by $\chi$ are bound to $c$.
Observe that $e\in\paths(p)$ is a preserved edge
iff it is not bound to an eliminated cut.

\begin{lem}
	\label{lem:bound}
	Let $c=\cut{t_0}{s_0}$:
	\begin{itemize}
		\item if $\chi\in\paths(p)$ is bound to $c=\cut{t_0}{s_0}$ and
			$s_0=\parr(\vec s)$ then there exists $s\in\vec s$
			such that $\chi\in\paths(p,I)$ whenever $I(s_0)=s$;
		\item if $\chi\in\paths(p,I)$ does not cross any
			edge bound to $c$ then $\chi\in\paths(p,I')$
			whenever $I$ and $I'$ differ only on $s_0$.
	\end{itemize}
\end{lem}
\begin{proof}
	It is sufficient to observe that the edge $s_0$ is bound to $c$;
	and the only edge $e\in\parrs(p)$ that may be visited 
	by a path bound to $c$ is $s_0$.
\end{proof}

A \definitive{$c$-bridge} is a path $\chi$ that 
is bound to $c$ and that crosses $c$.
Observe that $\chi$ is a $c$-bridge iff either $\chi$ or the reverse path $\reverse \chi$
is a subpath of some $t\sim_{(t_0,t)}t_0\sim_c s_0\sim_{s_0} s$ with $t\in\vec
t$ and $s\in\vec s$.
Moreover, given $t\in \{ t_0\}\cup \vec t$ and $s\in \{s_0\}\cup \vec s$ there is a unique
$c$-bridge $t\leadsto s$.

\begin{lem}
	\label{lem:bridge}
	Assume $\chi_1\xi\chi_2\in\paths(q)$ and $\xi=t\sim_{\cut ts} s$,
	where $\chi_1$ and $\chi_2$ are preserved paths,
	$t\in\vec t$ and $s\in\vec s$.
	Then there exists a $c$-bridge $\bridge\xi$ such that
	$\chi_1^-\bridge\xi\chi_2^-\in\paths(p)$.
\end{lem}
\begin{proof}
	Write $\chi_1:v_1\leadsto_{q,J} t$ and $\chi_2:s\leadsto_{q,J} v_2$:
	by Lemma~\ref{lem:preserved},
	we obtain
	$\chi_1^-:u_1\leadsto_{p,I} t'$ 
	and
	$\chi_2^-:s'\leadsto_{p,I} u_2$
	where $u_1$, $t'$, $s'$ and $u_2$ are anchors of 
	$v_1$, $t$, $s$ and $v_2$ respectively,
	for any extension $I$ of $J$.
	In particular, $t'\in\{t,t_0\}$ and $s'\in\{s,s_0\}$
	and we can fix $\bridge\xi:t'\leadsto s'$ to be
	the only $c$-bridge with those endpoints.
	Observe indeed that $\bridge\xi\in\paths(p,I)$ as soon as $I(s_0)=s$.
	Then by Lemmas \ref{lem:preserved} and \ref{lem:bound},
	we can concatenate $\chi_1^-\bridge\xi\chi_2^-:u_1\leadsto_{p,I} u_2$.
\end{proof}
Despite the notation, the definition of $\bridge\xi$ does depend
on $\chi_1$ and $\chi_2$: whenever we use Lemma~\ref{lem:bridge}, however,
the values of $\chi_1$ and $\chi_2$ should be clear from the context.

We say a path $\chi\in\paths(q)$ is a \definitive{straight path} 
if it has no slipknot. Such a path is essentially a path of $p$,
up to replacing residuals with bridges:
\begin{lem}
	\label{lem:straight}
	If $\chi$ is a straight path,
	there exists a unique sequence of pairwise distinct
	eliminated cuts
	$c_1,\dotsc,c_n\in\cuts(p)\setminus\cuts(q)$,
	such that we can write
	$\chi=\chi_1(d_1)\cdots\chi_n(d_n)\chi_{n+1}$
	where each $d_i$ is a residual of $c_i$ 
	and $\chi$ crosses no other residual.
	Moreover we can form
	$\chi^-=\chi^-_1\bridge{(d_1)}\cdots\chi^-_n\bridge{(d_n)}\chi_{n+1}^-\in\paths(p)$.
\end{lem}
\begin{proof}
	The first part is straightforward reformulation of the absence of slipknots.
	The second part follows by applying Lemma~\ref{lem:bridge}
	to each $\chi_i(d_i)\chi_{i+1}$ (or the reverse path):
	the concatenation of preserved paths and bridges 
	in the definition of $\chi^-$  is allowed by Lemmas~\ref{lem:preserved}
	and \ref{lem:bound}.
\end{proof}

We say two paths $\chi_1,\chi_2\in\paths(q)$
are \definitive{independent} if they are disjoint
and there is no eliminated cut $c$
such that both $\chi_1$ and $\chi_2$ cross 
a residual of $c$.

\begin{lem}
	\label{lem:compatstraight}
	Assume $\chi_1,\dotsc,\chi_n\in\paths(q,J)$ are
	pairwise independent straight paths
	and $c_1,\dotsc,c_k\in\cuts(p)\setminus\cuts(q)$
	are such that no $\chi_i$ crosses a residual of $c_j$,
	for $1\le i\le n$ and $1\le j\le k$.
	Then for any $\xi_1,\dotsc,\xi_k\in\paths(p)$
	such that $\xi_i$ is bound to $c$ for $1\le i\le k$,
	there exists an extension $I$ of $J$ such 
	that $\chi_i^-\in\paths(p,I)$ for $1\le i \le n$
	and $\xi_j\in\paths(p,I)$ for $1\le j \le k$.
\end{lem}
\begin{proof}
	It is sufficient to observe that if $\chi^-$
	crosses an edge bound to an eliminated cut $c$
	then $\chi$ crosses a residual of $c$.
	Then the result is a direct consequence of the definition of $\chi^-_i$,
	together with Lemmas~\ref{lem:preserved} and \ref{lem:bound}.
\end{proof}

The generation of a path is thus compatible with the concatenation
of independent straight paths:
\begin{lem}
	\label{lem:concatstraight}
	If $\chi_1:v\leadsto_q u$ and $\chi_2:u\leadsto_q t$ are independent straight paths 
	then $\chi_1\chi_2$ is a straight path and 
	$(\chi_1\chi_2)^-=\chi_1^-\chi_2^-$.
\end{lem}
\begin{proof}
	That $\chi_1\chi_2$ is a straight path follows directly from the hypotheses.
	Write $\chi_i=\chi^i_1(d^i_1)\cdots\chi^i_{n_i}(d^i_{n_i})\chi^i_{n_i+1}$:
	it is then sufficient to apply the definition of $\chi^-_i$ and observe that
	$(\chi^1_{n_1+1}\chi^2_1)^- = (\chi^1_{n_1+1})^-(\chi^2_1)^-$ by Lemma~\ref{lem:preserved}.
\end{proof}

\subsection{Bounces and slipknots}
Let $c=\cut{t_0}{s_0}\in \cuts(p)$ with $t_0=\otimes(\vec t)$.
A \definitive{$c$-bounce} is a path $\chi$ that 
is bound to $c$, that does not cross $c$ and that visits $t_0$:
$\chi$ is either the empty path $\epsilon_{t_0}$,
or $t_0\sim_{t_0,t} t$ or $t\sim_{t_0,t} t_0$ with $t\in\vec t$,
or $t\sim_{(t_0,t)}t_0\sim_{(t_0,t')} t'$ with $t\not=t'\in\vec t$.
Given $t,t'\in \{t_0\}\cup\vec t$, such that either $t=t_0$ or $t'=t_0$ or $t\not= t'$,
there is a unique $c$-bounce $t\leadsto t'$.

\begin{lem}
	\label{lem:bounce}
	Assume $\chi_1\xi\chi_2\in\paths(q)$ and
	$\xi=(t_1\sim_{\cut{t_1}{s_1}} s_1)\xi'(s_2\sim_{\cut{t_2}{s_2}} t_2)$,
	where $\chi_1$ and $\chi_2$ are preserved paths,
	$t_1,t_2\in\vec t$ and $s_1,s_2\in\vec s$.
	Then there exists a $c$-bounce $\bounce\xi$
	such that $\chi_1^-\bounce\xi\chi_2^-\in\paths(p)$.
\end{lem}
\begin{proof}
	Necessarily, $\cut{t_1}{s_1}\not=\cut{t_2}{s_2}$, hence $t_1\not=t_2$.
	Write $\chi_1:v_1\leadsto_{q} t_1$ and $\chi_2:t_2\leadsto_{q} v_2$:
	we obtain
	$\chi_1^-:u_1\leadsto_{p} t'_1$ 
	and
	$\chi_2^-:t'_2\leadsto_{p} u_2$
	where $u_1$, $t'_1$, $t'_2$ and $u_2$ are anchors of 
	$v_1$, $t_1$, $t_2$ and $v_2$ respectively,
	for any extension $I$ of $J$.
	In particular, $t'_1\in\{t_1,t_0\}$ and $t'_2\in\{t_2,t_0\}$ with $t_1\not=t_2$,
	and we can fix $\bounce\xi:t'_1\leadsto t'_2$ to be the only 
	the only $c$-bounce with those endpoints.
	Then we can concatenate $\chi_1^-\bounce\xi\chi_2^-:u_1\leadsto_p  u_2$
	by Lemmas \ref{lem:preserved} and \ref{lem:bound}.
\end{proof}
Again, the definition of $\bounce\xi$ does depend on $\chi_1$ and $\chi_2$
but these should be clear from the context when we use Lemma~\ref{lem:bounce}.

We are now ready to prove that paths in $q$ are alternations of 
straight paths and slipknots, and generate paths in $p$
by replacing slipknots with bounces:
\begin{thm}
	\label{thm:paths}
	For each path $\chi\in\paths(q)$,
	there exists a unique sequence of pairwise distinct
	eliminated cuts 
	$c_1,\dotsc,c_n\in\cuts(p)\setminus\cuts(q)$,
	such than we can write
	$\chi=\chi_1\xi_1\cdots\chi_n\xi_n\chi_{n+1}$
	where:
	\begin{itemize}
		\item each $\chi_i$ is a straight path that
			crosses no residual of $c_j$ for $1\le j\le n$;
		\item $\chi_i$ and $\chi_j$ are independent when $i\not=j$;
		\item each $\xi_i$ is a slipknot $(d_i)\xi'_i(d'_i):t_i\leadsto t'_i$
			where $d_i$ and $d'_i$ are residuals of $c_i$,
			and $t_i$ and $t'_i$ are distinct premises of the $\otimes$-tree of $c_i$.
	\end{itemize}
	Moreover $\chi^-=\chi^-_1\bounce{\xi_1}\cdots\chi^-_n\bounce{\xi}_n\chi^-_{n+1}\in\paths(p)$.
\end{thm}

\begin{figure}[t]
	\begin{center}
	\begin{tikzpicture}[
			baseline=(cdots.base),
		]

		\node[lien](tensbas) at (0,0){$\otimes$};
		\node[lien](parbas) at (1,0){$\parr$};
		\node[nomcoupure] at(0.5,-0.8){$c_n$};

		\node[lien](tenshaut) at (1,4){$\otimes$};
		\node[lien](parhaut) at (2,4){$\parr$};
		\node[nomcoupure] at(1.5,3.2){$c_1$};
		%Bas:
		\draw[dotted](-1,0.5)--(-0.6,0.5);%sortie en bas
		\draw[reseau](-0.5,0.5)--(-0.3,0.5)--(tensbas)--(0,-0.5)--(1,-0.5)--(parbas)--(1,1)--(2,1)--(2,0.8);%tenseur+coupure+axhaut
		\draw[reseau](parbas)--(1.5,0)--(1.5,0.3)--(1.8,0.3)--(1.8,-0.2)--(2,-0.2);%axbas
		\draw[dotted] (2,0)--(2,0.6);%chemin dans parbas
		%Lien entre haut et bas:
		\draw[reseau](tensbas)--(0.5,0.3)--(0.5,0.5);
		\draw[dotted](0.5,0.5)--(0.5,4);
		\draw[reseau](0.5,4.3)--(0.5,4.5)--(0.8,4.5)--(tenshaut);
		%haut:
		\draw[reseau](tenshaut)--(1,3.5)--(2,3.5)--(parhaut)--(2,5)--(3,5)--(3,4.8);%tenseur+coupure+axhaut
		\draw[reseau](parhaut)--(2.5,4)--(2.5,4.3)--(2.8,4.3)--(2.8,3.8)--(3,3.8);%axbas
		\draw[dotted](3,4)--(3,4.6);%chemin dans parhaut

		\draw[reseau](tenshaut)--(1.3,4.3)--(1.3,4.5);
		\draw[dotted](1.3,4.7)--(1.3,5);%entrée en haut

		\node[nomchemin]at(1.1,5.2){$\chi_1^-$};

		\node[nomchemin]at(-0.8,0.8){$\chi_{n+1}^-$};

	\end{tikzpicture}\qquad reduces to \qquad
	\begin{tikzpicture}[baseline=(cdots.base)]

			%Bas:
		\draw[dotted](-1,0.5)--(-0.6,0.5);%sortie en bas
		\draw[reseau](-0.5,0.5)--(-0.3,0.5)--(-0.3,-0.5)--(1.5,-0.5)--(1.5,0)--(1.5,0.3)--(1.8,0.3)--(1.8,-0.2)--(2,-0.2);%coupurebassebasse
		\draw[reseau](0.5,0.5)--(0.5,0)--(1,0)--(1,1)--(2,1)--(2,0.8); %coupurebassehaute
		\draw[dotted] (2,0)--(2,0.6);%chemin dans noeud bas 
		%Lien entre haut et bas:
		\draw[dotted](0.5,0.5)--(0.5,4);
		%\draw[reseau](0.5,4.3)--(0.5,4.5)--(0.8,4.5)--(tenshaut);
		%haut:
		\draw[dotted](3,4)--(3,4.6);%chemin dans noeud haut

		\draw[reseau](0.5,4.3)--(0.5,4.5)--(1,4.5)--(1,3.5)--(2.5,3.5)--(2.5,4)--(2.5,4.3)--(2.8,4.3)--(2.8,3.8)--(3,3.8);%coupurehautebasse
		\draw[reseau](1.5,4.5)--(1.5,4)--(2,4)--(2,5)--(3,5)--(3,4.8);%coupurehautehaute

		\draw[dotted](1.5,4.7)--(1.5,5);%entrée en haut

		\node[nomcoupure]at(1.8,4.3){$d_1$};
		\node[nomcoupure]at(1.8,3.2){$d_1'$};

		\node[slipknot]at(3.3,4.3){$\xi_1'$};
		\node[nomchemin]at(1.2,5.2){$\chi_1$};

		\node[nomcoupure]at(0.8,0.3){$d_n$};
		\node[nomcoupure]at(0.8,-0.8){$d_n'$};

		\node[slipknot]at(2.3,0.3){$\xi_n'$};
		\node[nomchemin]at(-0.8,0.8){$\chi_{n+1}$};
	\end{tikzpicture}

\end{center}
\caption{Schematic shape of slipknots on a path (axiom and cut nodes
ommited)}
\label{fig:paths_slipknots}
\end{figure}

Figure~\ref{fig:paths_slipknots} illustrates the relationship between 
$\chi^-$ and $\chi$, in the simple case where each cut is between binary
connectives, and no redirected jump is involved:
each $\chi_i^-$ bounces on the $\otimes$ side of $c_i$ and joins $\chi_{i+1}^-$
directly instead of crossing the cut.

The proof of Theorem~\ref{thm:paths} is by induction on the length of $\chi$.
We break it down into a series of intermediate results.
Formally, given $\chi\in\paths(q)$:
\begin{itemize}
	\item we first establish Lemma~\ref{lem:compat} for all paths 
		such that Theorem~\ref{thm:paths} holds;
	\item we deduce Lemma~\ref{lem:loop} for $\chi$
		from Theorem~\ref{thm:paths} and Lemma~\ref{lem:compat} applied to $\chi$;
	\item Lemma~\ref{lem:slipknot} for $\chi$ is a direct consequence 
		of Lemma~\ref{lem:loop} applied to strict subpaths of $\chi$;
	\item we prove Lemma~\ref{lem:bracketing} for $\chi$ by applying 
		Theorem~\ref{thm:paths}, Lemma~\ref{lem:compat} and Lemma~\ref{lem:slipknot}
		to strict subpaths of $\chi$;
	\item then we prove Theorem~\ref{thm:paths} for $\chi$ by applying
		Lemmas~\ref{lem:slipknot} and \ref{lem:bracketing} to $\chi$,
		and Lemma~\ref{lem:compat} to strict subpaths of $\chi$.
\end{itemize}

\begin{lem}
	\label{lem:compat}
	Assume $\chi_1,\dotsc,\chi_n\in\paths(q,J)$ are pairwise disjoint paths,
	and $c_1,\dotsc,c_k\in\cuts(p)\setminus\cuts(q)$,
	are such that no $\chi_i$ crosses a residual of $c_j$,
	for $1\le i\le n$ and $1\le j\le k$.
	Then for any $\xi_1,\dotsc,\xi_k\in\paths(p)$
	such that $\xi_i$ is bound to $c$ for $1\le i\le k$,
	there exists an extension $I$ of $J$ such 
	that $\chi_i^-\in\paths(p,I)$ for $1\le i \le n$
	and $\xi_j\in\paths(p,I)$ for $1\le j \le k$.
\end{lem}
\begin{proof}
	This is a direct consequence of the definition of $\chi^-_i$
	in Theorem~\ref{thm:paths},
	together with Lemma~\ref{lem:compatstraight}.
\end{proof}

\begin{lem}
	\label{lem:loop}
	Assume $d=\cut{t}{s}$ and $d'=\cut{t'}{s'}$ are distinct residuals of the same cut
	$c=\cut{t_0}{s_0}$ with $t_0=\otimes(\vec t)$, $s_0=\parr(\vec s)$, $t,t'\in\vec t$ and $s,s'\in\vec s$.
	If $u\in d$, $u'\in d'$ and $\chi:u\leadsto_q u'$ crosses no residual of $c$ then $u=s$ and $u'=s'$.
\end{lem}
\begin{proof}
	Write $\chi=\chi_1\xi_1\cdots\chi_n\xi_n\chi_{n+1}$
	as in Theorem~\ref{thm:paths},
	and $\chi^-:v\leadsto_{p,I} v'$ 
	where $v$ (resp. $v'$) is an anchor of $u$ (resp. $u'$).
	Observe that $\chi$ is non empty because $u\not=u'$.
	Moreover, the first edge of $\chi$ is a preserved edge:
	if $\chi=(u\sim_e v)\chi'$ then $e$ cannot be a cut,
	as otherwise we would have $e=d$ and $\chi$ would cross a residual of $c$.
	It follows that $\chi_1^-$ is non empty, hence $\chi^-$ is non empty.
	Since $\chi^-$ is non empty, there is no path
	$\zeta:v\leadsto v'$ bound to $c$: otherwise,
	by Lemma~\ref{lem:compat}, we could form a non empty cycle
	$\reverse\zeta\chi^-:v'\leadsto_p v'$ .

	If $u=t$ and $u'=t'$, we have $v\in\{t,t_0\}$ 
	and $v'\in\{t',t_0\}$ with $t\not=t'$,
	and we obtain a $c$-bounce $v\leadsto v'$,
	hence a contradiction.
	If $u=t$ and $u'=s'$, we have $v\in\{t,t_0\}$ 
	and $v'\in\{s',s_0\}$, and we obtain a $c$-bridge
	$v\leadsto v'$, hence a contradiction.
	We rule out the case $u=s$ and $u'=t'$ symmetrically.
\end{proof}

\begin{lem}
	\label{lem:slipknot}
	If $\chi\in\paths(q)$ and
	$c=\cut{\otimes(\vec t)}{\parr(\vec s)}\in\cuts(p)\setminus\cuts(q)$,
	then $\chi$ crosses at most two residuals of $c$,
	and in this case we can write 
	$\chi=\chi_1(t\sim_{\cut ts} s)\chi_2(s'\sim_{\cut{t'}{s'}}t')\chi_3$
	with $t,t'\in\vec t$ and $s,s'\in\vec s$.
\end{lem}
\begin{proof}
	If $\chi=\chi_1(d)\chi_2(d')\chi_3$, where
	$d=\cut{t}{s}$ and $d'=\cut{t'}{s'}$ are residuals of $c$
	with $t,t'\in\vec t$ and $s,s'\in\vec s$,
	and if $\chi_2$ crosses no residual of $c$, then 
	by Lemma~\ref{lem:loop} applied to $\chi_2$, we obtain
	$\chi=\chi_1(t\sim_{\cut ts} s)\chi_2(s'\sim_{\cut{t'}{s'}}t')\chi_3$.
	If moreover $\chi_1$ (resp. $\chi_3$) crossed another residual 
	of $c$,  we would obtain a contradiction by applying 
	Lemma~\ref{lem:loop} to a strict subpath of $\chi_1$ (resp. $\chi_3$)
	hence of $\chi$.
\end{proof}

\begin{lem}\label{lem:bracketing}
	Slipknots are well-bracketed in the following sense:
	there is no path $\chi=(d_1)\chi_1(d_2)\chi_2(d'_1)\chi_3(d'_2)\in\paths(q)$
	such that, for $1\le i\le 2$, $d_i$ and $d'_i$ are residuals of the same cut.
\end{lem}
\begin{proof}
	Assume $\chi=(d_1)\chi_1(d_2)\chi_2(d'_1)\chi_3(d'_2)\in\paths(q)$
	such that, for $1\le i\le 2$, $d_i$ and $d'_i$ are residuals of the same cut.
	We can assume w.l.o.g. that $\chi_1$ and $\chi_3$ are independent:
	otherwise there is a prefix of $\chi$ with this additional property.

	For $1\le i\le 2$,
	let $c_i=\cut{t_0^i}{s_0^i}$, with $t^i_0=\otimes(\vec t_i)$ and $s^i_0=\parr(\vec s_i)$,
	and assume $d_i=\{t_i,s_i\}$ and $d'_i=\{t'_i,s'_i\}$, 
	with $t_i,t'_i\in\vec t_i$ and $s_i,s'_i\in\vec s_i$.
	By Lemma~\ref{lem:slipknot}, we must have
	$\chi=(t_1\sim_{d_1}s_1)\chi_1(t_2\sim_{d_2}s_2)\chi_2(s'_1\sim_{d'_1}t'_1)\chi_3(s'_2\sim_{d'_2}t'_2)$.

	By Theorem~\ref{thm:paths},
	we obtain $\chi_1^-:u_1\leadsto_p v_2$ and 
	$\chi_3^-:v_1\leadsto_p u_2$, where
	$u_1\in\{s_1,s_0^1\}$, 
	$v_2\in\{t_2,t_0^2\}$,
	$v_1\in\{t'_1,t_0^1\}$
	and $u_2\in\{s'_2,s_0^2\}$.
	Write $\zeta_1:v_1\leadsto_p u_1$
	(resp. $\zeta_2:v_2\leadsto_p u_2$)
	for the only $c_1$-bridge (resp. $c_2$-bridge) with those endpoints.
	Then we obtain a non empty cycle
	$\chi^-_1\zeta_2(\reverse{\chi_3})^-\zeta_1:u_1\leadsto_p u_1$:
	the concatenation is allowed by Lemma~\ref{lem:compat}
	applied to $\chi_1$ and $\reverse{\chi_3}$.
\end{proof}

\begin{proof}[Proof of Theorem~\ref{thm:paths}]
	By Lemma~\ref{lem:bracketing}, we can write 
	$\chi=\chi_1\xi_1\cdots\chi_n\xi_n\chi_{n+1}$
	where $\xi_1,\dotsc,\xi_n$ are the slipknots of $\chi$
	that are maximal (\emph{i.e.} not strict subpaths of other slipknots)
	and this writing is unique.
	Let $c_1,\dotsc,c_n$ be the associated eliminated cuts.
	By Lemma~\ref{lem:slipknot}, the $c_j$'s are pairwise distinct,
	and $\xi_j:t_j\leadsto_p t'_j$ where
	$t_j$ and $t'_j$ are distinct premises of the $\otimes$-tree of $c_j$.
	Moreover, since each slipknot of $\chi$ is a subpath 
	of some $\xi_j$:
	\begin{itemize}
		\item each $\chi_i$ is a straight path and it crosses no residual of any $c_j$;
		\item the $\chi_i$'s are pairwise independent.
	\end{itemize}
	If $n=0$, we can set $\chi^-=\chi_1^-$. 
	Otherwise, we apply Lemma~\ref{lem:compat} to the $\chi_i$'s,
	which are strict subpaths of $\chi$, which allows to concatenate
	$\chi^-=\chi^-_1\bounce{\xi_1}\cdots\chi^-_n\bounce{\xi}_n\chi^-_{n+1}\in\paths(p)$.
\end{proof}

Again, the construction of $\chi^-$ is compatible with 
the concatenation of independent paths:
\begin{lem}
	\label{lem:concat}
	If $\chi_1\chi_2\in\paths(q)$ and $\chi_1$ and $\chi_2$ are independent,
	then $(\chi_1\chi_2)^-=\chi_1^-\chi_2^-$.
\end{lem}
\begin{proof}
	As for Lemma~\ref{lem:concatstraight},
	this is a direct consequence of the definition of $(\chi_1\chi_2)^-$, $\chi_1^-$ and $\chi_2^-$,
	this time using Lemma~\ref{lem:concatstraight} to concatenate a
	straight suffix of $\chi_1$ and a straight prefix of $\chi_2$.
\end{proof}

\subsection{Bounding the growth of $\ln$}
\label{subsection:ln}

Now we show that we can bound $\ln(q)$ depending only on $\ln(p)$.
We first need some basic properties relating the length of $\chi^-$ 
with that of $\chi$.
\begin{lem}
	\label{lem:length}
	Let $\chi,\xi\in\paths(q)$ and $\zeta\in\paths(p)$:
	\begin{enumerate}
		\item if $\chi$ is preserved then $\ln(\chi^-)=\ln(\chi)$;
		\item if $\zeta$ is a bridge then $1\le\ln(\zeta)\le 3$;
		\item if $\zeta$ is a bounce then $\ln(\zeta)\le 2$;
		\item if $\chi$ is straight, then $\ln(\chi)\le \ln(\chi^-)\le 3\ln(\chi)$;
		\item if $\xi$ is a prefix of $\chi$ and $\chi$ is straight then 
			$\xi^-$ is a prefix of $\chi^-$;
		\item if $\xi$ is a slipknot then $\ln(\xi)\ge 3$
			and $\ln(\bounce\xi)<\ln(\xi)$;
		\item in general $\ln(\chi^-)\le 3\ln(\chi)$.
	\end{enumerate}
\end{lem}
\begin{proof}
	The first three properties are direct consequences of the definitions.
	Item (4) follows from (1) and (2).
	Item (5) follows from Lemma~\ref{lem:concatstraight}.
	Item (6) follows from Lemma~\ref{lem:nearcut} and (3).
	And item (7) follows from (4) and (6).
\end{proof}
Observe that in general, we do not have $\ln(\zeta^-)\le\ln(\chi^-)$
when $\zeta$ is a prefix of $\chi$: $\zeta$ may enter an
arbitrarily long slipknot of $\chi$ that is replaced by a single
bounce in $\chi^-$.
For this reason, we introduce the following notion:
if $\chi\in \paths(q)$, we define the \definitive{width} of $\chi$ 
(relative to the reduction $p\totom q$ we consider) by
$\width(\chi)=\max\{\ln(\zeta^-)\tq \zeta \text{ prefix of } \chi\}$.
\begin{lem}
	\label{lem:width}
	For any path $\chi\in\paths(q)$,
	$\ln(\chi^-)\le \width(\chi)\le\ln(p)$
	and
	$\width(\chi)\le 3\ln(\chi)$.
	Moreover, if $\zeta$ is a prefix of $\chi$,
	we have $\width(\zeta)\le\width(\chi)$.
	If moreover $\chi$ is straight,
	$\width(\chi)=\ln(\chi^-)\ge\ln(\chi)$.
\end{lem}
\begin{proof}
	We obtain $\ln(\chi^-)\le \width(\chi)\le\ln(p)$ and 
	$\width(\zeta)\le\width(\chi)$ directly from the definition of width.
	Item (7) of Lemma~\ref{lem:length} gives $\width(\chi)\le 3\ln(\chi)$.
	If $\chi$ is straight $\width(\chi)=\ln(\chi^-)\ge\ln(\chi)$ 
	follows from items (4) and (5) of Lemma~\ref{lem:length}.
\end{proof}

Define $\varphi:\N\to\N$ inductively by $\varphi(0)=0$ and
$\varphi(n)=n+(n+1)(\varphi(n-1)+2)$ if $n>0$. 
Observe that $n\le \varphi(n)\le\varphi(n+1)$.

% \varphi must be increasing

\begin{lem}\label{lem:generatedlength}
	If $\chi\in\paths(q)$ then $\ln(\chi)\le\varphi(\width(\chi))$.
\end{lem}
\begin{proof}
	The proof is by induction on $\width(\chi)$. 
	If $\chi$ is straight then, by Lemma~\ref{lem:width},
	$\ln(\chi)\leq\width(\chi)\le\varphi(\width(\chi))$.

	Write $\chi=\chi_1\xi_1\cdots\chi_n\xi_n\chi_{n+1}$
	as in Theorem~\ref{thm:paths}: we have
	$\ln(\chi)=\sum_{i=1}^{n+1}\ln(\chi_i)+\sum_{j=1}^n\ln(\xi_j)$.
	Since $\chi^-=\chi_1^-\bounce{\xi_1}\cdots\chi_n^-\bounce{\xi_n}\chi_{n+1}^-$,
	we have $\sum_{i=1}^{n+1}\ln(\chi_i^-)\leq\ln(\chi^-)$.
	Since each $\chi_i$ is straight, 
	we obtain $\ln(\chi)\le\ln(\chi^-)+\sum_{j=1}^n\ln(\xi_j)$
	from the previous inequality, 
	by applying item (4) of Lemma~\ref{lem:length}.

	Moreover observe that, by Lemma~\ref{lem:nearcut},
	$\chi_i$ is non empty for $1<i<n+1$.
	Hence $\ln(\chi^-)\ge n-1$, and we obtain $n\le\width(\chi)+1$.

	It remains to bound $\ln(\xi_j)$ for $1\le j\le n$.
	We can write $\xi_j=(d_j)\chi'_j(d'_j)$
	where $d_j$ and $d'_j$ are the residuals of 
	the cut $c_j$ associated with $\xi_j$.
	Let $\zeta'_j$ be a prefix of $\chi'_j$
	and write $\zeta_j=\chi_1\xi_1\cdots\chi_j(d_j)\zeta'_j$
	which is a prefix of $\chi$.
	Observe that, by Theorem~\ref{thm:paths},
	$c_j$ has no residual in $\zeta_j$ other than $d_j$, and 
	$\chi_1\xi_1\cdots\chi_j(d_j)$ and $\zeta'_j$ are independent.
	Hence $\chi_j(d_j)$ is straight and 
	$\zeta_j^-=\chi_1^-\bounce{\xi_1}\cdots\chi_j^-\bridge{d_j}(\zeta'_j)^-$
	follows by Lemma~\ref{lem:concat}.
	Since $\ln(\bridge{d_j})\ge 1$, 
	we obtain $\ln((\zeta'_j)^-)\le \ln(\zeta_j^-)-1\le\width(\chi)-1$.

	Hence $\width(\chi'_j)\le\width(\chi)-1$:
	we apply the induction hypothesis and obtain
	$\ln(\chi'_j)\le\varphi(\width(\chi'_j))\le\varphi(\width(\chi)-1)$
	because $\varphi$ is monotonous.
	It follows that $\ln(\xi_j)\le \varphi(\width(\chi)-1)+2$,
	and we conclude:
	\begin{align*}
	\ln(\chi)
	&\le\ln(\chi^-)+\sum_{j=1}^n\ln(\xi_j)
	\\
	&\le\width(\chi)+(\width(\chi)+1)(\varphi(\width(\chi)-1)+2)
	\\
	&\le\varphi(\width(\chi)).
	\qedhere
	\end{align*}
\end{proof}

Using Lemma~\ref{lem:width} again, we obtain:
\begin{cor}\label{vc}
	Let $p\totom q$. Then, $\ln(q)\leq \varphi(\ln(p))$.
	%\footnote{
	%One can show that $\varphi(n)\sim 2e\times n!$,
	%In any case, we are only interested in the existence of a bound.
	%}
\end{cor} 

Notice that the previous result can be seen as a quantitative version
of the preservation of acyclicity in proof nets under reduction. In
the following example, we illustrate how acyclicity is mandatory for
the existence of a function $\varphi$ as in Corollary~\ref{vc}.
\begin{exa}\label{ex:cycle1}
Let $p=(\cut{t}{s};)$
with $t=\otimes(x_0,\dotsc,x_n)$ and $s=\parr(\bar x_1,\dotsc,\bar x_n,\bar x_0)$:
by setting $I(s)=\bar x_i$, we obtain a cycle
$x_i\sim_{\{x_i,\bar x_i\}}\bar x_i\sim_s s\sim_{\cut{t}{s}} t\sim_{t,x_i} x_i$.
Since each path $\chi\in\paths(p)$ can cross 
each of $s$ and $\cut{t}{s}$ at most once, 
it is easy to check that $\ln(p)=6$.

But $p\totom q=(\cut{x_0}{\bar x_1},\dotsc,\cut{x_{n-1}}{\bar x_n},\cut{x_{n}}{\bar x_0};)$
hence $\ln(q)=2(n+1)$.
The situation is illustrated in Figure~\ref{fig:cycle1}.
\end{exa}
\begin{figure}
\begin{center}
	\begin{tikzpicture}
		\node[lien](0)at(0,0){ax};
		\node[lien](1)at($(0)+(0,-1)$){ax};
		\node(2)at($(1)+(0,-.7)$){$\rvdots$};
		\node[lien](3)at($(2)+(0,-.7)$){ax};
		\node[lien](t)at($(3)+(-1,-1.5)$){$\otimes$};
		\node[lien](p)at($(3)+(1,-1.5)$){$\parr$};
		\node[lien](c)at($(0,-1)+1/2*($(t)+(p)$)$){cut};

		\draw[above left](0)node{$x_0\ $};
		\draw[above right](0)node{$\ \bar x_0$};
		\draw[above left](1)node{$x_1\ $};
		\draw[above right](1)node{$\ \bar x_1$};
		\draw[above left](3)node{$x_n\ $};
		\draw[above right](3)node{$\ \bar x_n$};

		\draw[par dessus](t)to[out=120,in=west](0);
		\draw[par dessus](t)to[out=100,in=west](1);
		\path (t) to [out=80,in=west] node [sloped, near start] {$\rvdots$} (2);
		\draw[par dessus](t)to[out=60,in=west](3);
		\draw[par dessus](p)to[out=125,in=east](1);
		\path (p) to [out=95,in=east] node [sloped, very near start] {$\rvdots$} (2);
		\draw[par dessus](p)to[out=60,in=east](3);
		\path[par dessus](p)to[out=50,in=east](0);
		\draw[par dessus](t)to[out=south,in=west](c);
		\draw[par dessus](p)to[out=south,in=east](c);

		\node[lien](0)at(6,0){ax};
		\node[lien](1)at($(0)+(0,-1)$){ax};
		\node (2)at($(1)+(0,-1)$){$\rvdots$};
		\node[lien](3)at($(2)+(0,-1)$){ax};

		\node[lien](c01)at($(3)+(0,-1.2)$){cut};
		\node[lien](c12)at($(c01)+(0,-1)$){cut};
		\node(c23)at($(c12)+(0,-1)$){$\rvdots$};
		\node[lien](c30)at($(c23)+(0,-1)$){cut};

		\draw[above left](0)node{$x_0\ $};
		\draw[above right](0)node{$\ \bar x_0$};
		\draw[above left](1)node{$x_1\ $};
		\draw[above right](1)node{$\ \bar x_1$};
		\draw[above left](3)node{$x_n\ $};
		\draw[above right](3)node{$\ \bar x_n$};

		\draw[par dessus](0)to[out=west,in=west](c01);
		\draw[par dessus](1)to[out=east,in=east](c01);
		\draw[par dessus](1)to[out=west,in=west](c12);
		\path (c12) to [out=east,in=east] node [very near end,inner sep=0pt,sloped] (c12a) {$\scriptstyle\cdots$}  node [very near end,inner sep=0pt] (c12b) {} (2);
		\draw[par dessus] (c12) to [out=east,in=-40] (c12a.east) ;
		\path (3) to [out=east,in=east] node [very near end,inner sep=0pt,sloped] (c23a) {$\scriptstyle\cdots$}  node [very near end,inner sep=0pt] (c23b) {} (c23);
		\draw[par dessus] (3) to [out=east,in=40] (c23a.east) ;
		\draw[par dessus](3)to[out=west,in=west](c30);
		\draw[par dessus](0)to[out=east,in=east](c30);
		
		\node at (3,-3.5){$\totom$};
	\end{tikzpicture}
\end{center}
\caption{A cyclic counterexample to Corollary~\ref{vc}}
\label{fig:cycle1}
\end{figure}

\subsection{Erratum}

In the extended abstract of the present paper, the
analogue~\cite[Subsection 3.2]{CV18} of Subsection~\ref{subsection:ln} claimed
to establish similar results for another measure on paths:
rather than the length $\ln(\xi)$ of a path $\xi$, we considered the number
$\vc(\xi)$ of all the cuts $\cut{t}{s}$ such that $\xi$ visits $t$ or $s$.

It is easy to check that Lemma~\ref{ax} still holds if we replace the
maximal length of a path with the maximum number of cuts in a chain of axiom 
cuts in the sense of Figure~\ref{fig:collapse}. 
Given the situation depicted in Figure~\ref{fig:paths_slipknots}, however, it
is evident that a bound on the number of cuts crossed by a path cannot be
preserved: the path on the left hand side crosses no cut, while the path in the
reduct crosses an arbitrary number of (possibly axiom) cuts.
We introduced $\vc(\xi)$ in our previous attempt,
precisely to capture this example: if $\xi$ is the path following the 
tensors  of the left hand side, then $\vc(\xi)\ge n$.
But this fix is actually not sufficient:
if we replace each $c_i=\cut{t_i}{s_i}$ in Figure~\ref{fig:paths_slipknots},
with $\cut{\otimes(t_i,s_i)}{\parr(x_i,\bar{x_i})}$,
a path in the obtained net can visit at most two of these new cuts,
but it reduces to the left hand side, with $\xi$ such that $\vc(\xi)\ge n$.

It might be possible to adapt our method for dealing with a relaxed definition
of visited cut: for instance, we might consider the number of cuts $c$ such
that $\xi$ visits a tree $t\in\ST(c)$ (instead of $t\in c$).
But this notion is no longer local, and would introduce further technicalities:
for that reason, we decided to focus on the length of paths instead, 
which is a more intuitive and standard notion,
without any \emph{ad hoc} reference to cuts.

\section{Variations of $\jd(p)$ under reduction}\label{section:qi}

For establishing that $\jd(q)$ is bounded as a function of
$\jd(p)$ and $\ln(p)$ we examine the reductions separately.

\begin{lem}\label{qi_growth_m}
	Let $p,q$ two proof nets. If $p\totom q$, then $\jd(q)\leq
	2\jd(p)$.
\end{lem}

\begin{proof}
	For all $t\in\ST(q)$ the only case in which $\jump^{-1}_q(t)\not=\jump^{-1}_p(t)$
	is that of redirected jumps:
	there must be $\w\in\counits(q)$ such that
	$t$ is part of a residual $\cut{t}{s}$ of an eliminated cut
	$\cut{t_0}{s_0}\in\cuts(p)$, with $\jump_p(\w)=t_0$.
	In this case we have $\jd_q(t)=\jd_p(t)+\jd_p(t_0)\le 2\jd(p)$.
\end{proof}

\begin{lem}\label{qi_growth_e}
	Let $p,q$ two proof nets. If $p\totoe q$, then $\jd(q)\leq
	(2\jd(p))^{\ln(p)+1}$.
\end{lem}

\begin{proof}
	Fix $t\in\ST(q)$. For any $\w\in\counits(q)$, if $\jump_q(\w)=t$,
	then $\w\in\I_p^n(t)$ for some $n\in\N$: this is precisely
	the purpose of the definition of $\I^n_p$.
	We obtain
	$\jd_q(t)\le \card\left(\bigcup_{n\in\N}\I_p^n(t)\right)\leq (2\jd(p))^{\ln(p)+1}$
	by Lemma~\ref{borne_qi} (1).
\end{proof}

\begin{lem}\label{qi_growth_a}
	Let $p,q$ two proof nets. If $p\toto_a q$, then $\jd(q)\leq
	(\ln(p)+1)\jd(p)$.
\end{lem}
\begin{proof}
	As in the proof of Lemma~\ref{ax}, we can write
	$p=(\vec{c_1},\dotsc,\vec{c_k},\vec c;\vec s)$
	and
	$q= (\vec c;\vec s)[t^1/\bar x^1_0]\cdots[t^k/\bar x^k_0]$
	where 
	$\vec c_i=(\cut{x^i_0}{\bar x^i_1},\dotsc,\cut{x^i_{n_i-1}}{\bar x^i_{n_i}},\cut{x^i_{n_i}}{t^i})$.

	By the definition of cut elimination, for all $\w\in\counits(q)$ 
	we have:
	\[
		\jump_q(\w)= \begin{cases}
			t^i [t^{i+1}/\bar x_0^{i+1}]\cdots[t^k/\bar x_0^k]
			&\text{if $\jump_p(\w)\in\{x^i_0,\bar x^i_0,x^i_1,\dots,\bar x^i_{n_i},t^i\}$}
			\\
			\jump_q(\w)[t^{1}/\bar x_0^{1}]\cdots[t^k/\bar x_0^k]
			&\text{otherwise}
		\end{cases}.
	\]
	It follows that:
	\begin{itemize}
		\item $\jd_q(t^i [t^{i+1}/\bar x_0^{i+1}]\cdots[t^k/\bar x_0^k])=\jd_p(t^i)+\sum_{j=0}^{n_i}(\jd_p(x^i_j)+\jd_p(\bar x^i_j))
			\le (2n_i+3)\jd(p)$;
		\item if $t\in\ST(p)\setminus\bigcup_{i=1}^k\{x^i_0,\bar x^i_0,x^i_1,\dots,\bar x^i_{n_i},t^i\}$,
			then $\jd_q(t[t^{1}/\bar x_0^{1}]\cdots[t^k/\bar x_0^k])=\jd_p(t)$.
	\end{itemize}
	To conclude, it is sufficient to observe that $2n_i+2\le \ln(p)$:
	indeed each $\vec c_i$ induces a path alternating between 
	$n_i+1$ axioms an $n_i+1$ cuts.
\end{proof}

\section{Bounding the size of antireducts: general and iterated case}\label{section:general}

The previous results now allow us to treat the general case of a reduction $p\toto q$.
\begin{thm}\label{size_bound}
	If $p\toto q$ then $\size(p)\leq
	\psi\left(2(\ln(p)+1)\size(q),\ln(p),\jd(p)\right)$.
\end{thm}
\begin{proof}
	Consider $q',q''$ such that $p\totoe q'\totoa q''\totom q$.
	We have:
	\begin{align*}
		\size(p)
		&\leq\psi(\size(q'),\jd(p),\ln(p))
		&&\text{(by Lemma~\ref{affs_size})}\\
		&\leq\psi((\ln(q')+1)\size(q''),\jd(p),\ln(p))
		&&\text{(by Lemma~\ref{ax})}\\
		&\leq\psi((\ln(p)+1)\size(q''),\jd(p),\ln(p))
		&&\text{(by Lemma~\ref{cc_ax_e})}\\
		&\leq\psi(2(\ln(p)+1)\size(q),\jd(p),\ln(p))
		&&\text{(by Lemma~\ref{mult})}\quad .
		\qedhere
	\end{align*}
\end{proof}

\begin{cor}
	\label{corollary:antireducts:onestep}
	If $q$ is an MLL net and $n,m\in\N$,
	then \[\{p\tq p\toto q,\ \jd(p)\le m\text{ and }\ln(p)\le n\}\] is finite.
\end{cor}

Of course, that result holds only up to the $\alpha$-equivalence mentioned in Remark~\ref{rem:alpha}:
then it is easy to check that the cardinality of $\{p\tq \size(p)\le k\}$ is bounded
by a function of $k\in\N$.
Also recall from Remark~\ref{rem:jumps} that we are actually interested in bare
nets rather than nets with jumps, so Corollary~\ref{corollary:antireducts:onestep}
should be read as follows:
given a bare net $q$ and $n,m\in\N$ there are finitely many bare nets $p$ such
that $p\toto q$ and that can be equipped with a jump function $\jump_p$
satisfying $\jd(p)\le m$ and $\ln(p)\le n$.
More precisely, Theorem~\ref{size_bound} entails that the number of such bare
nets $p$ can be bounded by a function of $m$, $n$ and $\size(q)$.

It follows that, given an infinite linear combination $\sum_{i\in I} a_i.p_i$,
assuming that we can equip each $p_i$ with a jump function $\jump_{p_i}$
so that $\{\ln(p_i)\tq i\in I\}\cup\{\jd(p_i)\tq i\in I\}$ is finite,
we can  always consider an arbitrary family of reductions $p_i\toto q_i$ for
$i\in I$ and form the sum $\sum_{i\in I} a_i.q_i$: this is always well defined.
But if we want to iterate this process and perform a reduction
from $\sum_{i\in I} a_i.q_i$ to $\sum_{i\in I} a_i.r_i$,
when $q_i\toto r_i$ for $i\in I$,
we need to ensure that a similar side condition holds for the $q_i$’s.
Again, this is a consequence of our previous results, which we sum up 
in the following two theorems.

\begin{thm}\label{vc_growth} Let $p\toto q$.
	Then $\ln(q)\le\varphi(\ln(p))$.
\end{thm}

\begin{proof}
	Consider $q',q''$ such that $p\totom q'\totoe q''\totoa q$.
	We have:
	\begin{align*}
		\ln(q)
		&\leq \ln(q'')
		&&\text{(by Lemma~\ref{cc_ax_e})}\\
		&\leq \ln(q')
		&&\text{(by Lemma~\ref{cc_ax_e})}\\
		&\leq \varphi(\ln(p))
		&&\text{(by Corollary~\ref{vc})}\quad .
		\qedhere
	\end{align*}
\end{proof}

\begin{thm}\label{qi_growth}
	There exists a function $\theta:\N\to\N$
	such that $\jd(q)\leq \theta(\ln(p),\jd(p))$
	whenever $p\toto q$.
\end{thm}

\begin{proof} 
	Consider $q',q''$ such that $p\totoa q'\totoe q'' \totom q$.
	We have
	\begin{align*}
		\jd(q)
		&\leq 2\jd(q'')
		&&\text{(by Lemma~\ref{qi_growth_m})}\\
		&\leq 2(2\jd(q'))^{\ln(q')+1}
		&&\text{(by Lemma~\ref{qi_growth_e})}\\
		&\leq 2(2(\ln(p)+1)\jd(p))^{\ln(q')+1}
		&&\text{(by Lemma~\ref{qi_growth_a})}\\
		&\leq 2(2(\ln(p)+1)\jd(p))^{\ln(p)+1}
		&&\text{(by Lemma~\ref{cc_ax_e})}\quad .
		\qedhere
	\end{align*}
\end{proof}

By the previous results, we can iterate Corollary~\ref{corollary:antireducts:onestep} and obtain:
\begin{cor}
	\label{corollary:antireducts}
	If $q$ is an MLL net and $k,n,m\in\N$,
	then \[\{p\tq p\toto^k q,\ \jd(p)\le m\text{ and }\ln(p)\le n\}\] is finite.
\end{cor}

\section{Taylor expansion}\label{section:Taylor}

We now show how the previous results apply to Taylor expansion.
For that purpose, we must extend our syntax to MELL proof nets.
Our presentation departs from Ehrhard’s \cite{Ehrhard16} in our treatment of
promotion boxes:
instead of introducing boxes as tree constructors labelled by nets, 
with auxiliary ports as inputs, we consider box ports as $0$-ary trees, 
that are related with each other in a \emph{box context},
associating each box with its contents.
This is in accordance with the usual presentation of promotion 
as a black box, and has two motivations:
\begin{itemize}
	\item in Ehrhard’s syntax, the promotion is not a net but an open
		tree, for which the trees associated with auxiliary ports
		must be mentioned explicitly: 
		this would complicate the expression of Taylor expansion;
	\item since we consider a single class of $\wn$-links instead of having a
		separate dereliction, we must impose constraints 
		on auxiliary ports, that are easier to express when these ports 
		are directly represented in the syntax.
\end{itemize}
Then we show that if $p$ is a resource net in the support of the Taylor
expansion of an MELL proof net $P$, then $\ln(p)$ and $\jd(p)$
are bounded by functions of $P$.

Observe that we need only to consider the support of Taylor expansion,
so we do not formalize the expansion of MELL nets into 
infinite linear combinations of resource nets:
rather, we introduce $\taylor(P)$ as a set of approximants.

\subsection{MELL nets}

In addition to the set of variables, 
we fix a denumerable set $\fp$ of \definitive{box ports}:
we assume given an enumeration $\fp=\{a^b_i\tq i,b\in\N\}$.
We call \definitive{principal ports} the ports $a^b_0$ and \definitive{auxiliary ports} the other ports.
Instead of separate contractions and derelictions, we consider a unified
$\wn$-link of arbitrary arity; auxiliary ports of boxes must be premises of such
links (or of auxiliary ports of outer boxes,
that must satisfy this constraint inductively).

The weakenings (and coweakenings, in the resource nets yet to be introduced)
are not essentially different from the multiplicative
units in our untyped nets. Indeed, we will see that the
geometrical and combinatorial behaviour of the $\wn$-link
(resp. the $\oc$-link) is identical 
to that of the $\parr$ (respectively,
of the $\otimes$). This will be reflected in our use of labels:
in addition to $\units$ and $\counits$, we will use labels from denumerable
sets $\coweaks$ and $\weaks$ (now assuming $\fv$, $\fp$, $\units$,
$\counits$, $\coweaks$ and $\weaks$ are pairwise disjoint),
and write $\pleaves=\units\cup\coweaks$ and $\nleaves=\counits\cup\weaks$.

We introduce the corresponding term syntax.
\definitive{Raw pre-trees} ($S^\circ$, $T^\circ$, \emph{etc.})
and \definitive{raw trees} ($S$, $T$, \emph{etc.})
are defined by mutual induction as follows:
\begin{align*}
	T&\recdef x \mid \unit\cw \mid \counit\w \mid \otimes(T_1,\dotsc,T_n)\mid \parr(T_1,\dotsc,T_n) \mid a^b_0 \mid \weak{\w'}  \mid \wn(T^\circ_1,\dotsc,T^\circ_n)
	\\
	T^\circ & \recdef T \mid a^b_{i+1}
\end{align*}
where $x$ ranges over $\fv$, $\cw$ ranges over $\units$, $\w$ ranges over $\counits$, 
$\w'$ ranges over $\weaks$, $b$ and $i$ range over $\N$ and we require $n\not=0$
in each of $\otimes(T_1,\dotsc,T_n)$, $\parr(T_1,\dotsc,T_n)$ and
$\wn(T^\circ_1,\dotsc,T^\circ_n)$.
The set $\SPT(S)$ of the sub-pre-trees of $S$ is defined in the natural way,
as well as the set $\ST(S)$ of sub-trees of $S$, from which 
we derive the definitions of $\fv(S)$, $\fp(S)$, $\counits(S)$, \emph{etc.}
The set $\atoms(S)$ of atoms of $S$ is then $\fv(S)\cup\fp(S)\cup\units(S)\cup\nleaves(S)$.

A \definitive{tree} (resp. a \definitive{pre-tree}) is a raw tree (resp. raw pre-tree) 
in which each atom occurs at most once.
A \definitive{cut} is an unordered pair of trees $C=\cut{T}{S}$
with disjoints sets of atoms.
Pre-trees and cuts only describe the surface level of MELL nets:
we also have to introduce promotion boxes.

We now define \definitive{box contexts} and \definitive{pre-nets} by mutual induction as follows.
A box context $\Theta$ is the data of a finite set $\boxes_{\Theta}\subset\N$,
and of a pre-net of the form $\Theta(b)=(\Theta_b;\vec C_b;T_b,\vec S^\circ_b;\jump_b)$,
for each $b\in\boxes_{\Theta}$.
We then write $\ar_\Theta(b)$, or simply $\ar(b)$ for the length of the family
$\vec S^\circ_b$, which we call the \definitive{arity} of the box $b$.
A pre-net is a tuple $P^\circ=(\Theta;\vec C;\vec S^\circ;\jump)$
where:
\begin{itemize}
	\item $\Theta$ is a box context;
	\item the \definitive{jump function} $\jump$ is a function
		$\nleaves(\vec C,\vec S^\circ)\to \ST(\vec C,\vec S^\circ)$;
	\item each atom occurs at most once in $\vec C,\vec S^\circ$;
	\item $a^b_i\in\fp(\vec C;\vec S^\circ)$ iff $b\in\boxes_\Theta$ and $0\le i\le\ar(b)$;
	\item $\fv(\vec C,\vec S^\circ)$ is closed under the involution $x\mapsto \bar x$.
\end{itemize}
Then a \definitive{net} is a pre-net
of the form $P=(\Theta;\vec C;\vec S;\jump)$,
\emph{i.e.} without auxiliary ports as conclusions.
In the following, we may write, e.g., $\Theta_{P}$ for $\Theta$ in
this case. An example is illustrated in Figure~\ref{fig:ex_mell}.

\begin{rem}\label{rem:mellnets}
	To be formal, in the definition of a box context $\Theta$,
	we should also fix an enumeration of the family $\vec S^\circ_b$ in $\Theta(b)$.
	Indeed, when we write $a_0^b,a_1^b,\dots,a_{\ar(b)}^b$ for the ports of a box $b$,
	and $\Theta(b)=(\Theta_b;\vec C_b;T_b,S^\circ_{b,1},\dots,S^\circ_{b,\ar(b)};\jump_b)$ for the 
	contents of the box, we implicitly assume a bijection which maps 
	each pre-tree $S^\circ_{b,i}$ to the auxiliary port $a_i^b$ of which it is a premise
	(which leaves $T$ to be mapped to $a_0^b$).
	We prefer to keep this information implicit in the following,
	as the notations should allow to recover it, whenever necessary.

	On the other hand, an analogue of Remark~\ref{rem:alpha}
	applies in this new setting,
	as pre-nets and nets should be considered up to some notion of 
	isomorphism preserving the interface,
	which amounts to:
	\begin{itemize}
		\item reindexing cuts, so that $\vec C$ is considered as a set;
		\item reordering premises of $\wn$-links, which accounts for the
			associativity and commutativity of the underlying binary
			contraction;
		\item renaming atoms and boxes and, simultaneously, 
			changing the enumeration of the family $\vec S^\circ_b$ 
			in each box, all this preserving the duality involution on variables,
			the partition $\{a^b_i\tq i\in\N\}_{b\in\N}$ of $\fp$, the jump functions,
			and the association of each $S^\circ_{b,i+1}$ to $a^b_{i+1}$.
	\end{itemize}
	We still consider this as a form of $\alpha$-equivalence as it only involves
	particular renamings of atoms or indices, preserving the rest of the
	structure.
	Again, we keep this quotient implicit whenever possible in the remaining.

	Also, as already mentioned in Remark~\ref{rem:jumps}, 
	we rely on jumps to control the combinatorics of 
	the elimination of evanescent cuts:
	we need nets to be equipped with jumps 
	only to ensure that the resource nets in the 
	Taylor expansion can also be equipped with jumps,
	that moreover enjoy uniform bounds.
	More precisely,
	we will show that if an MELL net $P$ can be equipped with $\jump_P$
	that satisfies the acyclicity criterion,
	then each $p\in\taylor(P)$ can be equipped with $\jump_p$,
	satisfying uniform bounds on $\ln(p)$ and $\jd(p)$ .

	The existence of such a jump function $\jump_P$ should be
	understood as side condition only:
	we keep it in the definition of nets by default 
	because we rely on it everywhere in the following,
	but in the end we are actually interested in the compatibility of Taylor
	expansion with cut elimination for nets without jumps.
	And the reader may check that, without jumps, our pre-nets
	(up to $\alpha$-equivalence) are essentially the same as, e.g.,
	the in-PS's (up to the names of internal ports)
	defined by de Carvalho~\cite{deCarvalho16}
	for his proof of the injectivity of Taylor expansion.

\end{rem}

\begin{figure}[t]
\begin{center}
\begin{tikzpicture}
	\node[lien](ax) at (0,0){ax};
	\node[lien](w) at ($(ax)+(-1.7,0)$){$\wn$};
	\node[lien,left=1 of w](b0){$\oc$};
	\coordinate[left=1 of b0](b01);
	\node[lien](cut) at ($($.5*(w)+.5*(b0)$)+(0,-0.7)$){cut};
	\node[lien] at ($(-.7,.8)+(b0)$) (ax0){ax};
	\node[lien] at ($(-1.7,0)+(ax0)$) (bot0){$\bot$};
	\coordinate(b02) at (bot0 |- b0);
	\node[lien] at (bot0 |- cut) (d0){$\wn$};
	\node[lien] at ($(.7,-1.4)+(ax)$) (b){$\oc$};
	\node[lien] at ($(-3.2,-.8)+(b)$) (c){$\wn$};
	\coordinate[left=1.2 of b](b1);
	\coordinate(b2) at (b-|b01);
	\node[lien] at (d0 |- c) (d){$\wn$};
	\coordinate(b3) at (d0|-b);

	\node[above left](x) at (ax0){$x\ $};
	\node[above right]at (ax0){$\ \bar x$};
	\node[above left](y) at (ax) {$y\ $};
	\node[above right]at (ax) {$\ \bar y$};
	\node[below right]at (w.south){$\w$};
	\node[below left=1ex]at (bot0){$\w'$};
	\node[below left]at (b01){$1$};
	\node[below left]at (b02){$2$};
	\node[below right]at (b1){$1$};
	\node[below left]at (b2){$2$};
	\node[below left]at (b3){$3$};
	\node[below right=1ex]at (b){$b$};
	\node[below left=1ex]at (b0){$b'$};

	\draw[saut] (w) to[out=80,in=100] (y);
	\draw[saut] (bot0) to[out=80,in=100] (x);
 	\draw (cut)to[out=west,in=south](b0);
 	\draw (cut)to[out=east,in=south](w);
 	\draw (ax0)to[out=east,in=north](b0);
 	\draw (ax)to[out=east,in=north](b);

	\coordinate[above=.5ex] (in0) at (b0);
	\coordinate[above=.5ex] (in) at (b);

	\draw
		(bot0)--(bot0|-in0)
		(b02)--(d0)--(b3|-in)
		(b3)--(d)
		;
	\draw
		(ax0) to[out=west,in=north] (b01|-in0)
		(b01)--(b2|-in) 
		(b2) to[out=south,in=north west] (c);
	\draw(ax)
		to[out=west,in=north] (b1|-in)
		(b1) to[out=south,in=north east] (c);
	\draw(b0.east) --++(.3,0) --++(0,1.7) --++(-3.8,0) --++(0,-1.7) -- (b0.west);
	\draw(b.east) --++(.3,0) --++(0,3.4) --++(-8.5,0) --+(0,-3.4) -- (b.west);
	\draw(c)--(c|- 0,-2.7);
	\draw(b)--(b|- 0,-2.7);
	\draw(d)--(d|- 0,-2.7);
\end{tikzpicture}
\end{center}
\caption{Representation of the net
$(\Theta;;\wn(a^b_3),\wn(a^b_2,a^b_1),a^b_0;\jump)$ 
where
$\boxes_\Theta=\{b\}$, $\ar(b)=3$
and
$\Theta(b)=(\Theta';\cut{a^{b'}_0}{\w};\bar y,y,a^{b'}_1,\wn(a^{b'}_2);\jump')$
where 
$\boxes_{\Theta'}=\{b'\}$, $\ar(b')=2$, $\jump'(\w)=y$
and 
$\Theta'(b')=(\Theta'';;\bar x,x,\w';\jump'')$
where $\boxes_{\Theta''}=\emptyset$ and $\jump''(\w')=x$.
}
\label{fig:ex_mell}
\end{figure}

Given a pre-net $P^\circ=(\Theta;\vec C;\vec S^\circ;\jump)$,
we write
$\fv(P^\circ)=\fv(\vec C;\vec S^\circ)$,
$\ST(P^\circ)=\ST(\vec C;\vec S^\circ)$,
\emph{etc.}
We define the \definitive{toplevel size} of MELL pre-nets by
$\tlsize(P^\circ)=\card\SPT(P^\circ)$.
We write $\depth(P^\circ)$ for the maximum level of nesting of boxes in $P^\circ$,
\emph{i.e.} the inductive depth in the above definition of pre-nets.
The \definitive{size} of MELL pre-nets includes that of their boxes: we set
$\size(P^\circ)=\tlsize(P^\circ)+\sum_{b\in\boxes_\Theta}\size(\Theta(b))$ ---
this definition is of course by induction on $\depth(P^\circ)$.

Notice that, by the above definition, for all $\w\in\nleaves(\vec C,\vec S^\circ)$,
$\jump_{P^\circ}(\w)$ must be at the same depth as $\w$, and
cannot be an auxiliary port.

We extend the switching functions of MLL to $\wn$-links: 
for each $T=\wn(T^\circ_1,\dots,T^\circ_n)\in\ST(P^\circ)$, $I(T)\in\{T^\circ_1,\dots,T^\circ_n\}$,
which induces a \definitive{$\wn$-edge} $T\sim^{P^\circ,I}_T I(T)$.
We also consider \definitive{box edges} $a^b_0 \sim^{P^\circ}_{b,i} a^b_i$ for $b\in\boxes_\Theta$
and $1\le i\le ar(b)$: w.r.t. paths,
a box $b$ behaves like $\ar(b)$ axiom links having the principal port of the box
as a common vertice, and the content is not considered.
Finally, \definitive{jump edges} also include the case of
weakenings: $\w\sim^{P^\circ}_{\w} \jump(\w)$ for $\w\in\nleaves(P^\circ)$.

We write $\paths(P^\circ,I)$ (resp. $\paths(P^\circ)$) for the set of $I$-paths (resp. paths) in $P^\circ$.
We say a pre-net $P^\circ$ is \definitive{acyclic} if there is no cycle in $\paths(P^\circ)$
and, inductively, each $\Theta(b)$ is acyclic. From now on, we consider acyclic pre-nets only.

\subsection{Resource nets and Taylor expansion}

The Taylor expansion of a net $P$ will be a set of \definitive{resource nets}:
these are the same as the multiplicative nets introduced before,
with the addition of term constructors for $\oc$ and $\wn$.
Raw trees are given as follows:
\[
	t\recdef x \mid \unit\cw \mid\counit\w\mid \otimes(t_1,\dotsc,t_n)\mid \parr(t_1,\dotsc,t_n)\mid \coweak{\cw'}\mid \weak{\w'} \mid \oc(t_1,\dotsc,t_n) \mid \wn(t_1,\dotsc,t_n).
\]
where $x$ ranges in $\fv$,
$\cw$ ranges over $\units$, $\w$ ranges over $\counits$, 
$\cw'$ ranges over $\coweaks$, $\w'$ ranges over $\weaks$, 
and we require $n\not=0$ in each case.
In resource nets, we extend switchings to $\wn$-links
and jumps from weakenings as in MELL nets,
associated with $\wn$-edges and jump edges.
Moreover, for each $t=\oc(t_1,\dots,t_n)$, we have \definitive{$\oc$-edges} 
$t\sim_{t,t_i}t_i$ for $1\le i\le n$.
Observe that, except for the notation of the root constructor, the trees
$\coweak{\cw}$, $\weak{\w}$, $\oc(t_1,\dotsc,t_n)$ and $\wn(t_1,\dotsc,t_n)$,
are exactly the same as 
$\unit\cw$, $\counit\w$, $\otimes(t_1,\dotsc,t_n)$ and $\parr(t_1,\dotsc,t_n)$ 
respectively:
in particular they induce the same geometry for paths.

During Taylor expansion, we need to replace a box in a pre-net
with an arbitrary number of approximants of this box.
Let us call \definitive{box replacement of arity $n$} the data
$r=(\vec s_0,\dotsc,\vec s_n)$ of $n+1$ families of pairwise distinct resource
trees $\vec s_0,\dotsc,\vec s_n$, such that $\atoms(s)\cap\atoms(s')=\emptyset$
whenever $s\in\vec s_i$ and $s'\in\vec s_{i'}$, and $i\not=i'$ or $s\not=s'$.
A family $\vec r=(r_b)_{b\in \boxes}$ of box replacements
such that $\atoms(r_b)\cap\atoms(r_{b'})=\emptyset$ for $b\not=b'\in\boxes$
is \definitive{applicable to the pre-term $T^\circ$} if
$\atoms(T^\circ)\cap\atoms(\vec r)=\emptyset$ and, for each
$a^b_i\in\fp(T^\circ)$, $b\in\boxes$ and $r_b$ is of arity at least $i$.

\begin{defi}\label{replacements}
	Let $\vec r$ be a $B$-indexed family of box replacements,
	and write $r_b=(\vec s^b_0,\dotsc,\vec s^b_{n_b})$ for each $b\in B$.
	Assuming that $\vec r$ is applicable to the tree $S$
	(resp. the pre-tree $S^\circ$),
	the \definitive{substitution of $\vec r$ for the boxes of $S$}
	(resp. of $S^\circ$) is the tree $\bsubs{S}{\vec r}$
	(resp. the family of pre-trees $\rsubs{S^\circ}{\vec r}$) defined
	by mutual induction on pre-trees and trees as follows:
	\begin{gather*}
		\begin{aligned}
			\bsubs{x}{\vec r}&=x&
			\bsubs{\cw}{\vec r}&=\cw&
			\bsubs{\w}{\vec r}&=\w&
			\bsubs{a^b_0}{\vec r}&=\begin{cases}
			\cw_b&\text{if $\vec s^b_0$ is empty}\\
			\oc(\vec s^b_0)&\text{otherwise}\\
			\end{cases}
		\end{aligned}
		\\
		\begin{aligned}
			\bsubs{\otimes(T_1,\dotsc,T_n)}{\vec r}&=
			\otimes(\bsubs{T_1}{\vec r},\dotsc,\bsubs{T_n}{\vec r})\\
			\bsubs{\parr(T_1,\dotsc,T_n)}{\vec r}&=
			\parr(\bsubs{T_1}{\vec r},\dotsc,\bsubs{T_n}{\vec r})\\
			\bsubs{\wn(T^\circ_1,\dotsc,T^\circ_n)}{\vec r}&=
		\begin{cases}
			\w_{\wn(T^\circ_1,\dotsc,T^\circ_n)}&
			\text{if $\rsubs{T^\circ_1}{\vec r},\dotsc,\rsubs{T^\circ_n}{\vec r}$ is empty}\\
			\wn(\rsubs{T^\circ_1}{\vec r},\dotsc,\rsubs{T^\circ_n}{\vec r})&
			\text{otherwise}
		\end{cases}
		\end{aligned}
		\\
		\rsubs{T}{\vec r}=\bsubs{T}{\vec r}\qquad
		\rsubs{a^b_{i+1}}{\vec r}=\vec s^b_{i+1}
	\end{gather*}
	where each $\cw_b\in\coweaks$ and each $\w_{\wn(T^\circ_1,\dotsc,T^\circ_n)}\in\weaks$
	is chosen fresh (not in $\atoms(S)$ nor $\atoms(S^\circ)$ nor $\atoms(\vec r)$)
	and unique.\footnote{
		So, formally, this construction should be parametrized by suitable injections 
		$\{a^b_0\in\ST(S^\circ)\}\to \coweaks$ and $\{\wn(T^\circ_1,\dotsc,T^\circ_n)\in \ST(S^\circ)\}\to \weaks$
		to ensure this linearity constraint.
		We keep this implicit in the following, but will rely on the fact 
		that, given $t\in\ST(\rsubs {S^\circ}{\vec r})$, one can
		recover unambiguously one of the following:
		either $T\in\ST(S^\circ)$ such that $t=\bsubs {T}{\vec r}$;
		or $b$ and $j$ such that $t\in\ST(\vec s^b_j)$.
	}
\end{defi}

We are now ready to introduce the expansion of MELL nets depicted
in Figure~\ref{fig:taylor}.\footnote{
	More extensive presentations of the Taylor expansion of MELL
	nets exist in the literature, in various styles
	\cite[among others]{PT09,GPF16,deCarvalho16}.
	Our only purpose here is to introduce sufficient notations
	to present our analysis of the jump degree and the length of paths in
	$\taylor(P)$ w.r.t. the size of $P$.
}
During the construction, we need to track the conclusions 
of copies of boxes, in order to collect copies of auxiliary ports 
in the external $\wn$-links: this is the role of the intermediate
notion of pre-Taylor expansion.

First, recall that we write $T_b,S^\circ_{b,1},\dots,S^\circ_{b,\ar(b)}$
for the trees of $\Theta(b)$ that are respectively mapped to
$a_0^b,a_1^b,\dots,a_{\ar(b)}^b$.
Also, in this case, let us write $\vec S^\circ_b=(S^\circ_{b,1},\dotsc,S^\circ_{b,\ar(b)})$.
\begin{defi}\label{defTaylor}
Given a closed pre-net $P^\circ=(\Theta;\vec C;\vec S^\circ;\jump)$,
a \definitive{pre-Taylor expansion} of $P^\circ$ is any pair $(p,f)$ 
of a resource net $p=(\vec c;\vec t;\jump_p)$,
together with a function $f:\vec t\to\vec S^\circ$ such 
that $f^{-1}(T)$ is a singleton whenever $T\in\vec S^\circ$ is a tree,
obtained as follows:
\begin{itemize}
	\item for each $b\in\boxes_{\Theta}$,
		fix a number $k_b\geq 0$ of copies;
	\item for $1\le j\le k_b$,
		fix inductively a pre-Taylor expansion $(p^b_j,f^b_j)$ of $\Theta(b)$,
		renaming the atoms so that the sets $\atoms(p^b_j)$ are pairwise disjoint,
		and also disjoint from $\atoms(\vec C)\cup\atoms(\vec S^\circ)$;
	\item write $p^b_j=(\vec c^b_j;t^b_j,\vec s^b_j;\jump^b_j)$ so that $f^b_j(t^b_j)=T_b$;
	\item write $\vec r=(r_b)_{b\in \boxes_\Theta}$ for the family of 
		box replacements $r_b=(\vec u^b_0,\dotsc,\vec u^b_{\ar(b)})$, where
		$\vec u^b_0=(t^b_1,\dotsc,t^b_{k_b})$
		and each $\vec u^b_i$ is an enumeration of  $\bigcup_{j=1}^{k_b}(f^b_j)^{-1}(S^\circ_{b,i})$
		for $1\le i\le \ar(b)$;
	\item set $\vec t=\rsubs{\vec S^\circ}{\vec r}$ and $\vec c=\bsubs{\vec C}{\vec r},\vec c'$ 
		where $\vec c'$ is the concatenation of the families $\vec c^b_j$
		for $b\in \boxes_\Theta$ and $1\le j\le k_b$
	\item for $t\in\vec t$, set $f(t)=a^b_i$ if $t\in\vec u^b_i$ with $1\le i\le \ar(b)$,
		otherwise let $f(t)$ be the tree $T\in\vec S^\circ$
		such that $t=\bsubs T{\vec r}$;
	\item for each $\w\in\nleaves(p)$, $\jump_p(\w)$ is defined as follows:
		\begin{itemize}
			\item if $\w\in\nleaves(p_j^b)$ then we set
				$\jump_p(\w)=\jump_{p_j^b}(\w)$.
			\item if $\w=\w_{\wn(T^\circ_1,\dotsc,T^\circ_n)}$
				then each $T_i^\circ\{\vec r\}$ is empty;
				then we select any $i\in\{1,\dotsc, n\}$
				and set $\jump_p(\w)=\bsubs{a_0^b}{\vec r}$
				where $b\in\boxes_\Theta$ is the box such that 
				$T^\circ_i=a^b_{j}$ for some $1\le j\le \ar(b)$;
			\item otherwise $\w\in\nleaves(\vec C;\vec S^\circ)$,
				and then we set $\jump_p(\w)=\bsubs{\jump(\w)}{\vec r}$
				(note that $\jump(\w)$ is a tree
				so this is a valid application of Definition~\ref{replacements}).
		\end{itemize}
\end{itemize}
		The \definitive{Taylor expansion}	of a net $P$ is then
		$\taylor(P)=\{p\tq (p,f)\text{ is a pre-Taylor expansion of }P\}$.
\end{defi}

\begin{exa}
Given the net 
$P=(\Theta;;\wn(a^b_3),\wn(a^b_2,a^b_1),a^b_0;\jump)$ 
of Figure~\ref{fig:ex_mell},
we construct an element $p$ of $\taylor(P)$ as follows.
First, we take two copies of box $b$, fixing $k_b=2$.
Recall that
$\Theta(b)=(\Theta';\cut{a^{b'}_0}{\w};\bar y,y,a^{b'}_1,\wn(a^{b'}_2);\jump')$
where $\boxes_{\Theta'}=\{b'\}$.
Hence, to construct $(p^b_j,f^b_j)$ we must first fix
a number $k_{b',j}$ of copies of the box $b'$:
we set $k_{b',1}=0$ and $k_{b',2}=1$,
and it remains to select a single pre-Taylor expansion 
$(p',f')$ of $\Theta'(b')$ for the only copy of $b'$ in $(p^b_2,f^b_2)$.
Since 
$\Theta'(b')=(\Theta'';;\bar x,x,\w';\jump'')$
contains no box,
we must have $p'=(;\bar x,x,\w';\jump'')$ 
with $f'(\bar x)=\bar x$, $f'(x)=x$ and $f'(\w')=\w'$.

Since $k_{b',1}=0$, we construct
$p^b_1=\big(\rsubs{\cut{a^{b'}_0}{\w}}{r_{b',1}};\rsubs{(\bar y,y,a^{b'}_1,\wn(a^{b'}_2))}{r_{b',1}};\jump^b_1\big)$
where $r_{b',1}$ is the empty replacement:
we obtain $p^b_1=(\cut{\cw}{\w};\bar y,y,\w'')$
where $\cw\in\coweaks$ and $\w''\in\weaks$ are fresh,
and we set $f^b_1(\bar y)=\bar y$, $f^b_1(y)=y$, $f^b_1(\w'')=\wn(a^{b'}_2)$,
and also $\jump^b_1(\w)=\bsubs{\jump'(\w)}{r_{b',1}}=y$
and $\jump^b_1(\w'')=\bsubs{a^{b'}_0}{r_{b',1}}=\cw$.

Having defined $(p',f')$ as above, we must set 
$r_{b',2}=((\bar x),(x),(\w'))$ and we define
$p^b_2=\big(\rsubs{\cut{a^{b'}_0}{\w}}{r_{b',2}};\rsubs{(\bar y,y,a^{b'}_1,\wn(a^{b'}_2))}{r_{b',2}};\jump^b_2\big)$:
we obtain $p^b_2=(\cut{\oc(\bar x)}{\w};\bar y,y,x,\wn(\w'))$,
with $f^b_2(\bar y)=\bar y$, $f^b_2(y)=y$, $f^b_2(x)=a^{b'}_1$ and $f^b_2(\wn(\w'))=\wn(a^{b'}_2)$,
and we set $\jump^b_2(\w)=\bsubs{\jump'(\w)}{r_{b',2}}=y$
and  $\jump^b_2(\w')=\jump''(\w')=x$.

We rename the atoms in both pre-Taylor expansions systematically as follows:
$p^b_1=(\cut{\cw_1}{\w_1};\bar y_1,y_1,\w''_1)$ and
$p^b_2=(\cut{\oc(\bar x_2)}{\w_2};\bar y_2,y_2,x_2,\wn(\w'_2))$,
also redefining $f^b_1$, $\jump^b_1$, $f^b_2$ and $\jump^b_2$
accordingly.

Finally, we set $\vec c=\cut{\cw_1}{\w_1},\cut{\oc(\bar x_2)}{\w_2}$
and $\vec t=\rsubs{(\wn(a^b_3),\wn(a^b_2,a^b_1),a^b_0)}{r_b}$ 
where  $r_b=((\bar y_1,\bar y_2),(y_1,y_2),(x_2),(\w''_1,\wn(\w'_2)))$.
We obtain:
\[p=(\cut{\cw_1}{\w_1},\cut{\oc(\bar x_2)}{\w_2};\wn(\w''_1,\wn(\w'_2)),\wn(x_2,y_1,y_2),\oc(\bar y_1,\bar y_2);\jump_p)\]
with $\jump_p(\w_1)=y_1$, $\jump_p(\w_2)=y_2$,
$\jump_p(\w''_1)=\cw_1$ and $\jump_p(\w'_2)=x_2$,
which is depicted in Figure~\ref{fig:ex_taylor}.
\end{exa}

\begin{figure}[t]
\begin{center}
\begin{tikzpicture}
	\node[lien] (coc){$\oc$};
	\node[lien] at ($(-3,-.8)+(coc)$) (con){$\wn$};

	\node[lien](ax1) at ($(coc)+(-1,1.5)$){ax};
	\node[lien](w1) at ($(ax1)+(-2.2,0)$){$\wn$};
	\node[lien,left=1 of w1](oc1){$\oc$};
	\node[lien](cut1) at ($($.5*(w1)+.5*(oc1)$)+(0,-0.7)$){cut};
	\node[lien] at ($(-4,0)+(cut1)$) (w11){$\wn$};
	\node[below left]at (w11){$\w''_1\ $};
	\node[above left](y1) at (ax1) {$y_1\ $};
	\node[above right]at (ax1) {$\ \bar y_1$};
	\node[below right]at (w1.south){$\w_1$};
	\node[below left=1ex](cw1) at (oc1){$\cw_1$};

	\draw[saut] (w1) to[out=80,in=100] (y1);
	\draw[saut] (w11) to[out=north east,in=north west] (oc1);
 	\draw (cut1)to[out=west,in=south](oc1);
 	\draw (cut1)to[out=east,in=south](w1);
 	\draw (ax1)to[out=east,in=105](coc);

	\node[lien](ax2) at ($(ax1)+(0,2)$) {ax};
	\node[lien](w2) at (ax2 -| w1){$\wn$};
	\node[lien](oc2) at (ax2 -| oc1){$\oc$};
	\node[lien] at ($(-.7,.7)+(oc2)$) (ax20){ax};
	\node[lien](cut2) at ($($.5*(w2)+.5*(oc2)$)+(0,-0.7)$){cut};
	\node[lien] at ($(-1.7,0)+(ax20)$) (bot2){$\bot$};
	\node[lien] at (bot2 |- cut2) (d2){$\wn$};
	\node[lien] at ($.5*($(d2|-con)+(w11|-con)$)$) (con'){$\wn$};

	\node[above left](y2) at (ax2) {$y_2\ $};
	\node[above right]at (ax2) {$\ \bar y_2$};
	\node[below right]at (w2.south){$\w_2$};
	\node[above left](x2) at (ax20){$x_2\ $};
	\node[above right]at (ax20){$\ \bar x_2$};
	\node[below left]at (bot2){$\w'_2\ $};

	\coordinate[left=1 of oc2](b21);
	\coordinate(b2) at (ax1-|b21);
	\coordinate[left=.5 of ax1](b1);

	\draw[saut] (w2) to[out=80,in=100] (y2);
	\draw[saut] (bot2) to[out=80,in=100] (x2);
 	\draw (cut2)to[out=west,in=south](oc2);
 	\draw (cut2)to[out=east,in=south](w2);
 	\draw (ax20)to[out=east,in=north](oc2);
	\draw(ax20)
		to[out=west,in=north] (b21)
		to[out=south,in=north] (b2) 
		to[out=south,in=north west] (con);
	\draw(ax2)
		to[out=west,in=north] (b1)
		to[out=south,in=55] (con);
 	\draw (ax2)to[out=east,in=75](coc);
 	\draw (bot2)to[out=south,in=north](d2);

	\draw[par dessus] (ax1) to[out=west,in=75] (con);

 	\draw (d2)to[out=south,in=75](con');
 	\draw (w11)to[out=south,in=105](con');

	\draw(con) --(con |- 0,-1.3);
	\draw(con')--(con'|- 0,-1.3);
	\draw(coc) --(coc |- 0,-1.3);
\end{tikzpicture}
\end{center}
\caption{Representation of the resource net $p\in\taylor(P)$
	where $P$ is the net of Figure~\ref{fig:ex_mell},
	$k_b=2$, $k_{b',1}=0$ and $k_{b',2}=1$.
}
\label{fig:ex_taylor}
\end{figure}

\subsection{Paths in Taylor expansion}

In the following, we fix a pre-Taylor expansion $(p,f)$ 
of $P^\circ=(\Theta;\vec C;\vec S^\circ;\jump)$
and we describe the structure of paths in $p$. We show that the
critical case depicted in Figure~\ref{fig:cheminstaylor} is maximal, 
so that a path of $p$ passes through at most two copies of each box of $P^\circ$.

\begin{figure}[t]
\begin{center}
	\begin{tikzpicture}[scale=.6,baseline=(cdots.base)]
	\draw[reseau](0.5,0)--(0.5,1)--(4.5,1)--(4.5,0)--cycle;%boîte gauche
	\draw(2.2,-0.3)node{$\cdots$};
	\draw(2.5,0.5)node{$p_1^b$};
	
	\draw(5,0.5)node{$\cdots$}; %boîtes intermédiaires
	
	\draw[reseau](5.5,0)--(5.5,1)--(9.5,1)--(9.5,0)--cycle;%boîte droite
	\draw(7.5,0.5)node{$p_{k_b}^b$};
	\draw(6.8,-0.3)node{$\cdots$};
	
	\node[lien] (c1) at (3.5,-1.75){$?$};
	\node[lien] (c2) at (5.5,-1.75){$?$};
	\node[lien] (cc) at (6.5,-1.75){$!$};%(co)contractions
	\node[above=2pt of c1, points] {$\ldots$};
	\node[above=2pt of c2, points] {$\ldots$};
	\node[above=2pt of cc, points] {$\ldots$};
	\node (cdots) at (4.5,-1.75) {$\cdots$};
	
	\draw(c1)--(3.5,-2.5);
	\draw(c2)--(5.5,-2.5);
	\draw(cc)--(6.5,-2.5);%conclusions
	
	\draw[multi](1,0) to [out=270,in=125] (c1);
	\draw[multi](3,0) to [out=270,in=125] (c2);
	\draw(4,0) to [out=270,in=125] (cc);%branchement gauche
	
	\draw[multi](6,0) to [out=270,in=55] (c1);
	\draw[multi](8,0) to [out=270,in=55] (c2);
	\draw(9,0) to [out=270,in=55] (cc);%branchement droit

	%chemin : 
	\draw[chemin](0.8,0) to [out=270,in=125] ($(c1)+(-0.5,0)$) to
	($(c1)+(-0.5,-1)$);
	\draw[chemin,dotted](0.8,0) to[out=90,in=90] (4.2,0);
	\draw[chemin](4.2,0) to[out=-90,in=125] (6.5,-1.75)
	to[out=55,in=-90] (9.2,0);
	\draw[chemin,dotted](9.2,0)to[out=90,in=90] (7.8,0);
	\draw[chemin](7.8,0) to[out=-90,in=55] ($(c2)+(-0.5,0)$)
	to ($(c2)+(-0.5,-1)$);

\end{tikzpicture}
\end{center}
	\caption{Box paths in Taylor expansion of $P^\circ$: critical
	case}
	\label{fig:cheminstaylor}
\end{figure}

Observe that 
\[
	\ST(p)=
	\{\bsubs{T}{\vec r}\tq T\in\ST(\vec C,\vec S^\circ)\}
	\cup\bigcup_{b\in\boxes_\Theta}\bigcup_{j=1}^{k_b} \ST(p^b_j)
\]
(using the notations of Definition~\ref{defTaylor}).
It follows that, for each $t\in\ST(p)$:
\begin{itemize}
	\item either $t$ is in a copy of a box, \emph{i.e.}
		(up to $\alpha$-equivalence)
		$t\in\ST(p^b_j)$ for some $b\in\boxes_\Theta$
		and $1\le j\le k_b$, and then we say $t$ is \definitive{inner}
		and write $\boxname(t)=b$ and $\boxcopy(t)=(b,j)$;
	\item or there exists a unique $T\in\ST(P^\circ)$
		such that $t=\bsubs{T}{\vec r}$,
		and then we say $t$ is \definitive{outer},
		and write $\untaylor t=T$.
\end{itemize}
We further distinguish the \definitive{cocontractions} of $p$,
\emph{i.e.} the outer trees $\oc(t_1^b,\dotsc,t_{k_b}^b)$ for
$b\in\boxes_\Theta$, which we denote by $\oc_b$, so that
$\untaylor{\oc_b}=a^b_0$.

We say an edge $t\sim^{p,I}_e s$ of $p$ is an \definitive{inner edge} (resp. an
\definitive{outer edge}) if $t$ and $s$ are both inner (resp. outer) trees.
We say a path $\xi\in\paths(p)$ is an
\definitive{inner path} (resp. an \definitive{outer path}) if 
it crosses inner edges (resp. outer edges) only.

If $t\sim^{p,I}_e s$ is an inner edge then $\boxcopy(t)=\boxcopy(s)$
and we also have $t\sim_e^{p^b_j,I^b_j} s$ where
$(b,j)=\boxcopy(t)$ and $I^b_j$ is the restriction of $I$
to $\nleaves(p^b_j)$.
In this case, we also set $\boxname(e)=b$ and $\boxcopy(e)=(b,j)$.
If $\xi$ is an inner path, we set $\boxname(\xi)$ (resp.  $\boxcopy(\xi)$)
for the common value of $\boxname$ (resp. $\boxcopy$) on the edges crossed by
$\xi$, and we obtain:
\begin{lem}
	\label{lemma:inner}
	If $\xi$ is an inner path
	then $\xi\in\paths(p^b_j,I^b_j)$
	where $(b,j)=\boxcopy(\xi)$.
\end{lem}

The classification of the outer edges of $p$ is more delicate.
First, we associate a switching $\untaylor I$ of $P^\circ$
with each switching $I$ of $p$ as follows:
\begin{itemize}
	\item if $I(\parr(\bsubs{T_1}{\vec r},\dotsc,\bsubs{T_n}{\vec r}))=\bsubs{T_i}{\vec r}$, 
		we set $\untaylor I(\parr(T_1,\dotsc,T_n))=T_i$;
	\item if $I(\wn(\rsubs{T^\circ_1}{\vec r},\dotsc,\rsubs{T^\circ_n}{\vec r}))
		\in\rsubs{T^\circ_i}{\vec r}$, 
		we set $\untaylor I(\wn(T^\circ_1,\dotsc,T^\circ_n))=T^\circ_i$;
	\item if $\bsubs{\wn(T^\circ_1,\dotsc,T^\circ_n)}{\vec r}=\w$,
		$T^\circ_i=a^b_{j}$ and $\jump_p(\w)=\oc_b$,
		we set $\untaylor I(\wn(T^\circ_1,\dotsc,T^\circ_n))=T^\circ_i$.\footnote{
			Observe that there might be several possible choices
			for $T^\circ_i$ so $\untaylor I$ is not uniquely defined
			in this manner: our following constructions thus 
			depend on the choices we make for $\untaylor I$.
		}
\end{itemize}

If $t\sim^{p,I}_e s$ is an outer edge 
then, in each of the following cases,
we can define an $\untaylor I$-edge $\untaylor e$ of $P^\circ$
such that $e=\bsubs{\untaylor e}{\vec r}$
and $\untaylor t\sim^{P^\circ,\untaylor I}_{\untaylor e} \untaylor s$:
\begin{itemize}
	\item $e$ is an axiom edge, and we set $\untaylor e=e$;
	\item $e$ is a $\otimes$-edge,
		e.g. $t=\otimes(t_1,\dotsc,t_n)$ and $s=t_i$,
		and we set $\untaylor e=(\untaylor t,\untaylor{t_i})$;
	\item $e$ is a $\parr$-edge or a $\wn$-edge,
		e.g. $t=\parr(t_1,\dotsc,t_n)$, $s=t_i$ and $I(t)=s$,
		and we set $\untaylor e=\untaylor t$;
	\item $e=\cut ts$ is a cut edge,
		and we set $\untaylor e=\cut{\untaylor t}{\untaylor s}$;
	\item $e=\w\in\nleaves(P^\circ)\subset\nleaves(p)$
		and we set $\untaylor e=e$
		(observe that in this case
		we have $\jump_{P^\circ}(\w)=\untaylor{\jump_p(\w)}$).
\end{itemize}
If any of the above cases holds, we say the outer edge $e$ is \definitive{superficial}.

If $e$ is an outer edge that is not superficial
then $e$ must be a \definitive{created jump}:
$e=\w\in\weaks(p)\setminus\big(\weaks(P^\circ)
	\cup\bigcup_{b\in\boxes_\Theta}\bigcup_{j=1}^{k_b} \weaks(p^b_j)\big)$.
If, e.g., $t=\w$, then we can write 
$\untaylor t=\wn(T^\circ_1,\dotsc,T^\circ_n)$
and $s=\jump_p(t)=\oc_b$ where $b$ is such that
$\untaylor I(\untaylor{t})=a^b_{i}$ with $1\le i\le \ar(b)$.
In this case, we obtain a path 
$\boxedge e=\untaylor t\sim^{P^\circ,\untaylor I}_{\untaylor t} a^b_i
\sim^{P^\circ,\untaylor I}_{b,i} a^b_0=\untaylor s$.

\begin{lem}
	\label{lem:outer}
	If $\xi$ is an outer $I$-path in $p$,
	then there exists an $\untaylor I$-path $\untaylor\xi$ in $P^\circ$ 
	with $\ln(\untaylor\xi)\ge\ln(\xi)$.
\end{lem}
\begin{proof}
	It is sufficient to replace each outer edge $e$ crossed by $\xi$ with:
	\begin{itemize}
		\item either $\untaylor e$ if $e$ is superficial,
		\item or the path $\boxedge e$ or $\reverse{\boxedge e}$
			if $e$ is a created jump.
	\end{itemize}
	Observe indeed that if $t\sim^{p,I}_e s$ and $t'\sim^{p,I}_{e'} s'$
	are outer paths of length $1$ with $e\not=e'$ then the paths
	$\untaylor{(t\sim^{p,I}_e s)}:\untaylor t\leadsto_{P^\circ,\untaylor I} \untaylor s$
	and 
	$\untaylor{(t'\sim^{p,I}_{e'} s')}:\untaylor {t'}\leadsto_{P^\circ,\untaylor I} \untaylor {s'}$
	thus defined are disjoint, and of length at least $1$.
\end{proof}

Some edges are neither inner nor outer:
a \definitive{boundary edge} is an edge $t\sim^{p,I}_e s$ such that 
$t$ is outer and $s$ is inner,
in which case we set $\boxcopy(e)=\boxcopy(s)$.
There are two kinds of boundary edges:
\begin{itemize}
	\item the \definitive{principal boundary} of the box copy $(b,j)$
		is the $\oc$-edge $(\oc_b,t^b_j)$;
	\item an \definitive{auxiliary boundary} $e$ of the box copy $(b,j)$
		is any $\wn$-edge $t\sim^{p,I}_t s$
		where $I(t)=s\in\vec s^b_j$ is such that $f^b_j(s)=S^\circ_{b,i}$ with $1\le i\le \ar(b)$,
		in which case we must have $\untaylor I(\untaylor t)=a^b_i$,
		and then we write $\boxdoor e=i$ for the index of the corresponding
		auxiliary port.
\end{itemize}
We call \definitive{box path} any path of the form 
$\chi=(e)\xi(e')$ where $e$ and $e'$ are boundaries and $\xi$ is an inner path:
in this case, we write $\boxname(\chi)=\boxname(\xi)$ and $\boxcopy(\chi)=\boxcopy(\xi)$.
Obviously, any path $\xi$ with outer endpoints 
is obtained as an alternation of outer paths and box paths:
we can write uniquely $\xi=\xi_0\chi_1\xi_1\cdots\chi_n\xi_n$
where each $\xi_i$ is an outer path,
and each $\chi_i$ is a box path.

Let $\chi=(e)\xi(e'):t\leadsto_p s$ be a box path with $\boxcopy(\chi)=(b,j)$.
Since $e\not=e'$, then at most one of $e$ and $e'$ is the principal boundary,
and if both $e$ and $e'$ are auxiliary boundaries,
then we must have $\boxdoor e\not=\boxdoor{e'}$:
indeed $e$ is the only $\wn$-edge whose premises include $\rsubs{a^b_{\boxdoor e}}{\vec r}$.
We can thus define $\untaylor\chi:\untaylor t\leadsto_{P^\circ} \untaylor s$ as follows:
\begin{itemize}
	\item if $e$ and $e'$ are auxiliary boundaries then
		$\untaylor\chi=\untaylor t\sim_{\untaylor t} a^b_{\boxdoor e} \sim_{b,\boxdoor e}
		a^b_0\sim_{b,\boxdoor{e'}} a^b_{\boxdoor{e'}} \sim_{\untaylor s} \untaylor s$
	\item otherwise, e.g. $e'$ is principal and we set
		$\untaylor\chi=\untaylor t\sim_{\untaylor t} a^b_{\boxdoor e} \sim_{b,\boxdoor e} a^b_0 =\untaylor s$.
\end{itemize}

\begin{lem}
	\label{lem:cheminsTaylor}
	Assume $\xi=\xi_0\chi_1\xi_1\cdots\chi_n\xi_n: t\leadsto_{p,I} s$
	where each $\xi_i$ is an outer path,
	and each $\chi_i$ is a box path.
	Then, setting
	$\untaylor\xi=\untaylor{\xi_0}\untaylor{\chi_1}\untaylor{\xi_1}\cdots\untaylor{\chi_n}\untaylor{\xi_n}$
	we obtain $\untaylor\xi:\untaylor t\leadsto_{P^\circ,\untaylor I} \untaylor s$.
	Moreover, if $\boxname(\chi_i)=\boxname(\chi_j)=b$ and $i<j$, 
	then $j=i+1$, $\xi_i=\epsilon_{\oc_b}$,
	and $\boxcopy(\chi_i)\not=\boxcopy(\chi_j)$.
\end{lem}
\begin{proof}
	We have already observed in the proof of Lemma \ref{lem:outer} that if $\xi$
	and $\xi'$ are disjoint outer paths then $\untaylor{\xi}$ and
	$\untaylor{\xi'}$ are also disjoint.
	Similarly, if $\xi$ is outer and $\chi$ is a box path, it follows directly
	from the definitions that $\untaylor{\xi}$ and $\untaylor{\chi}$ are disjoint.
	And if $\chi$ and $\chi'$ are box paths with disjoint boundaries,
	again  $\untaylor{\chi}$ and $\untaylor{\chi'}$ are disjoint paths
	by construction.
	It follows that, if $\xi=\xi_0\chi_1\xi_1\cdots\chi_n\xi_n: t\leadsto_{p,I} s$,
	each $\xi_i$ is an outer path,
	and each $\chi_i$ is a box path,
	then the concatenation 
	$\untaylor\xi=\untaylor{\xi_0}\untaylor{\chi_1}\untaylor{\xi_1}\cdots\untaylor{\chi_n}\untaylor{\xi_n}$
	is well defined.

	Write $\chi_i=(e_i)\xi'_i(e'_i):t_i\leadsto s_i$ for $1\le i\le n$.
	Assume $\boxname(\chi_i)=\boxname(\chi_j)=b$, 
	and moreover $\boxname(\chi_k)\not=b$ for $i<k<j$.
	We obtain a path
	$\xi'=\untaylor{(\xi_i\chi_{i+1}\cdots\xi_{j-1})}:\untaylor{s_i}\leadsto_{P^\circ}\untaylor{t_j}$:
	by construction, $\xi'$ does not cross any box edge $(b,l)$ for $1\le l\le\ar(b)$.
	If $(e'_i)$ and $(e_j)$ were both auxiliary, 
	we could form a cycle
	$\xi'(\untaylor{t_j}\sim_{\untaylor{t_j}}a^b_{\boxdoor{e_j}}\sim_{b,\boxdoor{e_j}}
	a^b_0\sim_{b,\boxdoor{e'_i}} a^b_{\boxdoor{e'_i}} \sim_{\untaylor{s_i}}\untaylor{s_i})$,
	since $\xi'$ would cross neither $\untaylor{t_j}$ nor $\untaylor{s_i}$.
	If, e.g., $(e'_i)$ was principal and $(e_j)$ was auxiliary, 
	we could form a cycle
	$\xi'(\untaylor{t_j}\sim_{\untaylor{t_j}}a^b_{\boxdoor{e_j}}\sim_{b,\boxdoor{e_j}} a^b_0)$,
	as $\xi'$ would not cross $\untaylor{t_j}$.
	So both must be principal and we have $\untaylor{s_i}=\untaylor{t_j}=a^b_0$:
	since $P^\circ$ has no non empty cycle, 
	we must have $\xi'=\epsilon_{a^b_0}$ hence 
	$\xi_i\chi_{i+1}\cdots\xi_{j-1}=\epsilon_{\oc_b}$
	and then $j=i+1$ and $\xi_i=\epsilon_{\oc_b}$.
	Since $e'_i\not=e_j$, 
	we moreover obtain $\boxcopy(\chi_i)\not=\boxcopy(\chi_j)$.

	It remains only to prove that, in general, we never have 
	$\boxname(\chi_i)=\boxname(\chi_j)$ with $j>i+1$:
	otherwise, by iterating our previous argument,
	we would obtain $\boxname(\chi_k)=\boxname(\chi_i)$ whenever $i\le k\le j$,
	and both $e_{k}$ and $e'_{k}$ would both be principal boundaries whenever $i<k<j$.
\end{proof}

It follows that $p$ is acyclic as soon as $P^\circ$ is.
Indeed, if $\xi$  is a cycle in $p$:
\begin{itemize}
	\item either $\xi$ contains an outer tree,
		and we can apply Lemma~\ref{lem:cheminsTaylor} to obtain a cycle in $P^\circ$;
	\item or $\xi$ is an inner path,
		and we proceed inductively in $\Theta(\boxname(\xi))$.
\end{itemize}
Our next result is a quantitative version of this property:
not only there is no cycle in $p$ but the length of paths in $p$ is bounded
by a function of $P^\circ$ (whereas the size of $p$ is obviously not bounded in general).
\begin{thm}
	\label{theorem:length}
	If $p\in\taylor(P^\circ)$ and $\xi\in\paths(p)$ then $\ln(\xi)\le 2^{\depth(P^\circ)}\size(P^\circ)$.
\end{thm}
\begin{proof}
	The proof is by induction on $\depth(P^\circ)$.

	First assume that $\xi=\xi_0(e_1)\chi_1(e'_1)\xi_1\cdots(e_n)\chi_n(e'_n)\xi_n$ where
	each $\xi_i$ is an outer path, and each $(e_i)\chi_i(e'_i)$ is a box path.
	Write $(b_i,j_i)=\boxcopy(\chi_i)$:
	by applying the induction hypothesis to $\chi_i\in\paths(p^{b_i}_{j_i})$,
	we obtain $\ln(\chi_i)\le 2^{\depth(\Theta(b_i))}\size(\Theta(b_i))$.
	Moreover observe that 
	$2n+\sum_{i=0}^n \ln(\untaylor{\xi_i})\le \ln(\untaylor{\xi})\le\tlsize(P^\circ)$.
	By Lemma~\ref{lem:outer},
	it follows that $2n+\sum_{i=0}^n \ln(\xi_i)\le \tlsize(P^\circ)$.
	We obtain:
	\[
		\ln(\xi)
		= 2n+\sum_{i=0}^n\ln(\xi_i)+\sum_{i=1}^{n}\ln(\chi_i)
		\le \tlsize(P^\circ)
		+ \sum_{i=1}^n 2^{\depth(\Theta(b_i))}\size(\Theta(b_i)).
	\]
	By Lemma~\ref{lem:cheminsTaylor},
	each $b\in\boxes_\Theta$ occurs at most twice in
	the sequence $(b_1,\dotsc,b_n)$, hence 
	we obtain:
	\[
		\ln(\xi) 
		\le \tlsize(P^\circ)
		+ 2\sum_{b\in\boxes_\Theta} 2^{\depth(\Theta(b))}\size(\Theta(b)).
	\]
	hence
	\[
		\ln(\xi) 
		\le 2^{\depth(P^\circ)} \big(\tlsize(P^\circ)
		+ \sum_{b\in\boxes_\Theta} \size(\Theta(b))\big).
	\]
	since $\depth(\Theta(b))<\depth(P^\circ)$ for each $b\in\boxes_\Theta$.
	We conclude recalling that $\size(P^\circ)=\tlsize(P^\circ)+\sum_{b\in\boxes_\Theta}\size(\Theta(b))$.

	The other possible cases are those of paths 
	$\chi_0(e'_0)\xi$, $\xi(e_{n+1})\chi_n$ or
	$\chi_0(e'_0)\xi(e_{n+1})\chi_n$
	where $\xi$ is as above $e'_0$ and $e_{n+1}$ are boundaries
	and $\chi_0$ and $\chi_n$ are inner paths.
	Reasonning as in the proof of Lemma~\ref{lem:cheminsTaylor},
	we also obtain that each $b\in\boxes_\Theta$ occurs at most twice in
	the sequence, e.g., $(b_0,\dotsc,b_{n+1})$,
	and then the proof follows similarly.
\end{proof}
In particular, we obtain $\ln(p)\le 2^{\depth(P^\circ)}\size(P^\circ)$,
In the following lemma, we show that our measure on jumps in the
Taylor expansion of $P^\circ$ is also entirely determined by $P^\circ$.
\begin{lem}\label{QI_Taylor}
	If $p\in\taylor(P^\circ)$
	then $\jd(p)\leq\size(P^\circ)$.
\end{lem}
\begin{proof}
	We show that if $t\in\ST(p)$ then $\jd(t)\le\size(P^\circ)$.
	The proof is, again, by induction on $\depth(P^\circ)$.
	If $t$ is inner with $\boxcopy(t)=(b,j)$, then we conclude directly by applying 
	the induction hypothesis to $p^b_j$ and $\Theta(b)$:
	indeed in this case, $\jump^{-1}_p(t)=\jump^{-1}_{p^b_j}(t)$,
	and $\size(\Theta(b))\le\size(P^\circ)$.

	So we can assume that $t$ is outer. In this case, observe
	from  Definition~\ref{defTaylor} that if
	$\jump_p(\w)=t$ then $\w=\bsubs T{\vec r}$ for some $T\in\ST(P^\circ)$.
	It follows that $\card{\jump_p^{-1}(t)}\le\card{\ST(P^\circ)}\le\size(P^\circ)$.
	\qedhere

\end{proof}

\subsection{Cut elimination and Taylor expansion}

In resource nets \cite{ER05}, the elimination of the cut \[\cut{\wn(t_1,\dots,t_n)}{\oc(s_1,\dots,s_m)}\]
yields the finite sum
\[\sum_{\sigma:\{1,\dotsc,n\}\stackrel\sim\to\{1,\dotsc,m\}}\cut{t_1}{s_{\sigma(1)}},\dots,\cut{t_n}{s_{\sigma(n)}}.\]
It turns out that the results of Sections~\ref{section:size} to \ref{section:general} apply directly to
resource nets: setting 
\[\cut{\wn(t_1,\dots,t_n)}{\oc(s_1,\dots,s_n)}\to
\cut{t_1}{s_{\sigma(1)}},\dots,\cut{t_n}{s_{\sigma(n)}} \]
for each permutation $\sigma$,
we obtain an instance of multiplicative reduction, as the order of premises is
irrelevant from a combinatorial point of view ---
this is all the more obvious because no typing constraint was involved in our
argument. In other words, Corollary~\ref{corollary:antireducts} also applies
to the parallel reduction of resource nets.
With Theorem~\ref{theorem:length} and Lemma~\ref{QI_Taylor} we obtain:
\begin{cor}
	\label{corollary:taylor}
If $q$ is a resource net and $P$ is an MELL net and $k\in\N$,
$\{p\in\taylor(P)\tq p\toto^k q\}$ is finite.
\end{cor}
As for Corollaries~\ref{corollary:antireducts:onestep} and \ref{corollary:antireducts},
this holds only up to $\alpha$-equivalence.
And, again, it should be read keeping in mind that jumps are only 
an additional control structure on top of the underlying net.
Indeed, if $P$ is a bare MELL net (\emph{i.e.} an MELL net without a jump function)
then we can define $\taylor(P)$ as a set of bare resource nets.
Then, given $k\in\N$, a bare resource net $q$, a bare MELL net $P$,
and a jump function $\jump$ such that $(P,\jump)$ acyclic,
there are finitely many bare resource nets $p\in\taylor(P)$ such that
$p\toto^k q$: it suffices to construct $\jump_p$ from $\jump$.

Beware that Corollary~\ref{corollary:taylor}
depends on the acyclicity of the original MELL net.
The following example shows how a cyclic net can induce infinite sets of
antireducts.
\begin{exa}\label{ex:cycle2}
	Let $P=(\Theta; \cut{\wn(a_1^b)}{a_0^b};;\jump)$ with
	$\Theta(b)=(\Theta';;x,\bar x;;\jump)$
	where the domain of $\Theta'$, $\jump$ and $\jump'$ is empty.
	Then, by definition, 
	\[\taylor(P)=\{p=(\cut{\cw}{\w};)\}\cup\{ p_n=(\cut{\wn(x_0,\dotsc,x_n)}{\oc(\bar x_0,\dotsc,\bar x_n)};)\tq n\in\N\} \]
	where $\jump_{p}(\cw)=\w$.
	Then, for each $n\in\N$, we have
	\[
		p_n\to (\cut{x_0}{\bar x_{\sigma(0)}},\dotsc,\cut{x_n}{\bar x_{\sigma(n)}};)
	\]
	for each permutation $\sigma$ of $\{0,\dotsc,n\}$.
	In particular, if we set $\sigma(i)=i+1\mod(n+1)$,
	then we obtain $p_n\to q_n=(\cut{x_0}{\bar x_1},\dotsc,\cut{x_n}{\bar x_0};)\totoa(\cut{x_0}{\bar x_0};)$,
	and it follows that $\{q\in\taylor(P)\mid q\toto^2 (\cut{x_0}{\bar x_0};)\}$ is infinite.
	This situation is illustrated in Figure~\ref{fig:cycle2}.
\end{exa}
\begin{figure}[t]
\begin{center}
	\begin{tikzpicture}
		\node[lien](a)at(0,.5){ax};
		\node[lien](c)at($(a)+(0,-1.5)$){cut};
		\node[lien](b)at($(a)+(0.5,-0.5)$){$\oc$};
		\node[lien](w)at($(a)+(-0.5,-1)$){$\wn$};
		\coordinate[above=.5ex] (in) at (b);

		\draw[above left](a)node{$x\ $};
		\draw[above right](a)node{$\ \bar x$};

		\draw[reseau](a)to[out=west,in=north](w|-in);
		\draw[reseau](w|-b)--(w);
		\draw[reseau](a)to[out=east,in=north](b);
		\draw[reseau](w)to[out=south,in=west](c);
		\draw[reseau](b)to[out=south,in=east](c);
		%boite
		\draw(b)--++(-1.5,0)--++(0,1)--++(1.9,0)--++(0,-1)--(b);

		\node (P) at (0,-2.5) {$P$};

		\node (pn) at ($(P)+(4.2,0)$) {$p_n\in\taylor(P)$};
		
		\node[lien](c)at($(pn)+(0,.8)$){cut};
		\node[lien](w)at($(c)+(-1.25,1)$){$\wn$};
		\node[lien](b)at($(c)+(1.25,1)$){$\oc$};
		\node[lien](a1)at($(0,3)+1/2*($(w)+(b)$)$){ax};
		\node[lien](an)at($(a1)+(0,-1.3)$){ax};
		
		\draw[above left](a1)node{$x_0\ $};
		\draw[above right](a1)node{$\ \bar x_0$};
		\draw[above left](an)node{$x_n\ $};
		\draw[above right](an)node{$\ \bar x_n$};

		\draw(a1)to[out=west,in=north west](w);
		\draw(an)to[out=west,in=north east](w);
		\draw(a1)to[out=east,in=north west](b);
		\draw[par dessus](an)to[out=east,in=north east](b);

		\draw(w)to[out=south,in=west](c);
		\draw(b)to[out=south,in=east](c);

		\draw($(w)+(0,0.5)$)node{$\dots$};
		\draw($(b)+(0,0.5)$)node{$\dots$};

		\draw ($1/2*($(a1)+(an)$)$) node{$\rvdots$};

		\node (to) at ($(pn)+(2.5,3)$) {$\to$};

		\node[lien](a1)at($(pn)+(5,5)$){ax};
		\node[lien](an)at($(a1)+(0,-1.3)$){ax};
		\draw ($1/2*($(a1)+(an)$)$) node{$\rvdots$};

		\draw[above left](a1)node{$x_0\ $};
		\draw[above right](a1)node{$\ \bar x_0$};
		\draw[above left](an)node{$x_n\ $};
		\draw[above right](an)node{$\ \bar x_n$};

		\node[lien](c1)at($(an)+(-0.8,-1.4)$){cut};
		\node[lien](cn)at($(c1)+(0,-1.4)$){cut};
		\draw ($1/2*($(c1)+(cn)$)$) node{$\rvdots$};
		
		\draw(a1)to[out=west,in=west](c1);
		\draw[par dessus](an)to[out=west,in=west](cn);
		\draw[above right](c1)node{$\ \bar x_{\sigma(0)}$};
		\draw[above right](cn)node{$\ \bar x_{\sigma(n)}$};

		\draw[dotted,rounded corners]
		($(c1)+(1.2,0.2)$)--++(0.7,0)
		--($(cn)+(1.9,-0.2)$)--++(-0.7,0)--cycle;
		\draw(c1)--++(1.2,0);
		\draw(cn)--++(1.2,0);

		\draw($(cn)+(1.5,0.6)$)node{$\sigma$};

		\draw($(c1)+(1.9,0)$)to[out=east,in=east](a1);
		\draw[par dessus]($(cn)+(1.9,0)$)to[out=east,in=east](an);
		
		\node[lien](a)at($(an)+(4,0)$){ax};
		\node[lien](c)at($(a)+(0,-1.3)$){cut};
		\draw(a)to[out=west,in=west](c);
		\draw(a)to[out=east,in=east](c);
		
		\node at ($(to)+(5,0)$) {$\totoa$};

		\draw[above left](a)node{$x_0\ $};
		\draw[above right](a)node{$\ \bar x_0$};

	\end{tikzpicture}

\end{center}
\caption{Resource nets $p_n$ of $\taylor(P)$ reducing to a single net}
\label{fig:cycle2}
\end{figure}

\section{Conclusion}
\label{section:conclusion}

Recall that our original motivation was the definition of 
a reduction relation on infinite linear combinations 
of resource nets, simulating cut elimination 
in MELL through Taylor expansion.
We claim that a suitable notion is as follows:

\begin{defi}
Write $\sum_{i\in I}a_i p_i \Rightarrow \sum_{i\in I}a_i q_i$
as soon as:
\begin{itemize}
	\item each $p_i$ is a resource net and each $q_i$ is a finite sum of resource nets
		such that $p_i\toto q_i$;
	\item for any resource net $p$,
		$\{i\in I\tq p_i=p\}$ is finite;
	\item for any resource net $q$,
		$\{i\in I\tq q\text{ is a summand of }q_i\}$ is finite.
\end{itemize}
\end{defi}

In particular, if $\sum_{i\in I}a_i p_i$ is a Taylor expansion,
then Corollary~\ref{corollary:taylor} ensures that the last condition of the definition
of $\Rightarrow$ is automatically valid.
The details of the simulation in a quantitative setting 
remain to be worked out, but the main stumbling block
is now over:
the necessary equations on coefficients are well established,
as they have been extensively studied in the various denotational models;
it only remained to be able to form the associated sums 
directly in the syntax.

Another incentive to publish our results is the 
\emph{normalization-by-evaluation} programme that we develop
with Guerrieri, Pellissier and Tortora de Falco~\cite{CGPV17}.
This approach is restricted to connected MELL proof nets,
\emph{i.e.} MELL proof nets without weakening, and whose switching graphs are
not only acyclic but also connected:\footnote{
	These are sufficiently expressive to simulate 
	the $\lambda I$-calculus, which is Turing-complete.
}
\begin{itemize}
	\item in this setting, a net $P$ is entirely determined by the point of order
		2 of its Taylor expansion, \emph{i.e.} the unique resource net
		$p\in\taylor(P)$ with binary cocontractions only~\cite{GPF16};
	\item moreover, 
		given two cut-free $\eta$-expanded nets $Q$ and $R$, both the size of the normal
		form of a cut between $Q$ and $R$ and the number of cut elimination steps 
		necessary to reach it can be bounded by a function of the relational semantics 
		of $Q$ and $R$ \cite{dCPTdF11};
	\item from this data, we obtain a bound on the size 
		of the point $p_0$ of order 2 of the normal form of the cut,
		as well as a bound on the number of parallel cut 
    elimination steps necessary to obtain $p_0$ from its antecedent 
		$p$ in the Taylor expansion of the cut.
\end{itemize}
Our results in the present paper then provide a bound on the size of $p$:
to find $p$ it is then sufficient to compute 
the relational semantics of all the elements of the 
Taylor expansion of the cut whose size 
does not exceed this bound, and to check which 
one gives a semantics of order 2;
then we can compute $p_0$ as the normal form of $p$,
and this is sufficient to determine the normal form of the cut.

The restriction to connected nets is necessary to apply
the injectivity result of Guerrieri, Pellissier and Tortora de Falco~\cite{GPF16},
based on a fixed order of Taylor expansion.
The injectivity of Taylor expansion and thus of the relational semantics
of full MELL has been proved by de Carvalho~\cite{deCarvalho18}:
to determine $P$ from $\taylor(P)$,
this result relies on a $k$-heterogeneous expansion of $P$,
\emph{i.e.} an expansion of $P$ for which the number of copies of each box
is a power of $k$, and those degrees of expansion are chosen pairwise distinct.
For the result to apply, the value of the parameter $k$ must be sufficiently large:
such a $k$ may be computed from the linear expansion of $P$, obtained by taking
exactly one copy of each box;
but the degrees of expansion of boxes cannot be bounded in advance, and it is thus not clear if 
the above normalization-by-evaluation procedure could be adapted in this setting.

Let us conclude with a remark about a possible adaptation of our results to
a (maybe) more standard representation of nets, including separate
derelictions and coderelictions, with a finer grained cut elimination procedure.
This introduces additional complexity in the formalism but
it essentially requires no new concept or technique:
the difficulty in parallel reduction is to control the chains
of cuts to be simultaneously eliminated,
and decomposing cut elimination into finer reduction steps 
can only decrease the length of such chains.
On the other hand, in that setting,
it is well known that cut elimination alone is not enough to capture the
$\beta$-reduction of $\lambda$-calculus, and it must be
extended with additional rewriting rules accounting for structural identities
(e.g., associativity and commutativity of contraction).
The details of the Taylor expansion analysis of cut elimination up to these
identities are worked out in the PhD thesis of the first
author~\cite[Chapter 2]{chouquet:phd}\footnote{In French.}, 
including the treatment of coefficients as mentioned above.

\subsection*{Acknowledgements}
We would like to thank the anonymous referee whose keen remarks have 
allowed us to clarify several key aspects of our contributions.

This work also owes much to the friendly and stimulating environment provided by the
International Research Network on Linear Logic\footnote{\url{http://www.linear-logic.org/}}
between the French CNRS and the Italian INDAM:
the first steps were actually taken on the occasion of the very first workshop
of the network at Roma Tre in December 2015, dedicated to \emph{New trends in linear logic proof-nets};
and successive versions were presented and discussed on various occasions organized by the IRN-LL.

\bibliographystyle{alpha}
\bibliography{biblio}

\begin{thebibliography}{dCPTdF11}

\bibitem[CGPV17]{CGPV17}
Jules Chouquet, Giulio Guerrieri, Luc Pellissier, and Lionel Vaux.
\newblock Normalization by evaluation in linear logic.
\newblock In Stefano Guerrini, editor, {\em Preproceedings of the International
  Workshop on Trends in Linear Logic and Applications, TLLA}, September 2017.

\bibitem[Cho19]{chouquet:phd}
Jules Chouquet.
\newblock {\em {A geometry of calculus}}.
\newblock Theses, {Universit{\'e} de Paris}, December 2019.

\bibitem[CVA18]{CV18}
Jules Chouquet and Lionel Vaux~Auclair.
\newblock An application of parallel cut elimination in unit-free
  multiplicative linear logic to the taylor expansion of proof nets.
\newblock In Dan~R. Ghica and Achim Jung, editors, {\em 27th {EACSL} Annual
  Conference on Computer Science Logic, {CSL} 2018, September 4-7, 2018,
  Birmingham, {UK}}, volume 119 of {\em LIPIcs}, pages 15:1--15:17. Schloss
  Dagstuhl - Leibniz-Zentrum fuer Informatik, 2018.

\bibitem[dC07]{deCarvalho07}
Daniel de~Carvalho.
\newblock {\em Sémantiques de la logique linéaire et temps de calcul}.
\newblock PhD thesis, Université d'Aix-Marseille II, Marseille, France, 2007.

\bibitem[dC09]{deCarvalho09}
Daniel de~Carvalho.
\newblock Execution time of lambda-terms via denotational semantics and
  intersection types.
\newblock {\em CoRR}, abs/0905.4251, 2009.

\bibitem[dC16]{deCarvalho16}
Daniel de~Carvalho.
\newblock The relational model is injective for multiplicative exponential
  linear logic.
\newblock In Jean{-}Marc Talbot and Laurent Regnier, editors, {\em 25th {EACSL}
  Annual Conference on Computer Science Logic, {CSL} 2016, August 29 -
  September 1, 2016, Marseille, France}, volume~62 of {\em LIPIcs}, pages
  41:1--41:19. Schloss Dagstuhl - Leibniz-Zentrum fuer Informatik, 2016.

\bibitem[dC18]{deCarvalho18}
Daniel de~Carvalho.
\newblock Taylor expansion in linear logic is invertible.
\newblock {\em Logical Methods in Computer Science}, 14(4), 2018.

\bibitem[dCPTdF11]{dCPTdF11}
Daniel de~Carvalho, Michele Pagani, and Lorenzo Tortora~de Falco.
\newblock A semantic measure of the execution time in {L}inear {L}ogic.
\newblock {\em Theoretical Computer Science}, 412, 04 2011.

\bibitem[DE11]{DE11}
Vincent Danos and Thomas Ehrhard.
\newblock Probabilistic coherence spaces as a model of higher-order
  probabilistic computation.
\newblock {\em Information and Computation}, 209(6):966--991, 2011.

\bibitem[DR89]{DR89}
Vincent Danos and Laurent Regnier.
\newblock The structure of multiplicatives.
\newblock {\em Archive for Mathematical Logic}, 28(3):181--203, 1989.

\bibitem[Ehr02]{Ehrhard02}
Thomas Ehrhard.
\newblock On {K}{\"{o}}the sequence spaces and linear logic.
\newblock {\em Mathematical Structures in Computer Science}, 12(5):579--623,
  2002.

\bibitem[Ehr05]{Ehrhard05}
Thomas Ehrhard.
\newblock Finiteness spaces.
\newblock {\em Mathematical Structures in Computer Science}, 15(4):615--646,
  2005.

\bibitem[Ehr10]{Ehrhard10}
Thomas Ehrhard.
\newblock A finiteness structure on resource terms.
\newblock In {\em Proceedings of the 25th Annual {IEEE} Symposium on Logic in
  Computer Science, {LICS} 2010, 11-14 July 2010, Edinburgh, United Kingdom},
  pages 402--410, 2010.

\bibitem[Ehr14]{Ehrhard14}
Thomas Ehrhard.
\newblock A new correctness criterion for {MLL} proof nets.
\newblock In {\em Joint Meeting of the Twenty-Third {EACSL} Annual Conference
  on Computer Science Logic {(CSL)} and the Twenty-Ninth Annual {ACM/IEEE}
  Symposium on Logic in Computer Science (LICS), {CSL-LICS} '14, Vienna,
  Austria, July 14 - 18, 2014}, pages 38:1--38:10, 2014.

\bibitem[Ehr16]{Ehrhard16}
Thomas Ehrhard.
\newblock An introduction to differential linear logic: proof-nets, models and
  antiderivatives.
\newblock {\em CoRR}, abs/1606.01642, 2016.

\bibitem[ER03]{ER03}
Thomas Ehrhard and Laurent Regnier.
\newblock The differential lambda-calculus.
\newblock {\em Theoretical Computer Science}, 309(1-3):1--41, 2003.

\bibitem[ER05]{ER05}
Thomas Ehrhard and Laurent Regnier.
\newblock Differential interaction nets.
\newblock {\em Electr. Notes Theor. Comput. Sci.}, 123:35--74, 2005.

\bibitem[ER08]{ER08}
Thomas Ehrhard and Laurent Regnier.
\newblock Uniformity and the taylor expansion of ordinary lambda-terms.
\newblock {\em Theoretical Computer Science}, 403(2-3):347--372, 2008.

\bibitem[FM99]{fm:calcinets}
Maribel Fern{\'{a}}ndez and Ian Mackie.
\newblock A calculus for interaction nets.
\newblock In Gopalan Nadathur, editor, {\em Principles and Practice of
  Declarative Programming, International Conference PPDP'99, Paris, France,
  September 29 - October 1, 1999, Proceedings}, volume 1702 of {\em Lecture
  Notes in Computer Science}, pages 170--187. Springer, 1999.

\bibitem[Gir87]{Girard87}
Jean{-}Yves Girard.
\newblock Linear logic.
\newblock {\em Theoretical Computer Science}, 50:1--102, 1987.

\bibitem[Gir88]{Girard88}
Jean-Yves Girard.
\newblock Normal functors, power series and lambda-calculus.
\newblock {\em Annals of Pure and Applied Logic}, 37(2):129, 1988.

\bibitem[Gir96]{Girard96}
Jean-Yves Girard.
\newblock Proof-nets : the parallel syntax for proof-theory.
\newblock In Aldo Ursini and Paolo Aglianò, editors, {\em Logic and Algebra},
  number 180 in Lecture Notes in Pure and Applied Mathematics. Marcel Dekker,
  New York, 1996.

\bibitem[GPTdF16]{GPF16}
Giulio Guerrieri, Luc Pellissier, and Lorenzo Tortora~de Falco.
\newblock Computing connected proof(-structure)s from their taylor expansion.
\newblock In {\em 1st International Conference on Formal Structures for
  Computation and Deduction, {FSCD} 2016, June 22-26, 2016, Porto, Portugal},
  pages 20:1--20:18, 2016.

\bibitem[HH16]{hh:mll}
Willem Heijltjes and Robin Houston.
\newblock Proof equivalence in {MLL} is pspace-complete.
\newblock {\em Logical Methods in Computer Science}, 12(1), 2016.

\bibitem[Laf90]{lafont:interactionnets}
Yves Lafont.
\newblock Interaction nets.
\newblock In Frances~E. Allen, editor, {\em Conference Record of the
  Seventeenth Annual {ACM} Symposium on Principles of Programming Languages,
  San Francisco, California, USA, January 1990}, pages 95--108. {ACM} Press,
  1990.

\bibitem[LMMP13]{LairdMMP13}
Jim Laird, Giulio Manzonetto, Guy McCusker, and Michele Pagani.
\newblock Weighted relational models of typed lambda-calculi.
\newblock In {\em 28th Annual {ACM/IEEE} Symposium on Logic in Computer
  Science, {LICS} 2013, New Orleans, LA, USA, June 25-28, 2013}, pages
  301--310. {IEEE} Computer Society, 2013.

\bibitem[MS08]{ms:chaminets}
Ian Mackie and Shinya Sato.
\newblock A calculus for interaction nets based on the linear chemical abstract
  machine.
\newblock {\em Electron. Notes Theor. Comput. Sci.}, 192(3):59--70, 2008.

\bibitem[PT09]{PT09}
Michele Pagani and Christine Tasson.
\newblock The inverse taylor expansion problem in linear logic.
\newblock In {\em Proceedings of the 24th Annual {IEEE} Symposium on Logic in
  Computer Science, {LICS} 2009, 11-14 August 2009, Los Angeles, CA, {USA}},
  pages 222--231, 2009.

\bibitem[PTV16]{PTV16}
Michele Pagani, Christine Tasson, and Lionel Vaux.
\newblock Strong normalizability as a finiteness structure via the taylor
  expansion of lambda-terms.
\newblock In {\em Foundations of Software Science and Computation Structures -
  19th International Conference, {FOSSACS} 2016, Held as Part of the European
  Joint Conferences on Theory and Practice of Software, {ETAPS} 2016,
  Eindhoven, The Netherlands, April 2-8, 2016, Proceedings}, pages 408--423,
  2016.

\bibitem[Reg92]{Regnier92}
Laurent Regnier.
\newblock {\em Lambda-calcul et réseaux}.
\newblock PhD thesis, Université Paris 7, Paris, France, December 1992.

\bibitem[Tas09]{Tasson09}
Christine Tasson.
\newblock {\em Sémantiques et syntaxes vectorielles de la logique linéaire}.
\newblock PhD thesis, Université Paris Diderot, Paris, France, December 2009.

\bibitem[TdF00]{Tortora00}
Lorenzo Tortora~de Falco.
\newblock {\em Réseaux, cohérence et expériences obsessionnelles}.
\newblock PhD thesis, Université Paris 7, Paris, France, 2000.

\bibitem[Vau17]{Vaux17}
Lionel Vaux.
\newblock Taylor expansion, $\beta$-reduction and normalization.
\newblock In {\em 26th {EACSL} Annual Conference on Computer Science Logic,
  {CSL} 2017, August 20-24, 2017, Stockholm, Sweden}, pages 39:1--39:16, 2017.

\end{thebibliography}
\end{document}